%% file: ex_article.tex
\documentclass[twocolumn,floatfix,amsmath,amssymb,aps, pre, superscriptaddress]{revtex4-2}

\pdfoutput=1
\usepackage{wrapfig}
\usepackage{extarrows} 
\usepackage{mathtools}
\usepackage{graphicx}
\usepackage[dvipsnames]{xcolor}
\usepackage{tikz}
\usepackage{amsthm}

\usepackage{relsize}
\usepackage{faktor}
\usepackage{blkarray}
\usepackage{booktabs}

\newtheorem{theorem}{Theorem}
\newtheorem{definition}{Definition}
\newtheorem{proposition}{Proposition}
\newtheorem{Lemma}{Lemma}
\newtheorem{corollary}{Corollary}
\newtheorem{example}{Example}

\DeclareMathOperator{\sinc}{sinc}
\DeclareMathOperator{\Ima}{Im}
\DeclareMathOperator{\rank}{rank}
\DeclareMathOperator{\diagm}{diag}

\newcommand{\ones}{\mathbbm{1}}
\newcommand{\inner}[2]{\left\langle #1,#2 \right\rangle}
\newcommand{\sset}[1]{\left\{ #1 \right\}} 
\newcommand{\abs}[1]{\left|#1\right|}
\newcommand{\norm}[1]{\left\lVert#1\right\rVert}
\newcommand{\R}{\mathbb{R}}
\newcommand{\defeq}{\stackrel{\mathsmaller{\mathsf{def}}}{=}}

\usepackage[colorlinks=true,citecolor=teal, urlcolor=Blue, linkcolor=Blue]{hyperref}
\usepackage[capitalise]{cleveref}

\usepackage{lipsum}
\usepackage{amsfonts,bbm}
\usepackage{graphicx}
\usepackage{epstopdf}
\usepackage{algorithmic}

\begin{document}

\title{A unified framework for Simplicial Kuramoto models}

\author{Marco Nurisso}\email{marco.nurisso@polito.it}
\affiliation{CENTAI Institute, Turin, Italy}
\affiliation{Dipartimento di Scienze Matematiche, Politecnico di Torino, Turin, 10129, Italy}
\affiliation{SmartData@PoliTO Center, Politecnico di Torino, Turin, 10129, Italy}

\author{Alexis Arnaudon}
\affiliation{
Blue Brain Project, École polytechnique fédérale de Lausanne (EPFL), Campus Biotech, 1202, Geneva, Switzerland}

\author{Maxime Lucas}
\affiliation{CENTAI Institute, Turin, Italy}

\author{Robert L. Peach}%
\affiliation{Department of Neurology, University Hospital Würzburg, Würzburg, Germany}
\affiliation{Department of Brain Sciences, Imperial College London, London, UK}

\author{Paul Expert}
\affiliation{UCL Global Business School for Health, UCL, London, UK}

\author{Francesco Vaccarino}
\affiliation{Dipartimento di Scienze Matematiche, Politecnico di Torino, Turin, 10129, Italy}
\affiliation{SmartData@PoliTO Center, Politecnico di Torino, Turin, 10129, Italy}

\author{Giovanni Petri}
\affiliation{CENTAI Institute, Turin, Italy}
\affiliation{IMT Lucca, Lucca, Italy}

\date{\today}

\begin{abstract}
    Simplicial Kuramoto models have emerged as a diverse and intriguing class of models describing oscillators on simplices rather than nodes. In this paper, we present a unified framework to describe different variants of these models, categorized into three main groups: ``simple'' models, ``Hodge-coupled'' models, and ``order-coupled'' (Dirac) models. 
    Our framework is based on topology, discrete differential geometry as well as gradient flows and frustrations, and permits a systematic analysis of their properties.
    We establish an equivalence between the simple simplicial Kuramoto model and the standard Kuramoto model on pairwise networks under the condition of manifoldness of the simplicial complex. 
    Then, starting from simple models, we describe the notion of simplicial synchronization and derive bounds on the coupling strength necessary or sufficient for achieving it. 
    For some variants, we generalize these results and provide new ones, such as the controllability of equilibrium solutions.
    Finally, we explore a potential application in the reconstruction of brain functional connectivity from structural connectomes and find that simple edge-based Kuramoto models perform competitively or even outperform complex extensions of node-based models.
\end{abstract}

\maketitle


\section{Introduction}
Synchronization is defined as the emergence of order from the interactions among many parts. 
It is a ubiquitous phenomenon that occurs in both natural and human-engineered systems~\cite{barahona2002synchronization,pikovsky2003synchronization, strogatz2004} and can be observed in a wide range of systems, including the firing of neurons~\cite{breakspear2010generative}, the twinkling of fireflies~\cite{strogatz1997spontaneous}, power grids~\cite{rohden2012self,nishikawa2015comparative} or audience applause~\cite{neda2000physics}. 
Despite the complexity and differences of these systems, the canonical Kuramoto model~\cite{Kuramoto} provides a unified framework for describing the onset of synchronization in systems of oscillators that interact in a pairwise fashion. 
While the original version of the model included interactions between all pairs of oscillators, later extensions of the model allowed the specification of arbitrary network topologies~\cite{arenas2008synchronization}.   
This, in turn, revealed interesting relationships between the dynamical properties of the model and the structure of the underlying network~\cite{Stability_kuramoto_model, watts1998collective}.

Traditional networks, however, provide a limited perspective on complex systems as they only consider pairwise interactions. 
To overcome this limitation, a new paradigm has recently emerged: networks with group (or higher order) interactions, i.e., interactions between any number of units~\cite{battiston2020networks,Battiston_2021,bick2021higher}. 
Group interactions have been recognized to play an important role in a rapidly growing list of systems, including brain networks~\cite{hoi_brain}, social~\cite{patania2017shape,benson2018simplicial,juul2022hypergraph} and biological communities~\cite{grilli2017higher,hoi_comm,sanchez2019high} among many others~\cite{battiston2020networks,Battiston_2021}. 
Group interactions can be represented by two main mathematical frameworks: hypergraphs or simplicial complexes. 
Although hypergraphs are more general, simplicial complexes have more structure because of the additional inclusion (or closure) condition: all sub-simplices of a simplex must be contained in a simplicial complex.
Consequently, simplicial complexes---like pairwise networks---possess a rich theory rooted in the mathematical field of discrete differential geometry and topology. 
Their expressive power is also greatly increased by the possibility of including weights~\cite{weights}, which naturally become embedded in their topological~\cite{petri2013topological,petri2014homological} and spectral structure~\cite{HORAK2013303}.  
The effect of a simplicial complex structure has been shown to induce new dynamical phenomena, such as explosive transitions~\cite{kuehn2021universal} and multistability~\cite{de2021multistability}, across a variety of dynamical processes, including random walks~\cite{SchaubSirev}, diffusion~\cite{schaub2018flow, carletti2020random,millan2021local}, consensus~\cite{neuhauser2021multibody,deville2021consensus,iacopini2022group}, spreading~\cite{lucas2023simplicially,iacopini2019simplicial,ferrazdearruda2021phase, chowdhary2021simplicial,st2021universal,st2021master,st2022influential}, percolation~\cite{sun2023dynamic,bianconi2019percolation, sun2020renormalization}, and evolutionary game theory~\cite{alvarez2021evolutionary}.

Naturally, this process of \emph{simplicialization} has also reached synchronization. 
One way to approach the modeling of synchronization in higher-order systems is to extend the family of possible interactions to include groups. 
From a network of interacting oscillators, we pass to a simplicial complex where node oscillators can also interact through triangles, tetrahedra, or higher order structures~\cite{skardal2019abrupt,lucas2020multiorder,gengel2020high,matheny2019exotic,zhang2023higher,gambuzza2021stability}. 
Another approach is to consider the simplicial Kuramoto~\cite{millan2020explosive,arnaudon2022connecting} as a model of synchronizing dynamics of higher-order topological signals. 
With it, we are not constrained to consider the evolution of oscillators placed on nodes, but we can place them on simplices of any order. 
This change, which at the beginning may appear arbitrary, allows us to consider higher-order interactions in a novel and powerful way: if an edge can connect only two nodes at a time, a triangle connects three edges, a tetrahedron four faces, etc\ldots 
More generally, simplicial oscillators of order $k$ will interact through $(k + 1)$-simplices, resulting in interactions of order $k + 2$. 
In line with the guiding principles of higher-order network theory, the essential difference between agents and \emph{carriers of interactions} fades away, leaving us with wider modeling freedom.
Its evolution equation, moreover, can be elegantly written by borrowing some of the concepts of discrete exterior calculus~\cite{grady2010discrete}, the discrete analogous to differential geometry on manifolds. This geometrical structure allows us to get precious insights into the dynamics of the model and how it is related to the topological properties of the simplicial complex.
This fruitful relation with topology has also recently put the simplicial Kuramoto at the center of the attention, resulting in different variants and extensions of the original model~\cite{calmon2022dirac,calmon2023local,Bianconi_2021_Dirac,arnaudon2022connecting}.

In this work, we aim to lay down the mathematical foundations for the study and derivation of Kuramoto models on simplicial complexes.
Our approach relies on consistent geometrical and dynamical structures such as discrete differential geometry and gradient flows to express the simplicial Kuramoto models in a strict mathematical form while allowing for several extensions able to couple the dynamics across Hodge subspaces or simplicial orders.

\subsection{Structure of the paper}

The work is structured as follows. 
We first state the Kuramoto model in \Cref{subsection:Kuramoto} and review the needed concepts of discrete differential geometry in \Cref{subsection:geometry}.

In \Cref{section:simplicial_kuramoto_model}, we introduce the standard simplicial Kuramoto model and interpret its interactions in terms of the geometry and topology of the underlying simplicial complex.  
With this approach, we find that the model is locally equivalent to the standard Kuramoto model when the complex is locally \emph{manifold-like}. 

Furthermore, in \Cref{section:equilibrium_analysis}, we define a natural notion of simplicial phase-locking, which we then relate to the projections of the dynamics on higher and lower dimensional simplices, allowing us to give a geometrical picture of its meaning. 
Taking inspiration from classic works on the node Kuramoto~\cite{Stability_kuramoto_model}, we discuss the phase-locked configurations and derive necessary and sufficient conditions on the coupling strength for their existence. 

Then, in \Cref{coupling}, we review and generalize some variants of the simplicial Kuramoto model that couple the dynamics across Hodge subspaces, such as the explosive model~\cite{millan2020explosive} or the simplicial Kuramoto-Sakaguchi~\cite{arnaudon2022connecting}.
Then, in \Cref{section:coupling_orders} we expand on the Dirac formulation of~\cite{calmon2022dirac} that couples oscillators across orders of interactions, and Hodge subspaces when coupled with the models of \Cref{coupling}.

Finally, in \cref{section:application}, we apply some of the models studied here to real-world brain data and show how simple, edge-based simplicial Kuramoto models can achieve better correlations with functional connectivity than the standard node-based Kuramoto model.

\section{Preliminaries}
\subsection{Kuramoto model}\label{subsection:Kuramoto}
We begin by briefly introducing the classical Kuramoto model. 
Let us consider a system of $n$ \emph{phase} oscillators, characterized solely by their phase $\theta_i$ and \emph{natural frequency} $\omega_i$, the frequency at which they oscillate when isolated from any interactions. 
The evolution of the uncoupled system can be described by a set of differential equations: $\Dot{\theta}_i = \omega_i$ for each oscillator $i$. 
To account for the interaction among oscillators, various approaches can be employed, depending on the underlying physics of the phenomenon under investigation. 
A particularly elegant and widely studied model, renowned for its simplicity, analytical tractability, and rich behavior, was introduced by Kuramoto \cite{Kuramoto}. 
Known as the Kuramoto model, it is described by the following system of first-order differential equations:

\begin{equation}\label{eq::kuramoto_model_2}
\Dot{\theta}_i = \omega_i - \sigma \sum_{j=1}^n \sin(\theta_i - \theta_j).
\end{equation}

In this formulation, an additional term captures the effect of interactions between oscillator $i$ and every other oscillator $j$, modulated by a positive \emph{coupling} or \emph{interaction strength} parameter, denoted as $\sigma$. 
By the properties of the sine function, we observe that the interaction force between oscillators $i$ and $j$ becomes zero when $\theta_i - \theta_j = k\pi$ i.e. when $\theta_i$ and $\theta_j$ represent the same or opposite angles modulo $2\pi$. 
Conversely, the interaction is strongest when the phase difference between the oscillators corresponds to odd multiples of $\frac{\pi}{2}$, implying orthogonal states on the unit circle. 
This simple interaction mechanism forms the basis of the Kuramoto model, which can be further generalized, enabling the study of diverse synchronization phenomena and their intricate dynamics.
\begin{figure*}[htpb]
    \centering
    \includegraphics[width=0.85\linewidth]{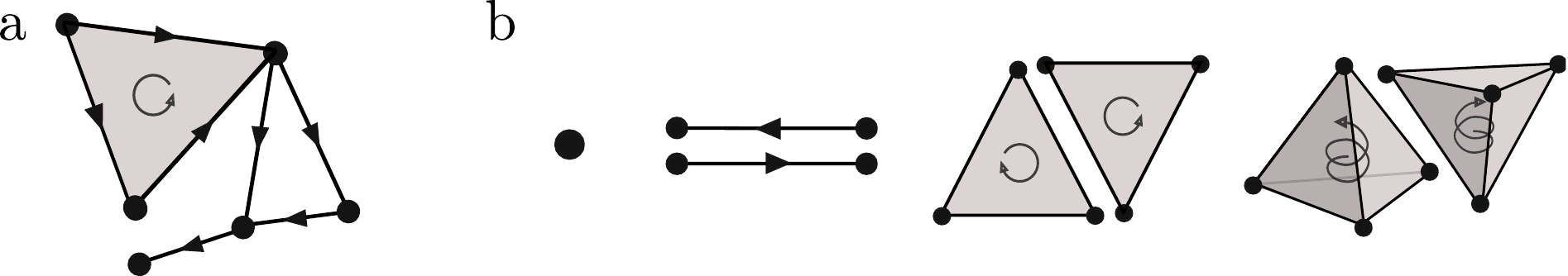}
    \caption{{\bf a.} Geometrical representation of a small oriented simplicial complex. {\bf b.} Oriented simplices of orders 0 (nodes), 1 (edges), 2 (triangles) and 3 (tetrahedra).}
    \label{fig:simplicial_complex}
\end{figure*}
\subsection{Discrete differential geometry}
\label{subsection:geometry}
A simplicial complex~\cite{hirani2003discrete,grady2010discrete} is a generalization of a graph that, along nodes and edges, can include triangles, tetrahedra, and their higher-dimensional analogs. Given a set of $N$ vertices $\mathcal{V} = \sset{v_0,\dots,v_{N-1}}$ we call $k$\emph{-simplex} any subset of $\mathcal{V}$ with $k+1$ elements. 
The dimension of a $k$-simplex $\sigma$ is $k$ and $\mathrm{dim}\,\sigma = k$. 
Geometrically, we think of $0$-simplices as nodes, $1$-simplices as edges, $2$-simplices as triangles, and so on. 
A \emph{simplicial complex} $\mathcal{X}$ is a set of simplices closed by inclusion, that is, every subset of a simplex is itself a simplex belonging to the complex (\cref{fig:simplicial_complex}a). 
We call $n_k$ the number of $k$-simplices in $\mathcal{X}$ and consider \emph{oriented} simplicial complexes, where each $k$-simplex is given an ordering of its vertices $[v_0,\dots,v_{k-1}]$, where two orderings are considered equivalent if they are related by an even number of swaps. This means that each simplex can have only two possible orientations (\cref{fig:simplicial_complex}b).   

Simplices in a simplicial complex can be related in two different ways. 
A \emph{subface} of a $k$-simplex $\sigma\in\mathcal{X}$ is any $(k-1)$-simplex $\tau$ contained in $\sigma$. The subfaces of a triangle $[a,b,c]$ are, for example, all its edges $[a,b],[b,c],[a,c]$. 
We write $\tau < \sigma$ when $\tau$ is a subface of $\sigma$.
A \emph{superface} of a $k$-simplex $\sigma\in\mathcal{X}$ is any $(k+1)$-simplex $\tau$ which contains $\sigma$. In this case, we write $\tau > \sigma$.
An oriented $k$-simplex $\sigma = [v_0,\dots,v_k]$ is said to be \emph{coherently oriented} with its subface $\tau$, with nodes $\sset{v_0,\dots,v_{i-1},v_{i+1},\dots,v_k}$, if the orientation given to $\tau$ is equivalent to
\begin{align}
[v_0,\dots,v_{i-1},v_{i+1},\dots,v_k]\, .
\end{align}
We write $\sigma \sim \tau$ when $\sigma > \tau$ and they are coherently oriented, while $\sigma\nsim\tau$ when $\sigma>\tau$ and they are incoherently oriented.
In addition, we say that two $k$-simplices $\sigma,\tau$ are \emph{lower adjacent} if they share a common subface (we write $\sigma \smile \tau$) while they are said to be \emph{upper-adjacent} if there exists a $(k+1)$-simplex which contains both of them (we write $\sigma\frown\tau$). 

For this work, it is important to highlight two special types of subfaces. We call \emph{free}, a subface which belongs only to a single simplex, and \emph{manifold-like} a subface $\tau$ which belongs to exactly two simplices $\sigma_1,\sigma_2$, one coherently oriented with $\tau$ ($\sigma_1\sim\tau$) and the other incoherently ($\sigma_2\nsim\tau$). This last definition comes from the fact that a simplicial complex where all the $(n-1)$-simplices are manifold-like (which we call \emph{simplicial} $n$-\emph{manifold}), when embedded into a Euclidean space, is an oriented topological $n$-manifold, in the sense that it locally looks like $\R^n$. If the simplicial complex is not a simplicial manifold, we can still have manifold-like subfaces which, when the complex is embedded, correspond to subspaces that are manifolds. In dimension $1$, for example, a manifold-like subface is a node incident to only two edges so that the complex is locally a line (\cref{fig:simplicial_manifold}a). 

From a geometric point of view, simplicial complexes take the role of geometrical domains upon which we define \emph{cochains}, algebraic objects which correspond to differential forms. A $k$-cochain is simply a function associating a real number to every $k$-simplex. The vector space of $k$-cochains is named $C^k(\mathcal{X})$ with a natural basis given by the functions associating $1$ to a particular simplex, and $0$ to all the others. Any $k$-cochain can therefore be written as
\begin{align} 
C^k(\mathcal{X}) \ni x = \sum_{i=1}^{n_k} x_i \sigma^i\, ,
\end{align}
with basis cochains $\sigma^i(\sigma_j) = \delta^i_j$ associated to every $k$-simplex $\sigma_i$.
Moreover, it is conventional to algebraically impose that a change of sign corresponds to a change of orientation
\begin{align*}
[v_0,v_1,\dots,v_k] = - [v_1,v_0,\dots,v_k]\, .
\end{align*}
If we assign positive weights to the $k$-simplices $w^k_1,\dots,w^k_{n_k}>0$, then we can endow the cochain space with an inner product given by the inverse of the diagonal matrix $W_{k}$
\begin{align}\label{eq:weight}  
W_{k} = \mathrm{diag}\left(w^k_1,\dots,w^k_{n_k}\right)\, .
\end{align}
We denote the inner product of cochains by $\inner{v}{w}_{w^k} \defeq v^TW^{-1}_{k} w$, and its induced norm by $\norm{v}_{w^k}\defeq \sqrt{v^TW_{k}^{-1}v}$, explicitly given as
\begin{align}
\norm{v}_{w^k} = \sqrt{\sum_{i=1}^{n_k} \frac{1}{w^k_i} v_i^2} \, .
\label{k-norm}
\end{align}

The inner product and the norm reduce to the standard Euclidean inner product and $2$-norm when the complex is \emph{unweighted}, i.e. $W_{k} = I_{n_k}$ for all $k=1,\dots, K$. 
In the rest of this work, we will always consider cochain spaces endowed with weights, meaning that inner products and norms will be weighted, and transposes will become adjoints. 
While this approach requires some care, it allows us to avoid carrying weight matrices along in every formula, resulting in more elegant and concise expressions that do not sacrifice generality.

The adjacency structure of the simplicial complex, and thus the complex itself, can be encoded in a family of linear operators acting on cochains. 
We define the $k$-th order \emph{incidence matrix} $B_k\in\R^{n_{k-1}\times n_k}$ describing the adjacency relations between $k$-simplices and $(k-1)$-simplices, as
\begin{align}\label{eq:incidence_matrix} 
B_k(i,j) =
\begin{cases}
    +1 &\text{ if } \mathrm{dim}\, \sigma_i = k-1,\, \sigma_j>\sigma_i \text{ and } \sigma_j \sim \sigma_i\,  , \\
    -1 &\text{ if } \mathrm{dim}\,\sigma_i = k-1,\, \sigma_j>\sigma_i \text{ and } \sigma_j \nsim \sigma_i\,  ,\\
    0\ &\text{ otherwise}\,  .
\end{cases}
\end{align}
We then define the \emph{coboundary operator}
\begin{align}\label{eq:coboundary}
    D^k=B_{k+1}^\top\,,
\end{align}
sending $k$-cochains to $(k+1)$-cochains. 
Its adjoint with respect to the inner product, which we name \emph{weighted boundary operator}, is
\begin{align}\label{eq:boundary}
    B^k = (D^{k-1})^* = W_{k-1} B_k W_{k}^{-1}\, .
\end{align}
Indeed, by definition of adjointness, for a $(k-1)$-cochain $x$ and a $k$-cochain $y$, $\langle D^{k-1} x, y\rangle_{w^k} = \langle D^{k-1} x, W_k^{-1}y\rangle_2 = \langle x, B_k W_k^{-1} y\rangle_2 = \langle x, W^{-1}_{k-1}W_{k-1}B_k W_k^{-1} y \rangle_2 = \inner{x}{B^k y}_{w^{k-1}}$.
The coboundary and boundary operators should be thought of as the discrete analog of the divergence and curl operators of differential calculus. They satisfy what is known as the ``fundamental theorem of topology''
\begin{align}
B^{k}B^{k+1} = 0,\, D^{k}D^{k-1} = 0\ \ \forall k\, ,
\end{align}
which is a linear-algebraic formalization of the topological fact that a boundary has no boundary. We call $k$-\emph{cocycle} a $k$-cochain $x$ such that
\begin{align}
D^kx = 0\, , 
\end{align}
and a \emph{weighted} $k$-\emph{cycle} a $k$-cochain $x$ such that
\begin{align}\label{cycle-def}
B^kx = 0\, .
\end{align}
With these two operators, we can define the \emph{discrete Hodge Laplacians} \cite{Lim_Laplacians}, which generalize the well-known graph Laplacian to act on higher order cochains
\begin{align}
    L^k = L^k_{\downarrow} + L^k_{\uparrow} = D^{k-1}B^k + B^{k+1}D^k\, . 
    \label{hodge_laplacian}
\end{align}
It can be easily proven that the kernel of the discrete $k$-Hodge Laplacian is isomorphic to the $k$-th (co)\emph{homology group} of the simplicial complex
\begin{align*}
\ker L^k = \ker B^k \cap \ker D^k \cong \mathcal{H}^k(\mathcal{X};\R) = \ker B^k\big/ \Ima B^{k+1}\, , 
\end{align*}
meaning that its dimension $\mathrm{dim} \ker L^k $ is equal to the $k$-th \emph{Betti number} of $\mathcal{X}$, i.e. the number of $k$-dimensional holes of the simplicial complex~\cite{ghrist2014elementary}. 
Intuitively, the $0$-dimensional holes are the connected components, $1$-dimensional holes are empty regions bounded by $1$-simplices, whereas the $2$-dimensional holes are cavities bounded by $2$-simplices.

\section{The simplicial Kuramoto model}\label{section:simplicial_kuramoto_model}

In this section, we formulate and study the Kuramoto model for interacting simplicial oscillators proposed in~\cite{millan2020explosive}. 
The rest of this section is organized as follows:
\begin{itemize}
    \item In \cref{subsection:simplicial_kuramoto}, we formulate the simplicial Kuramoto model using the tools of discrete differential geometry introduced in \cref{subsection:geometry}.
    \item In \cref{subsection:types}, we describe the local form of the two types of interactions in the model: from below and from above. In the case of interactions from below, we identify the presence of self-interactions resulting from free subfaces.
    \item In \cref{subsection:manifold-like}, we show that the $k$-th order simplicial Kuramoto model and the standard node Kuramoto model are equivalent when the simplicial complex is a simplicial $k$-manifold.
    \item In \cref{subsection:Hodge}, we describe how the dynamics naturally split into three independent subdynamics using the combinatorial Hodge decomposition theorem.
    \item In \cref{subsection:SOP}, we recall the definition of simplicial order parameter proposed in~\cite{arnaudon2022connecting}, discuss its implications on the meaning of synchronization in the simplicial model, and its differences with the standard Kuramoto order parameter.
\end{itemize}

\begin{figure*}[htpb]
    \centering
\includegraphics[width=0.6\linewidth]{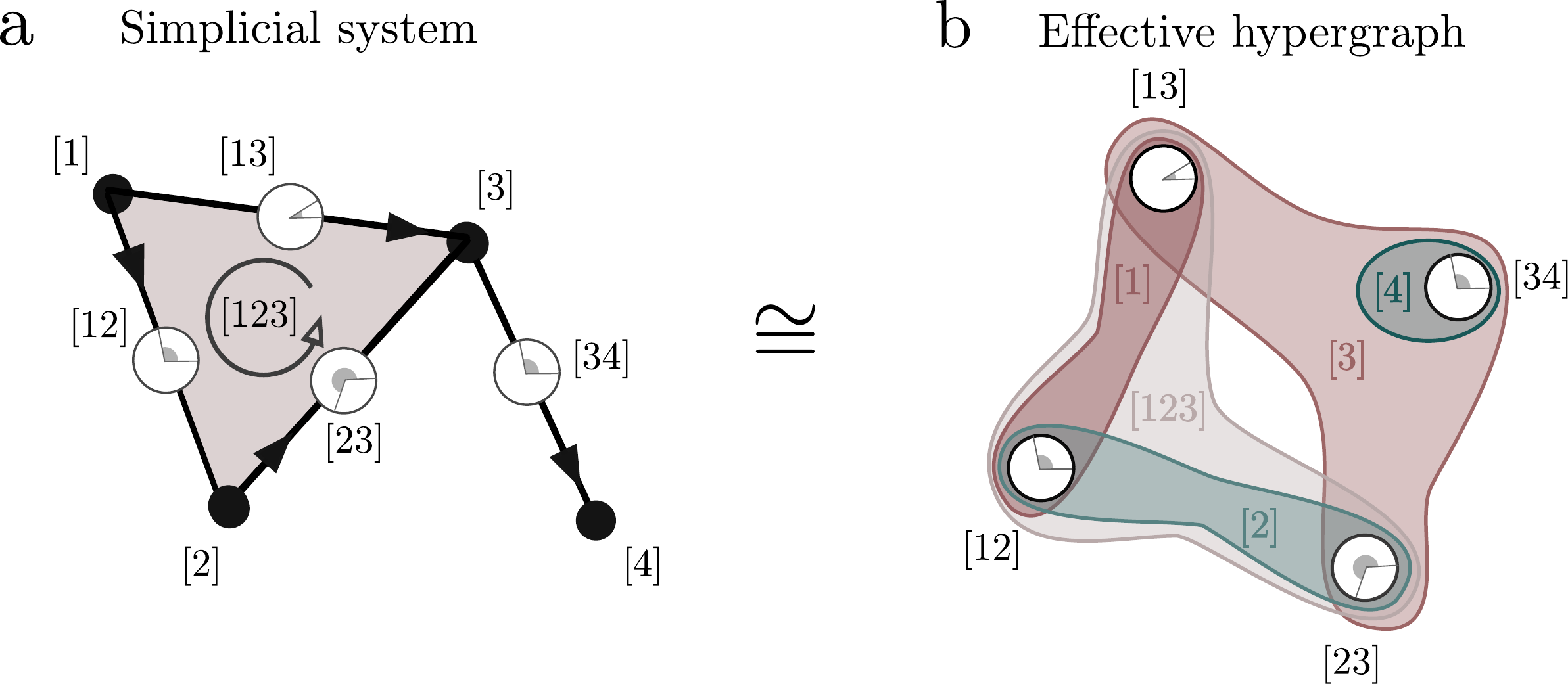}
    \caption{ {\bf a.} Edge simplicial Kuramoto on the simplicial complex described in \cref{example_1}. {\bf b.} The effective hypergraph of the dynamics describing the actual interactions taking place between the oscillators. The interaction hyperedges are labeled by the name of the simplex in the original complex which generates them. Note how the hyperedge $[4]$, representing the term $\sin(\theta_{[34]})$ in \cref{eq:example}, is interpreted here as a self-interaction.}
\label{fig:effective_hypergraph}
\end{figure*}

\subsection{Simplicial Kuramoto model}\label{subsection:simplicial_kuramoto}

Given a simplicial complex $\mathcal X$, the $k$-th order simplicial Kuramoto model~\cite{millan2020explosive,arnaudon2022connecting} describes a system where the $k$-simplices are oscillators interacting through common subfaces and superfaces. 
For example, one can consider oscillating edges that interact through common nodes and triangles (see \cref{fig:interactions}a). 
The model can be elegantly formulated with the boundary and coboundary operators as
\begin{align}\label{eq:simplicial_kuramoto}
    \Dot{\theta}_{(k)} = \omega - \sigma^\uparrow B^{k+1}\sin\left(D^k\theta_{(k)}\right) - \sigma^\downarrow D^{k-1}\sin\left(B^k\theta_{(k)}\right)\, .
\end{align}
Here, the phases of the oscillating $k$-simplices are gathered in the $n_k$-dimensional vector $\theta_{(k)}$, formalized as a $k$-cochain $\theta_{(k)}\in C^k(\mathcal{X})$, while $\omega\in C^k(\mathcal{X})$ represents the natural frequencies, i.e. $\omega_i$ is the frequency at which oscillator $i$ oscillates when no interactions are present. 
The parameters $\sigma^\uparrow,\sigma^\downarrow>0$ represent respectively the strength of the coupling through superfaces and subfaces. 
As shown in~\cite{consensus_simplicial_complexes}, the two interaction terms $B^{k+1}\sin(D^k\theta)$ and $D^{k-1}\sin(B^k\theta)$ describe, respectively, interactions from \emph{above} and \emph{below}, i.e. each oscillating $k$-simplex interacts with its adjacent simplices through both higher $(k+1)$ and lower dimensional $(k-1)$-simplices (\cref{fig:interactions}). In \Cref{subsection:types}, we unpack the matrix formulation and see the explicit form of these interaction terms. 
For ease of notation, from now on we will drop the subscript from $\theta_{(k)}$, as the order of oscillation can be easily inferred by the indices of the boundary and coboundary matrices in \cref{eq:simplicial_kuramoto}.

The form of~\cref{eq:simplicial_kuramoto} is not arbitrary but comes from the fact that, for $k=0$, it reduces to the standard Kuramoto model (from now on referred to as ``node Kuramoto'') on a network 
\begin{align}\label{eq:standard_kuramoto} 
\Dot{\theta}_i = \omega_i - \sigma\sum_{j}A_{ij}\sin(\theta_i-\theta_j)\, ,
\end{align}
where $A$ is the graph adjacency matrix. To see why, notice that \cref{eq:standard_kuramoto} can be rewritten in matrix form using the boundary and coboundary matrices (see Appendix \ref{section:kuramoto_with_boundary}) as
\begin{align}\label{eq:standard_kuramoto_boundary} 
\Dot{\theta} = \omega - \sigma B^1\sin(D^0\theta)\, .
\end{align}
One can think of $D^0$ as projecting the node phases on the edges by associating to each edge, which describes an interaction, the difference of its endpoints' phases. 
The boundary operator $B^1$ then projects the interactions back to the nodes, so that each node receives contributions from all edges that are incident to it. 
The extension of this term to higher-order oscillators is straightforward once one sees the model as a nonlinear extension of the graph Laplacian $L^0 \theta = B^1D^0\theta\ \rightarrow B^1\sin(D^0\theta)$, which can be naturally generalized with the discrete Hodge Laplacian defined in \cref{hodge_laplacian}.

Notice that in the case of the node Kuramoto, no simplices with order lower than the nodes exist, hence the dynamics results from interactions from \emph{above}, as the left term of \cref{eq:simplicial_kuramoto}. 
The interaction term from \emph{below} is naturally introduced to account for the lower adjacency structure present in simplicial complexes, but absent in graphs. 
Simply put, two triangles can be adjacent through a common edge and a common tetrahedron, but two nodes can only be adjacent through an edge, i.e. a higher-order simplex. This dynamics belongs to the wider class of dynamical systems on simplicial complexes, whose stability properties have been studied when the sine is replaced with a general nonlinearity (e.g.~\cite{consensus_simplicial_complexes, symmetries}).

\begin{example}\label{example_1}
Let us consider, as an example, the simplicial Kuramoto dynamics on the edges of the unweighted small simplicial complex shown in \Cref{fig:effective_hypergraph}a. 
The simplicial complex of interest is $\mathcal{X} = \sset{[1],[2],[3],[4],[12],[13],[23],[34],[123]}$ and the incidence matrices of order $1$ and $2$, accounting for the orientations, are
\begin{align*}
B_1 =~\begin{blockarray}{ccccc}
[12] & [13] & [23] & [34]\\
\begin{block}{(cccc)c}
  -1 & -1 & 0 & 0 &\ [1] \\
  1 & 0 & -1 & 0 &\ [2] \\
  0 & 1 & 1 & -1 &\ [3] \\
  0 & 0 & 0 & 1 &\ [4] \\
\end{block}
\end{blockarray}\, ,\ 
B_2 =~\begin{blockarray}{cc}
[123]\\
\begin{block}{(c)c}
  1  &\ [12] \\
  -1  &\ [13] \\
  1  &\ [23] \\
  0 &\ [34] \\
\end{block}
\end{blockarray}\, .
\end{align*}
Being $\mathcal{X}$ unweighted, we also have from \cref{eq:boundary} that $B^1 = B_1$, $D^0 = B_1^\top$, $B^2 = B_2$, $D^1 = B_2^\top$. If we consider the vector of phases on the edges $\theta = (\theta_{[12]},\theta_{[13]},\theta_{[23]},\theta_{[34]})$ and their natural frequencies $\omega =0$, then, after some algebra, we see that \cref{eq:simplicial_kuramoto} becomes
\begin{equation}\label{eq:example}
\begin{cases}
\begin{aligned}
\Dot{\theta}_{[12]} = &- \sigma^\downarrow\left(\sin(\theta_{[12]}+\theta_{[13]})+\sin(\theta_{[12]}-\theta_{[23]})\right) \\ &-\sigma^\uparrow\sin(\theta_{[12]}-\theta_{[13]}+\theta_{[23]})
\end{aligned}\\ 
\begin{aligned}
\Dot{\theta}_{[13]} = &- \sigma^\downarrow\left(\sin(\theta_{[13]}+\theta_{[12]})+\sin(\theta_{[13]}+\theta_{[23]}-\theta_{[34]})\right)\\ &+ \sigma^\uparrow\sin(\theta_{[12]}-\theta_{[13]}+\theta_{[23]})
\end{aligned}\\ 
\begin{aligned}
\Dot{\theta}_{[23]} = &- \sigma^\downarrow\left(\sin(\theta_{[23]}-\theta_{[12]})+\sin(\theta_{[23]}+\theta_{[13]}-\theta_{[34]})\right)\\ &- \sigma^\uparrow\sin(\theta_{[12]}-\theta_{[13]}+\theta_{[23]})
\end{aligned}\\ 
\begin{aligned}
\Dot{\theta}_{[34]} = - \sigma^\downarrow\left(\sin(\theta_{[34]}-\theta_{[13]}-\theta_{[23]})-\sin(\theta_{[34]})\right)
\end{aligned}
\end{cases}\, ,
\end{equation}
where the interaction terms from below and from above are clearly identifiable by their coupling strengths. 
Notice that some of the interactions from below are pairwise, others are higher-order (they involve three oscillators) and one, $\sin(\theta_{[34]})$, is of order $0$ as it depends on the value of a single oscillator. Moreover, the interaction from above through the triangle, as we expected, is higher-order and involves three oscillators.
\end{example}

While it is natural to define the dynamics on the edges and formulate the model using the incidence matrices of the simplicial complex, it is interesting to look at \Cref{eq:example} from another point of view. 
If we forget about the underlying simplicial complex and that $\theta_{[ij]}$ is a phase associated with an edge oscillator, what we are left with is a dynamical system where the phases of $4$ different oscillators evolve by interacting with each other in a way specified by the functional form of the equations. 
It is natural, therefore, to consider these oscillators as nodes and represent their interactions with \emph{hyperedges}, i.e. arbitrary groups of nodes. 
What we get, by neglecting the signs inherited by the orientations, is an \emph{effective hypergraph} (\cref{fig:effective_hypergraph}b) which does not resemble the original simplicial complex but has the advantage of clearly representing the actual interaction structure underlying the dynamics. 
Thus, the simplicial Kuramoto model can be seen as a particular kind of hypergraph oscillator dynamics where the coupling functions depend on the orientations of the original simplices.

Notice how the coupling functions in \cref{eq:example} are different from the ones classically used in hypergraph oscillator models \cite{skardal2019abrupt,lucas2020multiorder,gengel2020high,matheny2019exotic,zhang2023higher,gambuzza2021stability} as, in general, they do not vanish when the phases of the interacting oscillators are all equal.
Moreover, this hypergraph formulation, despite being more expressive, is harder to treat analytically as, with no knowledge of the underlying simplicial complex, one cannot resort to the powerful tools of topology and discrete geometry.

\begin{figure*}[htpb]
    \centering
\includegraphics[width=0.85\linewidth]{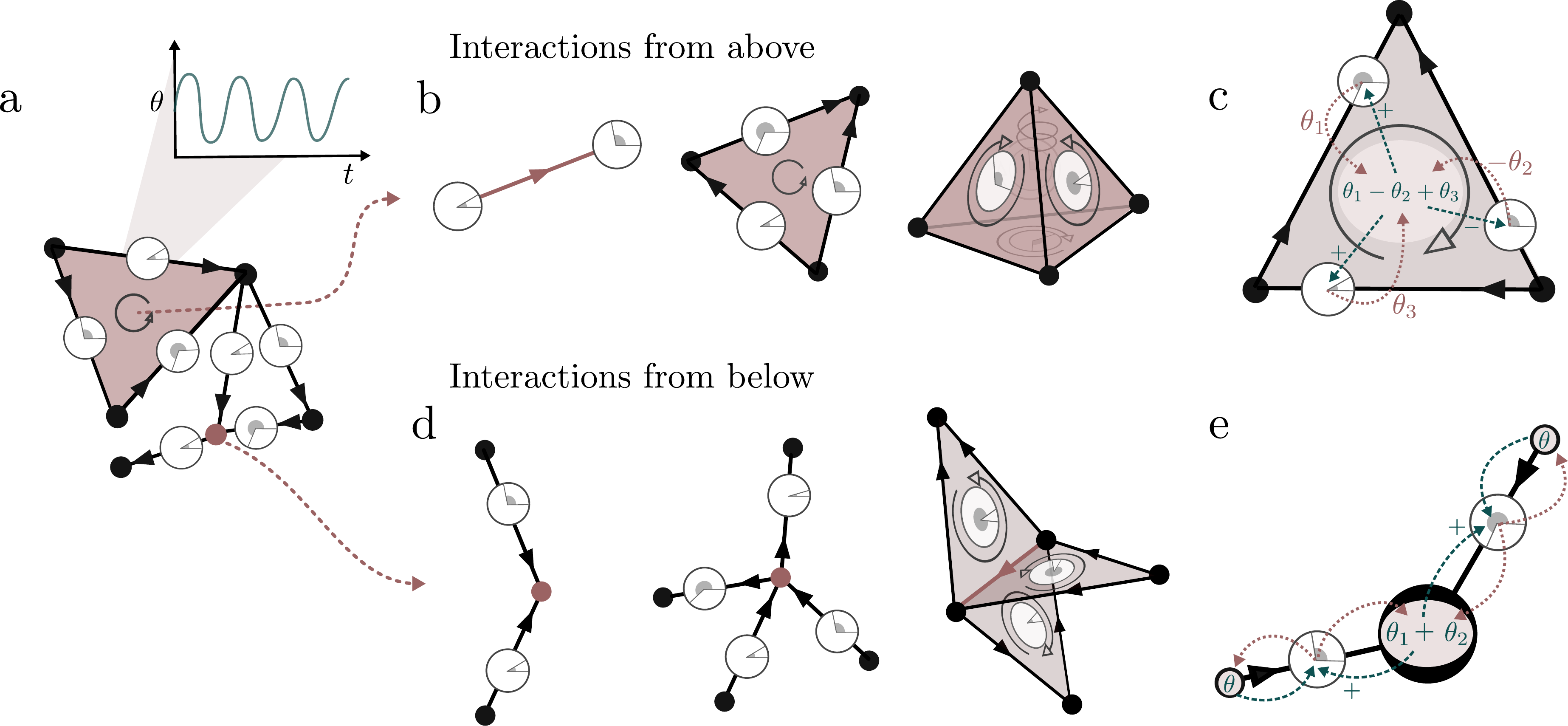}
    \caption{{\bf a.} The simplicial Kuramoto model allows us to consider oscillators, shown here as clocks, on the edges of a simplicial complex, interacting through nodes and triangles. {\bf b.} The interaction from \emph{above} \cref{eq:interaction_above} happens between $k+2$ oscillating $k$-simplices through a single $(k+1)$-simplex, here highlighted in red. {\bf c.} In the interaction from above, each oscillator involved is influenced by a term depending on the oriented sum of the phases. The phase $\theta_2$ appears with a minus sign because its edge is oriented in the opposite direction relative to the triangle.  {\bf d.} The interaction from \emph{below} \cref{eq:interaction_below} happens between an arbitrary number of oscillating $k$-simplices through a single $(k-1)$-simplex, highlighted here in red. {\bf e.} Unlike interactions from above, interactions from below through free subfaces are akin to self-interactions.}
    \label{fig:interactions}
\end{figure*}

\subsection{As above, \textit{not} so below: the two types of interactions}\label{subsection:types}
The introduction of \cref{eq:simplicial_kuramoto} was motivated by purely formal and symmetry arguments.
It is then important to study the \emph{local} form of the different interaction terms to understand what kind of system is being described. 
Following a similar procedure to the one proposed in ~\cite{consensus_simplicial_complexes}, we treat the two types of interactions separately.

Let us start with the interaction from \emph{above}
\begin{align}\label{eq:interaction_above}
I_{\uparrow}(\theta) \defeq -B^{k+1}\sin(D^k\theta)\, ,
\end{align}
which is a direct generalization of the standard node Kuramoto interaction term $-B^1\sin(D^0\theta)$. 
To understand its behavior, we look at the simplest possible interaction of its kind, where we have a single $(k+1)$-simplex regulating the interaction between its $k+2$ oscillating subfaces (\cref{fig:interactions}b). In this case, the incidence matrix is a column vector of the form
\begin{align}
B_{k+1} = \xi^{\uparrow}\in\sset{-1,1}^{k+2}\, ,
\end{align}
where $\xi^{\uparrow}_i$ is $1$ if the subface $i$ is coherently oriented with the $(k+1)$-simplex, and $-1$ if it is incoherently oriented.
It follows that $D^k = (B_{k+1})^\top = (\xi^{\uparrow})^\top$. 
\Cref{eq:interaction_above} then becomes
\begin{align}\label{eq:interaction_above_example}
    I_{\uparrow}(\theta) = -\xi^{\uparrow}\sin\left((\xi^{\uparrow})^\top\theta\right) \, , 
\end{align}
which means that each oscillator will be influenced by the same scalar value given by the oriented sum of the phases $(\xi^{\uparrow})^\top\theta$, with a sign depending on the coherence or incoherence of the orientations (see \cref{fig:interactions}c).
In the nodes case, this simply reduces to $I_\uparrow(\theta) = (-\sin(\theta_1-\theta_2),-\sin(\theta_2-\theta_1))^\top$. 
Given that a $(k+1)$-simplex always has $k+2$ subfaces, this kind of interaction involves $k+2$ oscillators and thus, for $k>0$, is genuinely higher order, in the sense that it does not result from the composition of multiple pairwise interaction terms. 

The interaction from \emph{below}
\begin{align}\label{eq:interaction_below}
I_{\downarrow}(\theta) \defeq -D^{k-1}\sin(B^k\theta)
\end{align}
describes the interactions of simplicial oscillators through lower-order simplices.
This interaction from below is absent in the node Kuramoto, and it represents the true novelty of the simplicial model. 
First, while only $k+2$ simplices of order $k$ can interact through a $(k+1)$-simplex, an arbitrary number of $k$-simplices can have a common subface and interact from below. 
This allows us to consider arbitrary higher-order interactions, not restricted by the order of the oscillating simplices.
It is then natural to ask if the interactions from below are locally of a similar form to the interactions from above, as in the case of \cref{eq:interaction_above}. 
For this, let us consider again the simplest possible interaction, i.e. the general case of $N$ $k$-simplices lower adjacent through a common subface with arbitrary orientations, as illustrated in \cref{fig:interactions}d. 
By considering an appropriate ordering of the simplices, this configuration is described by the incidence matrix
\begin{align}
B_k = 
\begin{pmatrix}
\xi^{\downarrow}_1 & \xi^{\downarrow}_2 & \cdots & \xi^{\downarrow}_N \\
o_1 & 0 & \cdots & 0\\
0 & o_2 & \cdots & 0\\
\vdots&\vdots & \vdots &\vdots\\
0 & 0 & \cdots & o_N
\end{pmatrix}\, , 
\end{align}
where the entries $\xi^{\downarrow}_i\in\sset{-1,1}$ describe the relative orientation between simplex $i$ and the common subface, while $o_i\in\sset{-1,1}^k$ contains the relative orientations between simplex $i$ and its other subfaces not involved in this interaction from below. 
Given that $D^{k-1}=(B_k)^\top$, \cref{eq:interaction_below} becomes
\begin{align}\label{eq:self_int}   
I_{\downarrow}(\theta) = -\xi^{\downarrow}\sin\left((\xi^{\downarrow})^\top\theta\right) - k\sin(\theta)\, ,
\end{align}
where $\xi^\downarrow = (\xi^\downarrow_1,\dots,\xi^\downarrow_N)^\top$.
Notice that the first term is formally the same as in \cref{eq:interaction_above_example} for the interaction from above. 
Each oscillator receives a contribution depending on the phases of all oscillators involved in the interaction. 
This means that the higher order interaction given by a $(k+1)$-simplex shares the same structure as the higher order interaction given by $k+2$ oscillators sharing a common subface. 

However, an extra term $-k\sin(\theta)$ appears, but by carrying out the computations which lead from \cref{eq:interaction_below} to \cref{eq:self_int}, it appears that this extra term is a sum of contributions coming from the subfaces not involved in the interaction, which, in this case, are free i.e. they are subfaces of only one simplex. 
In fact, here each oscillator has $k$ free subfaces, hence the multiplication factor $k$ in front of $\sin(\theta)$.
In general, due to this term, each oscillator modulates its own frequency based on its own phase, which is akin to a self-interaction through its free subfaces (\cref{fig:interactions}e). 
Formally, this term also appears in the Adler equation~\cite{pikovsky2003synchronization} describing the phase difference of a system with one oscillator driven by another one. 
From that point of view, the self-interaction terms can be seen as the driving of each oscillator by another non-existent oscillator that has a constant phase set to zero.

\subsection{Manifold-like simplicial complexes}\label{subsection:manifold-like}
Interestingly, if the interaction from below involves exactly two oscillating simplices, one coherent and the other incoherent with respect to the common subface so that the complex at that subface is \emph{manifold-like} (see \cref{subsection:geometry}), the interaction term will be the same as the standard node Kuramoto, i.e. of the form $\sin(\theta_1 - \theta_2)$.

\begin{figure*}[htp]
    \centering
    \includegraphics[width=\linewidth]{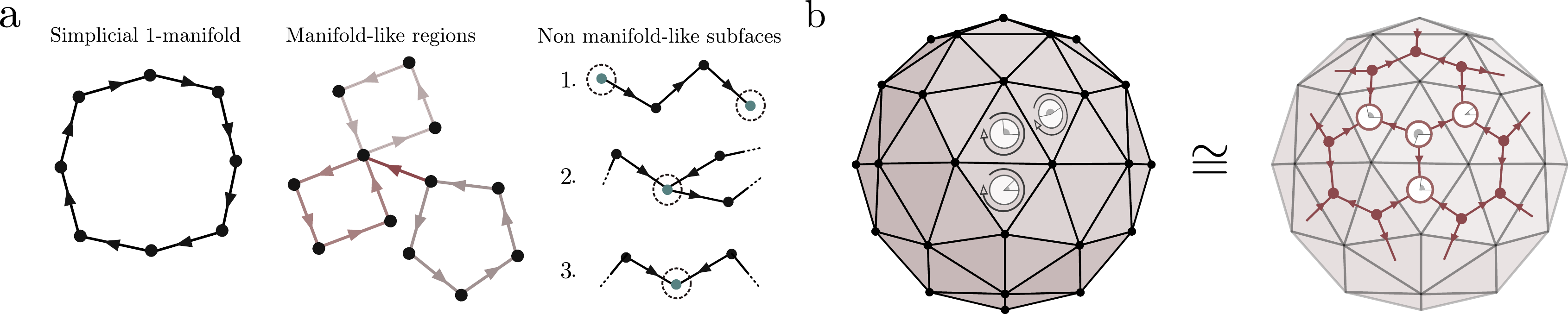}
    \caption{The form of the interactions from below of the simplicial Kuramoto is equivalent to ones of the standard node Kuramoto on manifold-like subfaces of the simplicial complex. {\bf a.} From left to right: a $1$-dimensional  simplicial manifold, where every node is manifold-like as it is incident to exactly two edges with different orientations. In the middle is a simplicial complex where the $1$-dimensional manifold-like regions are highlighted with different colors. On the right, the different ways in which a subface can produce an interaction different from a standard Kuramoto interaction: 1. the subface is free, 2. there are more than two oscillators incident to it, 3. there are two oscillators incident to it which are both coherently or incoherently oriented. {\bf b.} If the complex is an oriented simplicial manifold, then the interaction term from below is equivalent to a node Kuramoto taking place on the $1$-skeleton of the dual cell complex.}
    \label{fig:simplicial_manifold}
\end{figure*}

Different kinds of interactions occur at non-manifold-like subfaces (\cref{fig:simplicial_manifold}a), that is:
\begin{enumerate}
    \item at subfaces that are free, resulting in self-interactions;
    \item at subfaces that are adjacent to more than two simplices of order $k$, i.e. genuinely high-order interactions;
    \item at subfaces adjacent to two simplices that are both coherently or incoherently oriented, resulting in interactions of the form $\sin(\theta_1 + \theta_2)$.
\end{enumerate}

If every $(k-1)$-simplex which has at least a $k$-simplex incident to it is manifold-like, so that the simplicial complex is a simplicial manifold, we have the following equivalence result.
\begin{theorem}[Simplicial Kuramoto on a manifold]
\label{thm:simplicial-kuramoto-manifold}
Let $\mathcal{X}$ be a $k$-dimensional oriented simplicial manifold. Then it follows that the simplicial Kuramoto dynamics of order $k$ is equivalent to the standard node Kuramoto taking place on the $1$-skeleton of the dual cell complex to $\mathcal{X}$, that is, the graph with a node for each $k$-simplex and an edge for each $(k-1)$-simplex.
\end{theorem}
\begin{proof}
Under the assumptions that $\mathcal{X}$ is manifold-like, we can apply the discrete analogous to Poincar\'e duality (see~\cite{grady2010discrete} p.50) to obtain $B^k = \widetilde{D}^0,\quad D^{k-1} = \widetilde{B}^1$, where $\widetilde{B}$ and $\widetilde{D}$ are, respectively, the weighted boundary and coboundary operators of the dual cell complex to $\mathcal{X}$. 
The interaction term from below in the primal complex becomes the interaction term from above in the $1$-skeleton of the dual, i.e.
\begin{align*}
\Dot{\theta} = \omega - \sigma^\downarrow D^{k-1}\sin(B^k\theta) = \omega - \sigma^\downarrow \widetilde{B}^1\sin(\widetilde{D}^0\theta)\, ,
\end{align*}
which has the same form as the standard node Kuramoto model in \cref{eq:standard_kuramoto_boundary}.
\end{proof}
An illustration of this result can be seen in \cref{fig:simplicial_manifold}b with a triangulated sphere. 
Notice also that the dual graph to a simplicial $k$-manifold will necessarily be a $(k+1)$-regular graph, as every oscillating $k$-simplex has exactly $k+1$ subfaces.

\subsection{Hodge decomposition of the dynamics}\label{subsection:Hodge}
Thanks to the particular form of the two interaction terms, one can use a well-known result in combinatorial topology to decompose the dynamics into three independent subdynamics.
To show this, let us consider a simplicial complex $\mathcal{X}$, weighted or unweighted, which describes the interactions between $k$-th order oscillators. 
Then, the simplicial Hodge decomposition theorem~\cite{jiang2011statistical} states that every cochain can be decomposed into three orthogonal components 
\begin{align}\label{eq:hodge_decomposition}  
C^k(\mathcal{X}) \cong \R^{n_k} = \Ima B^{k+1}\oplus \ker L^k \oplus \Ima D^{k-1}\, ,
\end{align}
which can be interpreted as analogous to divergence-free, harmonic, and curl-free vector fields.
We use the theorem to decompose both the phases cochain $\theta$ and the natural frequencies $\omega$
\begin{align}
\theta = \theta_{\mathrm{df}} + \theta_{\mathrm{H}} + \theta_{\mathrm{cf}},\ \omega = \omega_{\mathrm{df}} + \omega_{\mathrm{H}} + \omega_{\mathrm{cf}}\, , 
\end{align}
where $\mathrm{cf}$ stands for \textit{curl-free}, $\mathrm{H}$ for \textit{harmonic}, and $\mathrm{df}$ for \textit{divergence-free}.
Rewriting the simplicial Kuramoto dynamics leveraging the orthogonality of the components, \cref{eq:simplicial_kuramoto} is equivalent to the following system
\begin{align}\label{eq:simplicial_kuramoto_decomposed}
\begin{cases}
\Dot{\theta}_{\mathrm{df}} = \omega_{\mathrm{df}} - \sigma^\uparrow B^{k+1}\sin(D^k\theta_{\mathrm{df}})\\
\Dot{\theta}_{\mathrm{H}} = \omega_{\mathrm{H}} \\ 
\Dot{\theta}_{\mathrm{cf}} = \omega_{\mathrm{cf}} - \sigma^\downarrow D^{k-1}\sin(B^k\theta_{\mathrm{cf}})\, .
\end{cases}
\end{align}
These three equations are of crucial importance. They tell us that under the simplicial Kuramoto dynamics: {\bf i)} the curl-free, the harmonic, and the divergence-free components evolve independently of one another, and {\bf ii)} the harmonic component is not affected by the interaction terms. 
Notice also that the interaction from above affects only the divergence-free component, while the one from below affects only the curl-free component.

Moreover, if $\omega_{\mathrm{H}} \neq 0$, there can be no equilibrium of the system as each component of $\theta_{\mathrm{H}}$ will always evolve with a fixed angular speed. 
It follows that it is always possible to pass to a frame of reference where the harmonic component is constant in time, simply by performing the change of variables $\theta \rightarrow \theta - \omega_{\mathrm{H}}$. In the case of the node Kuramoto, $\omega_{\mathrm{H}} = \Bar{\omega}\ones$, i.e. the constant vector of the average natural frequency. 
This is part of a more general observation that the addition of a harmonic cochain $x\in\ker L^k$ to the phases has no effect on the dynamics. In fact, it can be proven that $\ker L^k = \ker B^k \cap\ker D^k$ and thus both $B^kx$ and $D^kx$ are zero. Any change of variable $\gamma = \theta + x$ will thus leave \cref{eq:simplicial_kuramoto} formally unchanged. In this sense, we can say that the harmonic space is the \textit{gauge} of the simplicial Kuramoto. 

\subsection{Simplicial order parameters and gradient flow}\label{subsection:SOP}

To measure the degree of synchronization of a phase configuration, it is common to employ the \emph{order parameter}, which, for an unweighted network of $N$ oscillators, is defined as
\begin{align}\label{eq:standard_OP}
\widetilde{R}(\theta) = \frac{1}{N} \abs{\sum_{\alpha=1}^N e^{i\theta_\alpha}}\, .
\end{align}
By definition, it is non-negative, and it reaches its maximum value of $1$ when the oscillators are \emph{fully synchronized}, i.e. when they all have the same phase $\theta \propto \ones$. 

The order parameter defined in \cref{eq:standard_OP}, however, is independent of the network structure underlying the dynamics. As first proposed in~\cite{Stability_kuramoto_model}, we can generalize it in the following way
\begin{align}\label{eq:OP_2} 
R(\theta) = \frac{N^2 - 2e + 2\ones^\top\cos(B_1^\top\theta)}{N^2}\, ,
\end{align}
where $e$ is the number of edges. While \cref{eq:OP_2} reduces to \cref{eq:standard_OP} in the case of a fully connected network, it is, in fact, a more natural measure of synchronization in the general case, as
\begin{align}
\nabla_{\theta} R(\theta) \propto -B_1\sin(B_1^\top\theta) = I_\uparrow(\theta)\, ,
\end{align}
i.e. $R(\theta)$ is the potential function of the node Kuramoto dynamics, viewed as a gradient flow, i.e. $\dot \theta = \nabla_{\theta} R(\theta)$. 
As proposed in~\cite{arnaudon2022connecting}, we can extend this intuition to the  simplicial case and, neglecting constants that do not appear in the gradient, define the \emph{simplicial order parameter}
\begin{align}\label{eq:SOP}   
R_k(\theta) = \frac{1}{C_k}\left(\ones^\top W_{k-1}^{-1}\cos(B^k\theta) + \ones^\top W_{k+1}^{-1}\cos(D^k\theta)\right)\, ,
\end{align}
with the normalization constant 
\begin{align}
C_k = \ones^\top W^{-1}_{k-1}\ones + \ones^\top W^{-1}_{k+1}\ones\, .
\end{align}
The weight matrices are added to further generalize the construction to weighted simplicial complexes, and generate the weighted simplicial Kuramoto model as the gradient flow
\begin{align}
\label{eq:single-order-gradient-flow}
W_k\nabla_{\theta} R_k(\theta) \propto I_\uparrow(\theta) + I_\downarrow(\theta)\, .
\end{align}
This order parameter reaches a maximum value of $1$ if $\theta\in \ker B^k \cap \ker D^k$ i.e. when the phases cochain belongs to the harmonic space. 
This is a direct generalization of synchronization in the node Kuramoto model: in a connected network, $\ker L^0 = \mathrm{span}\sset{\ones}$, and the full synchronization condition is $\theta\propto \ones$, which is equivalent to the phase cochains being harmonic. 

Hence, under this definition, full synchronization in the simplicial model does not mean that the phases are all equal, but that $\theta$ is harmonic \cite{arnaudon2022connecting}.
Moreover, as the $k$-th harmonic space of a simplicial complex is isomorphic to the $k$-th homology group, we can think of fully synchronized configurations as, intuitively, localized around the $k$-dimensional holes.

\begin{definition}[Full synchronization]
A configuration $\theta$ is said to be \emph{fully synchronized} under the $k$-th order simplicial Kuramoto dynamics if $\theta\in\ker L^k$. 
\end{definition}

\begin{figure*}[htp]
    \centering
    \includegraphics[width=0.9\linewidth]{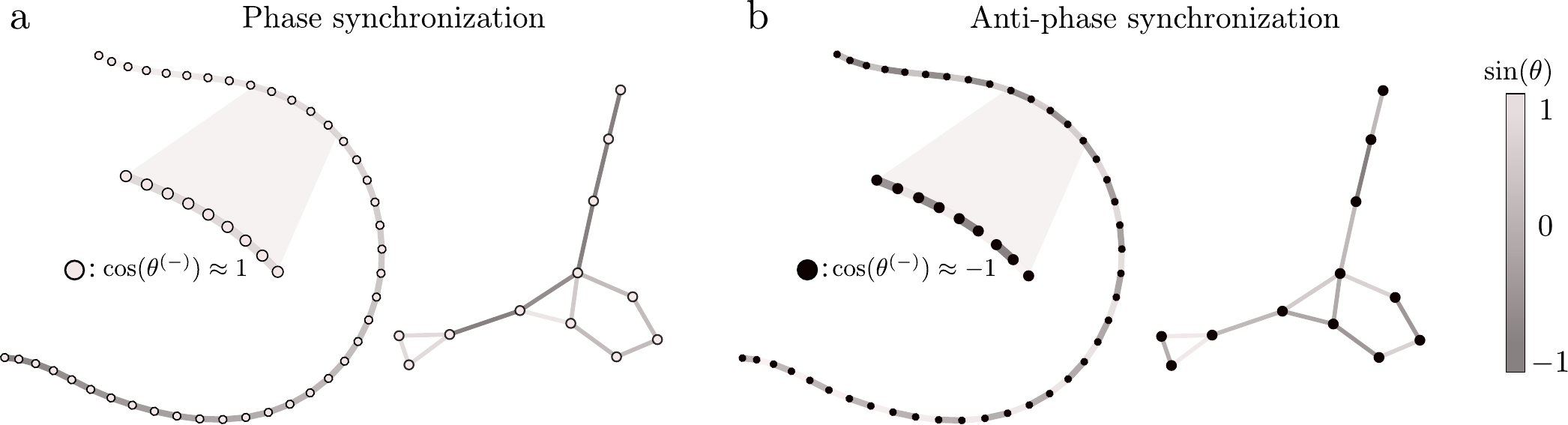}
    \caption{Configurations of phases $\theta\in C^1$, whose sine is shown here in color, which are phase ($R_1(\theta) \approx 1$) and anti-phase ($R_1(\theta) \approx -1$) synchronized in the case of a chain of edges, which is ``quasi''-manifold as all subfaces except the endpoints are manifold-like, and a more general $1$-dimensional simplicial complex.
    {\bf a.} In the case of phase synchronization, close oscillators on manifold-like regions have similar phases. {\bf b.} Anti-phase synchronized configurations, instead, are such that, on manifold-like regions, adjacent oscillators have opposite phases. 
    }
    \label{fig:phaseantiphase}
\end{figure*}

From the simplicial order parameter \cref{eq:SOP}, we can extract two \emph{partial order parameters}
\begin{subequations}
\begin{align}
R_k^-(\theta) &\defeq \frac{1}{C^{-}_k}\ones^\top W^{-1}_{k-1}\cos\left(B^k\theta\right) \label{def:partial_order_parameters_1}\\
\qquad R_k^+(\theta) &\defeq \frac{1}{C^{+}_k}\ones^\top W^{-1}_{k+1}\cos\left(D^k\theta\right)\, , 
 \label{def:partial_order_parameters_2}
\end{align}
\end{subequations}
where the normalization constants $C^{\pm}_k = \ones^\top W^{-1}_{k\pm 1}\ones$ ensure that they take values in $[-1,1]$. 
In this way, it holds that
\begin{align}
C_k R_k(\theta) = C_k^+ R^+_k(\theta) + C_k^- R^-_k(\theta) \, , 
\end{align}
and thus, aside from normalization, the order of a configuration is computed by measuring separately the local order induced respectively on $(k-1)$ and $(k+1)$-simplices. 

Notice that, by neglecting the constants in passing from \cref{eq:OP_2} to \cref{eq:SOP}, we have an order parameter that has values in the interval $[-1,1]$. This allows us to meaningfully distinguish two different types of synchronized configurations. We call a configuration of phases \emph{phase synchronized} when its order is close to $1$ and \emph{anti-phase synchronized} when it is close to $-1$. Phase synchronization generalizes to simplicial complexes the situation where close oscillators have similar phases (\cref{fig:phaseantiphase}a), while in anti-phase synchronization close oscillators have opposite phases forming ``checkerboard'' patterns, resembling an antiferromagnetic Ising model (\cref{fig:phaseantiphase}b).

Notice also how, in this work, with ``phase synchronization'' and ``full synchronization'' we refer to the static properties of a configuration of phases $\theta$, with no information on how it evolves under the dynamics. The notion of a configuration that ``stays synchronized'' under the dynamics will be tackled with the concept of phase-locking in \cref{subsection:phase-locking}.

\section{Equilibrium analysis}\label{section:equilibrium_analysis}

We now study the equilibrium properties of the simplicial Kuramoto model, extending to the simplicial cases concepts and results known in the node case.
\begin{itemize}
    \item In \cref{subsection:phase-locking}, we extend the notion of phase-locking to simplicial complexes, we look at its geometric meaning and see how it reduces to standard node synchronization on manifold-like regions of the complex.
    \item In \cref{subsection:equilibria}, we develop the necessary framework to discuss the equilibrium properties of the simplicial Kuramoto model, define reachable equilibria (Def.~\ref{def:reachable_equilibria}) and relate their existence to the presence of simplicial phase-locked configurations.
    \item In \cref{subsection:necessary}, we derive two bounds on the coupling strength providing necessary conditions for the existence of equilibria. We define the critical coupling (Def.~\ref{def:critical_coupling}) and characterize it as the solution to a linear optimization problem.
    \item In \cref{subsection:sufficient}, we prove a simple lower bound on the coupling strength which gives a sufficient condition for the existence of reachable equilibria. 
\end{itemize}

\subsection{Simplicial phase-locking}\label{subsection:phase-locking}

It directly follows from the Hodge decomposition of the simplicial Kuramoto model (see \cref{eq:simplicial_kuramoto_decomposed}) that studying its equilibrium properties is equivalent to separately studying the equilibria of the curl-free and divergence-free components. 
If these two converge to equilibrium, then the complete system will converge to a configuration evolving with constant harmonic angular speed, given by $\omega_{\mathrm{H}}$ ($\Dot{\theta} = \Dot{\theta}_{\mathrm{H}} = \omega_{\mathrm{H}}$).

\begin{definition}[Simplicial phase-locking]\label{def:phase-locking}
We say that the $k$-th order simplicial Kuramoto dynamics is \emph{phase-locked from above} if $\Dot{\theta}_{\mathrm{df}} = 0$ and \emph{phase-locked from below} when $\Dot{\theta}_{\mathrm{cf}} = 0$. 
\end{definition}
To motivate this definition, we consider the projections of the dynamics on lower and upper order simplices, defined as 
\begin{align}\label{eq:up_projection_hodge}   
\theta^{(+)} &\defeq D^k\theta = D^k\theta_{\mathrm{df}}\\
\theta^{(-)} &\defeq B^k\theta = B^k\theta_{\mathrm{cf}}\, . \label{eq:down_projection_hodge}   
\end{align}
We can think of $\theta^{(+)}$ and $\theta^{(-)}$ as the discrete versions of, respectively, the curl and divergence of the vector field $\theta$ and, we prove, they can equivalently capture simplicial phase-locking.
\begin{proposition}[Phase-locking equivalence]\label{prop:phase-locking-equivalence}
A configuration of phases $\theta$ is phase-locked from above (from below) if and only if its projection onto higher (lower) dimensional simplices is in equilibrium.
\begin{align}
\Dot{\theta}_{\mathrm{df}} = 0 \iff \Dot{\theta}^{(+)}=0,\quad \Dot{\theta}_{\mathrm{cf}} \iff \Dot{\theta}^{(-)}=0 = 0\, .
\end{align}
\end{proposition}
\begin{proof}
If $\Dot{\theta}_{\mathrm{df}}=0$, then
\begin{align*}
\Dot{\theta}^{(+)} = D^k\Dot{\theta}_{\mathrm{df}} = 0\, ,
\end{align*}
because of \cref{eq:up_projection_hodge}.
If instead $\Dot{\theta}^{(+)}=0$, then
\begin{align*}
 \Dot{\theta}_{\mathrm{df}} = (D^k)^\dagger D^k\Dot{\theta} = (D^k)^\dagger\Dot{\theta}^{(+)} = 0\, ,
\end{align*}
where $(D^k)^\dagger$ is the weighted Moore-Penrose pseudoinverse~\cite{weighted_pseudoinverse} and $(D^k)^\dagger D^k$ is the orthogonal projection operator onto $\Ima (D^k)^* = \Ima B^{k+1}$.
\end{proof}
This result allows us to include in \cref{def:phase-locking} the standard concept of phase-locking for the node Kuramoto. In fact, the node Kuramoto on a heterogeneous network is classically said to be phase-locked when the phase difference of connected oscillators stays constant in time. 
This means that $\Dot{\theta}^{(+)}_e = (D^0\Dot{\theta})_e = \Dot{\theta}_{i} - \Dot{\theta}_j = 0$ for every edge $e = (i,j)$ in the network. According to \cref{prop:phase-locking-equivalence}, the divergence-free component of the dynamics is in equilibrium and the system is, by \Cref{def:phase-locking}, \emph{simplicially} phase-locked from above. 
If the network is connected, moreover, $\Dot{\theta}$ must be harmonic, i.e. $\Dot{\theta}\propto \ones$, a situation which is usually named \emph{frequency-synchronized} as the frequencies of all oscillators coincide.

While phase-locking in the standard Kuramoto model means that all oscillators evolve with the same angular frequency, it is not clear how this extends to the simplicial case. In the case of phase-locking from below, we see that
\begin{align*}
\Dot{\theta}^{(-)} = 0 \iff \dfrac{d}{dt}(B^k\theta)=B^k\Dot{\theta} = 0\, ,
\end{align*}
or equivalently that $\Dot{\theta}(t)$, the cochain containing the angular frequencies, is a weighted cycle (see \cref{cycle-def}). 
This, in turn, means that for each $(k-1)$-simplex $\alpha$
\begin{align}\label{eq:flow_conservation}   
\sum_{i>\alpha} \xi_{\alpha,i}\Dot{\theta}_i = 0\, ,
\end{align}
where $\xi_{\alpha,i}\in\sset{-1,1}$ is the relative orientation of $k$-simplex $i$ with respect to its subface $\alpha$. 
Interestingly, \cref{eq:flow_conservation} corresponds to a flow conservation condition. Indeed, if we consider graphs with oscillating edges, at each node, the total phase flow of the incoming edges is the same as the total flow of the outgoing ones. 
In general, we can apply \cref{eq:flow_conservation} to understand phase-locking from below in some particular situations.
\begin{proposition}[Phase-locking on manifold-like regions]
If $\theta$ is phase-locked from below, i.e. $\Dot{\theta}^{(-)}=0$, 
\begin{enumerate}
    \item the connected manifold-like regions (see \cref{subsection:geometry}) of the complex (with respect to order $k$) evolve with the same angular frequency and are thus frequency-synchronized;
    \item oscillators with free subfaces (see \cref{subsection:geometry}) are frozen, i.e. $\Dot{\theta}_i=0$.
\end{enumerate}
\end{proposition}
\begin{proof}
$1.$ At a manifold-like subface we have only two incident simplices, one coherently oriented and one incoherent, which we name respectively $i$ and $j$. 
Condition \cref{eq:flow_conservation} gives us $\Dot{\theta}_i - \Dot{\theta}_j  = 0 \iff \Dot{\theta}_i = \Dot{\theta}_j$, so the incident oscillators are frequency-synchronized.
$2.$ If oscillator $i$ has a free subface then, by definition, that particular subface will be incident only to oscillator $i$. Phase-locking at that subface implies that $\Dot{\theta}_i = 0$, hence the thesis.
\end{proof}

\begin{figure*}[htp]
    \centering \includegraphics[width=\linewidth]{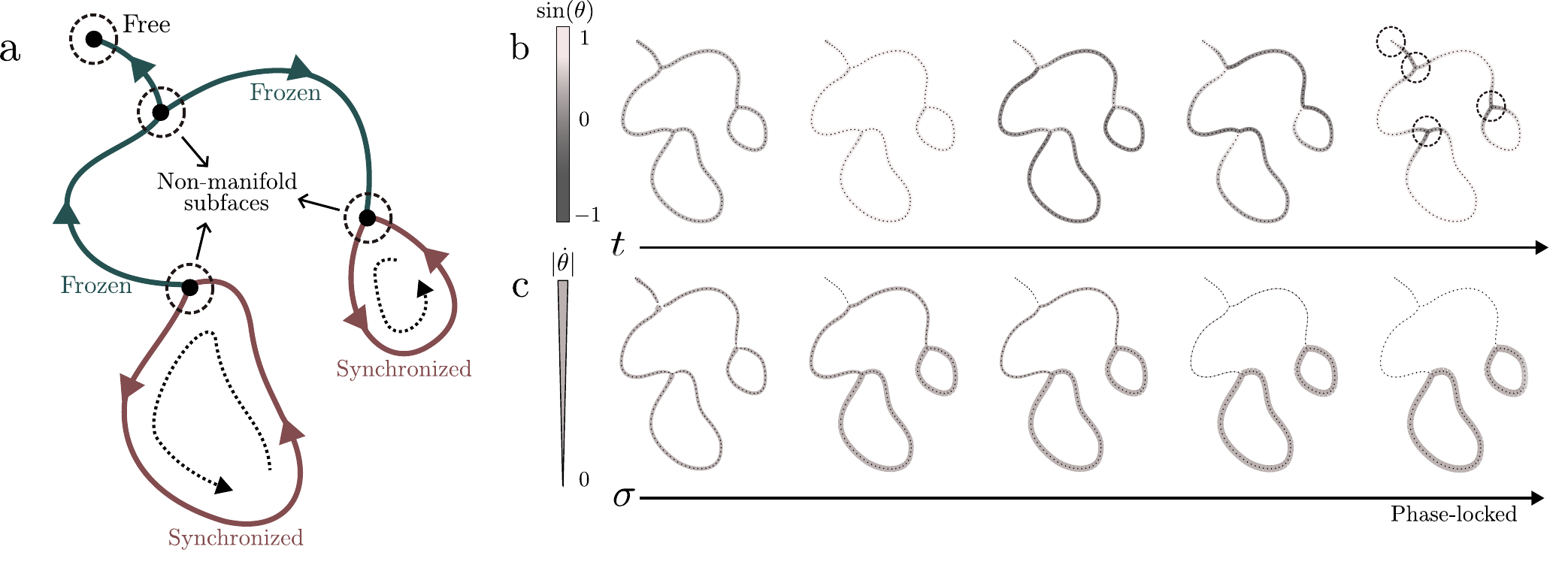}
    \caption{Simplicial Kuramoto dynamics on a simple graph with two holes (drawn here as a continuous space), where the edges are identical oscillators ($\omega = \ones$) with starting phase $\theta(0)=0\ones$. {\bf a.} The panel shows a diagram of the graph, highlighting the non-manifold points responsible for the non-triviality of the dynamics. The different branches of the graph are colored according to their frequency in a phase-locked (\cref{def:phase-locking}) state. In particular, the oscillation is frequency-synchronized on the two holes and exhibits traveling waves, while it is frozen ($\Dot{\theta}=0$) on the branch connecting them. {\bf b.} 
    A few snapshots of the dynamics on the graph are shown, with edges colored according to the sine of their phases. 
    In the last frame, the effect of the non-manifold points is evident.  
    {\bf c.} The dynamics is run for different values of the coupling strength, and the absolute value of the frequency ($|\Dot{\theta}|$) at the final integration time is shown with the edges' widths. 
    The last frame shows how the system reaches the same phase-locked configuration predicted with \cref{eq:flow_conservation} and depicted in panel {\bf a}. } 
    \label{fig:phase_lock}
\end{figure*}

A simple application of this result is shown in \cref{fig:phase_lock}a where the behavior of a phase-locked configuration can be inferred a priori by looking at the geometry of the graph. 
It can also be empirically seen that frequency-synchronized manifold regions exhibit phenomena akin to traveling waves localized around the $k$-holes of the complex when the homology is not trivial. 

\subsection{Existence of equilibria}\label{subsection:equilibria}

To study the equilibrium of the simplicial Kuramoto model, it is convenient to work with the $\theta^{(\pm)}$ defined in \cref{eq:up_projection_hodge,eq:down_projection_hodge} as \emph{projections} of the phases onto upper and lower simplices.
Their evolution equations~\cite{millan2020explosive} are readily obtained by multiplying \cref{eq:simplicial_kuramoto} by $D^k$ and $B^k$ to get
\begin{align}\label{eq:projection_dynamics}
\begin{split}
    \Dot{\theta}^{(+)} &= \omega^{(+)} - \sigma^\uparrow L^{k+1}_{\downarrow} \sin(\theta^{(+)})\\
    \Dot{\theta}^{(-)} &= \omega^{(-)} - \sigma^\downarrow L^{k-1}_{\uparrow} \sin(\theta^{(-)})\, , 
\end{split}
\end{align}
where $L^{k+1}_{\downarrow},L^{k-1}_{\uparrow}$ are the half Laplacian matrices from \cref{hodge_laplacian}, and we defined the projected natural frequencies as
\begin{align}
\omega^{(+)} \defeq D^k\omega,\qquad \omega^{(-)} \defeq B^k\omega\, . 
\end{align}

We then have two independent conditions for the existence of equilibrium solutions for each projection
\begin{align}\label{eq:equilibrium_conditions_1}
    \begin{cases}
    \Dot{\theta}^{(+)} = 0 \iff 
    L^{k+1}_{\downarrow}\sin(\theta^{(+)}) = \frac{\omega^{(+)}}{\sigma^\uparrow}\\
    \Dot{\theta}^{(-)} = 0\iff
    L^{k-1}_{\uparrow}\sin(\theta^{(-)}) = \frac{\omega^{(-)}}{\sigma^\downarrow}\, .
    \end{cases}
\end{align}
Since $\omega^{(+)}\in \Ima D^{k} = \Ima L^{k+1}_{\downarrow}$ and $\omega^{{(-)}}\in \Ima B^{k} = \Ima L^{k-1}_{\uparrow}$, the equilibrium equations can be solved using the pseudoinverse
\begin{align}\label{eq:equilibrium_conditions_2}
\begin{cases}
\sin(\theta^{(+)}) = (L^{k+1}_{\downarrow})^\dagger\frac{\omega^{(+)}}{\sigma^\uparrow} + x^{(+)}\\
\sin(\theta^{(-)}) = (L^{k+1}_{\uparrow})^\dagger \frac{\omega^{(-)}}{\sigma^\downarrow} + x^{(-)}
\end{cases}\, ,
\end{align}
for any weighted $(k+1)$-cycle $x^{(+)} \in \ker B^{k+1}$ and $(k-1)$-cocycle $x^{(-)} \in \ker D^{k-1}$. 
Moreover, by applying well-known properties of the Moore-Penrose pseudoinverse, we can simplify these expressions.
\begin{Lemma}
We have the following equalities
\begin{align}
(L^{k+1}_\downarrow)^\dagger\omega^{(+)} = (B^{k+1})^\dagger\omega,\ \ (L^{k-1}_\uparrow)^\dagger\omega^{(-)} = (D^{k-1})^\dagger\omega\, .
\end{align}
\end{Lemma}
\begin{proof}
We have
\begin{align*}
(L^{k+1}_{\downarrow})^\dagger\omega^{(+)} &= (D^k B^{k+1})^\dagger D^k\omega 
=(B^{k+1})^\dagger (D^k)^\dagger D^k\omega \\
&= (D^{k})^{*\dagger} (D^k)^\dagger D^k\omega 
= (B^{k+1})^\dagger\omega,
\end{align*}
and, analogously, $(L^{k-1}_{\uparrow})^\dagger\omega^{(-)} = (D^{k-1})^\dagger\omega$.
\end{proof}

\begin{definition}[Natural potentials]\label{def:characteristic_vectors}
We call \emph{natural potentials} of order $k$ the quantities
\begin{align}\label{eq:natural_potentials}   
\beta^{(+)} = (B^{k+1})^\dagger\omega\in\R^{n_{k+1}},\  \beta^{(-)} = (D^{k-1})^\dagger\omega\in\R^{n_{k-1}}\, .
\end{align}
\end{definition}
The name \emph{potential} comes from the fact that, for $k=1$, $\beta^{(-)}$ is an assignment of potentials to the nodes such that for each edge the difference of potential between its end-points (i.e. the voltage) is equal to $\omega$. 
Moreover, it can be easily proven that they correspond to the higher and lower order signals that appear in the Hodge components of the natural frequency vector $\omega$ as
\begin{align*}
\omega = B^{k+1}\beta^{(+)} + \omega_{\mathrm{H}} + D^{k-1}\beta^{(-)}\, .
\end{align*}
The values of the natural potentials are expressed in terms of the weighted Moore-Penrose pseudoinverse, which in \cref{eq:natural_potentials} is computed with respect to the inner products on the cochain spaces $W^{-1}_{k}$ $k=1,\dots,K$ \cref{eq:weight}. 
To compute the natural potentials of a weighted simplicial complex, the weights have to be included correctly in the pseudoinverse. The explicit formula, written in terms of the standard unweighted pseudoinverse is the following (\cite[Remark~2]{weighted_pseudoinverse})
\begin{align}
\beta^{(+)} &= W^{\frac{1}{2}}_{k+1}(W^{-\frac{1}{2}}_{k}B^{k+1}W^{\frac{1}{2}}_{k+1})^\dagger W^{-\frac{1}{2}}_{k}\omega \\ 
\beta^{(-)} &= W^{\frac{1}{2}}_{k-1}(W^{-\frac{1}{2}}_{k}D^{k-1}W^{\frac{1}{2}}_{k-1})^\dagger W^{-\frac{1}{2}}_{k}\omega \, .
\end{align}

Using the definition of natural potentials, we can rewrite the equilibrium conditions of \cref{eq:equilibrium_conditions_2} as
\begin{align}\label{eq:equilibrium_conditions_2.5}
\begin{cases}
\sin(\theta^{(+)}) = \frac{\beta^{(+)}}{\sigma^\uparrow} + x^{(+)}\\
\sin(\theta^{(-)}) = \frac{\beta^{(-)}}{\sigma^\downarrow} + x^{(-)}
\end{cases}\, ,
\end{align}
where we see that a \emph{necessary} condition for the existence of a solution is
for the right-hand sides to be bounded in $[-1,1]$ or, equivalently, for $x^{(+)}\in \ker B^{k+1}$ and $x^{(-)}\in \ker D^{k-1}$ to satisfy the following condition of \emph{admissibility}.
\begin{definition}[Admissible cycles]\label{def:admissible_vectors}
We call a (weighted) cycle $x^{(+)}\in \ker B^{k+1}$ \emph{admissible} if
\begin{align}
\norm{\frac{\beta^{(+)}}{\sigma^\uparrow} +x^{(+)}}_\infty \leq 1\, .
\end{align}
We call a cocycle $x^{(-)}\in \ker D^{k-1}$
\emph{admissible} if
\begin{align}
\norm{\frac{\beta^{(-)}}{\sigma^\downarrow} +x^{(-)}}_\infty \leq 1\, .
\end{align}
With a slight abuse of notation, we call them both \emph{admissible cycles}, and we name their sets $\mathcal{A}^{(+)}$ and $\mathcal{A}^{(-)}$.
\end{definition}
\begin{proposition}[Necessary condition from admissible cycles]
A \emph{necessary} condition for the existence of equilibrium solutions of the $(\pm)$ dynamics is that $\mathcal{A}^{(\pm)}\neq \emptyset$.
\end{proposition}
Intuitively, each cochain $\beta^{(\pm)}/\sigma^\updownarrow$ should be \emph{close} to, respectively, the vector space of weighted $(k+1)$-cycles (for $(+)$) and $(k-1)$-cocycles (for $(-)$).

When both $x^{(+)}$ and $x^{(-)}$ are admissible, we can invert the sine function in \cref{eq:equilibrium_conditions_2} and get an explicit expression for the set of equilibrium configurations of the projections.
\begin{definition}[Equilibrium sets]\label{def:eq_sets}
We define the equilibrium sets of the projections as
\begin{gather}
\begin{split}
\mathcal{E}^{(+)} = \bigg\{&(-1)^{s_+}\odot\arcsin\left(\frac{\beta^{(+)}}{\sigma^\uparrow} +x^{(+)}\right) + s_+\pi:\\ &s_+\in\sset{0,1}^{n_{k+1}},\ x^{(+)}\in\mathcal{A}^{(+)}\bigg\}
\end{split}\, ,
\\
\begin{split}
\mathcal{E}^{(-)} = \bigg\{&(-1)^{s_-}\odot\arcsin\left(\frac{\beta^{(-)}}{\sigma^\downarrow} +x^{(-)}\right) + s_-\pi:\\ &s_-\in\sset{0,1}^{n_{k-1}},\ x^{(-)}\in\mathcal{A}^{(-)}\bigg\}\, ,
\end{split}
\end{gather}
where $\odot$ is the component-wise Hadamard product. 
\end{definition}
Any $\theta^{(\pm)}\in\mathcal{E}^{(\pm)}$ will thus be a fixed point of the projected dynamics from \cref{eq:projection_dynamics}. 

From a geometrical point of view, the equilibrium set $\mathcal{E}^{(\pm)}$ is a subset of $\R^{n_{k\pm 1}}$ and, for each given $s_\pm$, is a manifold of dimension given by $\mathrm{dim} \ker D^{k-1}$ for $(-)$ and $\mathrm{dim} \ker B^{k+1}$ for $(+)$. For example, for the projection on the nodes of edge dynamics, we have that $\mathrm{dim}\ker D^0$ is the number of connected components. If the simplicial complex is connected then $\mathcal{E}^{(-)}$ is a collection of curves in an $n_0$-dimensional space (see \cref{fig:bounds_compare}b). 
\begin{proposition}\label{prop:equilibrium_pm_dynamics}
$\Dot{\theta}^{(\pm)} = 0$ if $\theta^{(\pm)}\in\mathcal{E}^{(\pm)}$ modulo $2\pi$\, .
\end{proposition}

This condition, however, is only necessary and does not fully characterize the equilibrium configurations of the projected dynamics. To see why, notice that the dynamics for the $(-)$ component (the same holds for $(+)$) in \cref{eq:projection_dynamics} states that the time derivative of $\theta^{(-)}$ will be the vector $\omega^{(-)}-\sigma^\downarrow L^{k-1}_{\uparrow}\sin(\theta^{(-)})$, which always belongs to $\Ima B^k$. This, together with the initial configuration $\theta^{(-)}_0 = B^{k}\theta_0\in \Ima B^k$, means that the trajectories of the dynamics live in the subspace $\Ima B^k$. Only the equilibria of $\mathcal{E}^{(-)}$, also in $\Ima B^k$, are \emph{reachable} by the dynamics. 
\begin{definition}[Reachable equilibria]\label{def:reachable_equilibria}
We define the sets of reachable equilibria as
\begin{align}
  \mathcal{R}^{(-)} & = \mathcal{E}^{(-)}\cap  \Ima B^k \\
  \mathcal{R}^{(+)} & = \mathcal{E}^{(-)}\cap  \Ima D^k\, .
\end{align}
\end{definition}
We then have our final result.
\begin{proposition}[Equivalence phase-locking reachability]\label{prop:curldiv-projected_equivalence}
The curl-free (divergence-free) component admits equilibria if and only if $\theta^{(-)}$ (resp. $\theta^{(+)}$) admits reachable equilibria. 
\end{proposition}

The framework we developed in this section can be fruitfully exploited to independently discuss the existence of equilibrium configurations of the divergence-free and of the curl-free components of the simplicial Kuramoto dynamics (\cref{eq:simplicial_kuramoto_decomposed}). 
In particular, \cref{prop:curldiv-projected_equivalence} tells us that such configurations will exist if and only if there are reachable equilibria for the projections. 
These, in turn, are subsets of the larger sets of fixed points (\cref{def:eq_sets}), whose explicit expression is known and whose non-emptiness can thus be controlled more easily, giving us necessary conditions for equilibrium. 

\subsection{Necessary conditions for phase-locking}\label{subsection:necessary}

In this section, we investigate the relation between the equilibrium properties of the simplicial Kuramoto model and the value of the coupling strength. It is natural to think that having a stronger interaction would make it easier for the system to reach a synchronized configuration as the intrinsic differences among the oscillators, encoded by their natural frequencies, become secondary. This intuition is extensively confirmed by numerous results proved about the node Kuramoto (see \cite{critical_coupling,Stability_kuramoto_model,synch_survey}), some of which we extend to the simplicial case. We will thus derive bounds on the coupling strength, which gives us necessary and sufficient conditions for the existence of reachable equilibria, i.e. for phase-locking from below and from above. 
Note that all the results below refer to the \textit{existence} of phase-locked configurations and provide no information about whether the dynamics will actually converge to them.

Let us consider a simplicial complex whose $m$-simplices have weights $w^m_1,\dots,w^m_{n_{m}}$ for any order $m$, and focus on the $k$-th order simplicial Kuramoto dynamics.
The easiest conditions to derive are those that ensure that there are no admissible cycles $\mathcal{A}^{(\pm)}=\emptyset$. If it holds then the equilibrium sets are empty $\mathcal{E}^{(\pm)}=\emptyset$ (\cref{def:eq_sets}) and, by inclusion, the reachable sets are as well $\mathcal{R}^{(\pm)}=\emptyset$ i.e. there are no reachable equilibria/phase-locked configurations.
\begin{proposition}[Sufficient condition for no phase-locking]\label{prop:necessary_condition_2_ball}
If 
\begin{align}
\sigma^\updownarrow < \sigma^{(\pm)}_s \defeq \frac{1}{\sqrt{n_{k\pm 1}^w}}\norm{\beta^{(\pm)}}_{w^{k\pm 1}}\, ,
\end{align}
where
\begin{align}
n_{k\pm 1}^w \defeq \sum_{i=1}^{n_{k\pm 1}} \frac{1}{w^{k\pm 1}_i}\, ,
\end{align}
then $\mathcal{E}^{(\pm)}=\emptyset$ and the $(\pm)$ projection admits no equilibria.
\end{proposition}
\begin{proof}
First, see that we can bound the weighted $w^{k\pm 1}$ norm [\cref{k-norm}] with the $\infty$-norm:
\begin{align*}
\norm{v}_{w^{k\pm 1}} &= \sqrt{\sum_{i=1}^{n_{k\pm 1}} \frac{1}{w^{k\pm 1}_i} v_i^2} \leq \sqrt{n_{k\pm 1}^w \left(\max_i v_i^2\right)}\\ 
&\leq \sqrt{n_{k\pm 1}^w} \norm{v}_\infty.
\end{align*}
With this in mind, we can write
\begin{align*}
\norm{\frac{\beta^{(\pm)}}{\sigma^\updownarrow} + x^{(\pm)}}_\infty \geq \frac{1}{\sqrt{n_{k\pm 1}^w}}\norm{\frac{\beta^{(\pm)}}{\sigma^\updownarrow} + x^{(\pm)}}_{w^{k\pm 1}}.
\end{align*}
The two addenda in the norm are orthogonal with respect to the inner product $W^{-1}_{k\pm 1}$ because, in the $(-)$ case, $
x^{(-)} \in \ker D^{k-1}$ and 
\begin{align*}
\beta^{(-)} &\in \Ima (D^{k-1})^\dagger = \Ima (D^{k-1})^* = (\ker D^{k-1})^\perp\, ,
\end{align*}
thus
\begin{align*}
&\frac{1}{\sqrt{n_{k\pm 1}^w}}\norm{\frac{\beta^{(\pm)}}{\sigma^\updownarrow} + x^{(\pm)}}_{w^{k\pm 1}} \\
&= \frac{1}{\sqrt{n_{k\pm 1}^w}}\sqrt{\norm{\frac{\beta^{(\pm)}}{\sigma^\updownarrow}}^2_{w^{k\pm 1}}+\norm{x^{(\pm)}}^2_{w^{k\pm 1}}} \\ 
&\geq \frac{1}{\sqrt{n_{k\pm 1}^w}}\norm{\frac{\beta^{(\pm)}}{\sigma^\updownarrow}}_{w^{k\pm 1}}\, .
\end{align*}
If this last term is strictly greater than $1$ then there will be no admissible cycles and, therefore, no equilibria.
\end{proof}

The condition in \cref{prop:necessary_condition_2_ball} is easy to check and provides a way to tune the coupling constants to make the set of admissible cycles empty, and thus phase-locking (from above or from below) impossible. 
It is now natural to ask what is the minimum value of $\sigma$ such that there are admissible cycles, to get a sharper necessary condition for the existence of phase-locked configurations.
\begin{definition}[Critical coupling]\label{def:critical_coupling}
We call \emph{critical coupling} $\sigma_*^{(\pm)}$ for the $(\pm)$ projection the minimum value of $\sigma$ such that there are admissible cycles ($\mathcal{A}^{(\pm)}\neq\emptyset)$.
\end{definition}
It follows directly from the definition that $\sigma^{(\pm)}_s < \sigma^{(\pm)}_*$. To find its value, notice first that there can be admissible cycles $x^{(\pm)}$ (Def.~\ref{def:admissible_vectors}) if and only if
\begin{align}\label{eq:existence_of_admissible_cycles}    
\min_{x\in \ker D^{k-1}}\norm{\frac{\beta^{(\pm)}}{\sigma^\updownarrow} + x}_\infty \leq 1\, .
\end{align}
By manipulating this expression, we can get the exact value of the critical coupling as a solution to a linear optimization problem. 

\begin{theorem}[Value of the critical coupling]\label{theorem:critical_coupling}
The critical coupling $\sigma^{(\pm)}_*$ can be found in the solution of a linear optimization problem
\begin{align}\label{eq:critical_sigma_problem}
\sigma^{(+)}_* &= \min_{x\in \ker B^{k+1}}\norm{\beta^{(+)}+ x}_\infty\\
\sigma^{(-)}_* &= \min_{x\in \ker D^{k-1}}\norm{\beta^{(-)}+ x}_\infty\, ,
\end{align}
which corresponds to the $\infty$-distance of $\beta^{(\pm)}$ from the space of weighted $(k+1)$-cycles (resp. $(k-1)$-cocycles).
\end{theorem}
\begin{proof}
Using \cref{eq:existence_of_admissible_cycles}, we first show that the critical couplings $\sigma_*^{(-)}, \sigma_*^{(+)}$ satisfy, respectively.
\begin{align}
\min_{x\in \ker D^{k-1}}&\norm{ \frac{\beta^{(-)}}{\sigma_*^{(-)}} + x}_\infty = 1,\ \\
\min_{x\in \ker B^{k+1}}&\norm{ \frac{\beta^{(+)}}{\sigma_*^{(+)}} + x}_\infty = 1\, . 
\label{cond1}
\end{align}
If the statement were false and
\begin{align*}
\min_{x\in \ker D^{k-1}}\norm{ \frac{\beta^{(-)}}{\sigma_*^{(-)}} + x}_\infty = a\, ,
\end{align*}
with $0 < a<1$, then we could divide both sides by  $a$ and get
\begin{align*}
\min_{x\in \ker D^{k-1}}\norm{ \frac{\beta^{(-)}}{a\sigma_*^{(-)}} + \frac{1}{a}x}_\infty = 1\, ,
\end{align*}
which means that for $\sigma = a\sigma_*^{(-)}<\sigma_*^{(-)}$ there is an admissible cycle $\frac{x}{a}$, which is impossible because we assumed that $\sigma_*^{(-)}$ is the smallest coupling with that property.

Then, multiplying both terms of \cref{cond1} by $\sigma_*^{(-)}$, we have
\begin{align*}
\min_{x\in \ker D^{k-1}}\norm{\beta^{(-)} + \sigma_*^{(-)}x}_\infty = \sigma_*^{(-)}\, .
\end{align*}
It is now possible to perform a linear change of variable in the optimization problem $\sigma_*^{(-)}x \rightarrow \widetilde{x}$ which will change the optimal solution position but not the optimum itself. This means that $\sigma_*^{(-)}$ disappears from the left-hand side and it is found as the solution to the optimization problem above.
\end{proof}

In some special cases, the critical coupling admits a closed formula. For example, for the $(-)$ projection of the simplicial Kuramoto dynamics on the edges of a connected simplicial complex, the set of admissible vectors and the critical coupling can both be found explicitly.
\begin{theorem}[Critical coupling in the edge simplicial Kuramoto]\label{theorem:critical_coupling_edge}
For the $(-)$ component of the edge dynamics on a connected simplicial complex, it holds that
\begin{align}
x^{(-)}\in\mathcal{A}^{(-)} \iff x^{(-)} = x\ones\, , 
\end{align}
where 
\begin{align}
 -\min\left(\frac{\beta^{(-)}}{\sigma^\downarrow}\right)-1\leq x\leq -\max\left(\frac{\beta^{(-)}}{\sigma^\downarrow}\right)+1\, , 
 \label{cond}
 \end{align}
and
\begin{align}
\sigma^{(-)}_* = \frac{\max\left (\beta^{(-)}\right )-\min\left (\beta^{(-)}\right )}{2}\, .
\end{align}
\end{theorem}
\begin{proof}
If the complex is connected we have that $D^0$ has a 1-dimensional kernel given by $\mathrm{span}\sset{\ones}$. This means that there are admissible vectors if and only if
\begin{align*}
\norm{\frac{\beta^{(-)}}{\sigma^\downarrow} + x\ones}_\infty \leq 1 \iff -1\leq \frac{\beta^{(-)}_i}{\sigma^\downarrow} + x \leq 1\, ,
\end{align*}
$\forall i=1,\dots,n_0$, which holds if and only if \cref{cond} holds, and has solutions only when
\begin{align*}
&\max\left(-\frac{\beta^{(-)}}{\sigma^\downarrow}\right)-1 \leq \min\left(-\frac{\beta^{(-)}}{\sigma^\downarrow}\right)+1 \\
&\iff \sigma^\downarrow \geq \frac{\max(\beta^{(-)})-\min(\beta^{(-)})}{2}\, .
\end{align*}
\end{proof}

\begin{figure*}
    \centering
    \includegraphics[width=0.96\linewidth]{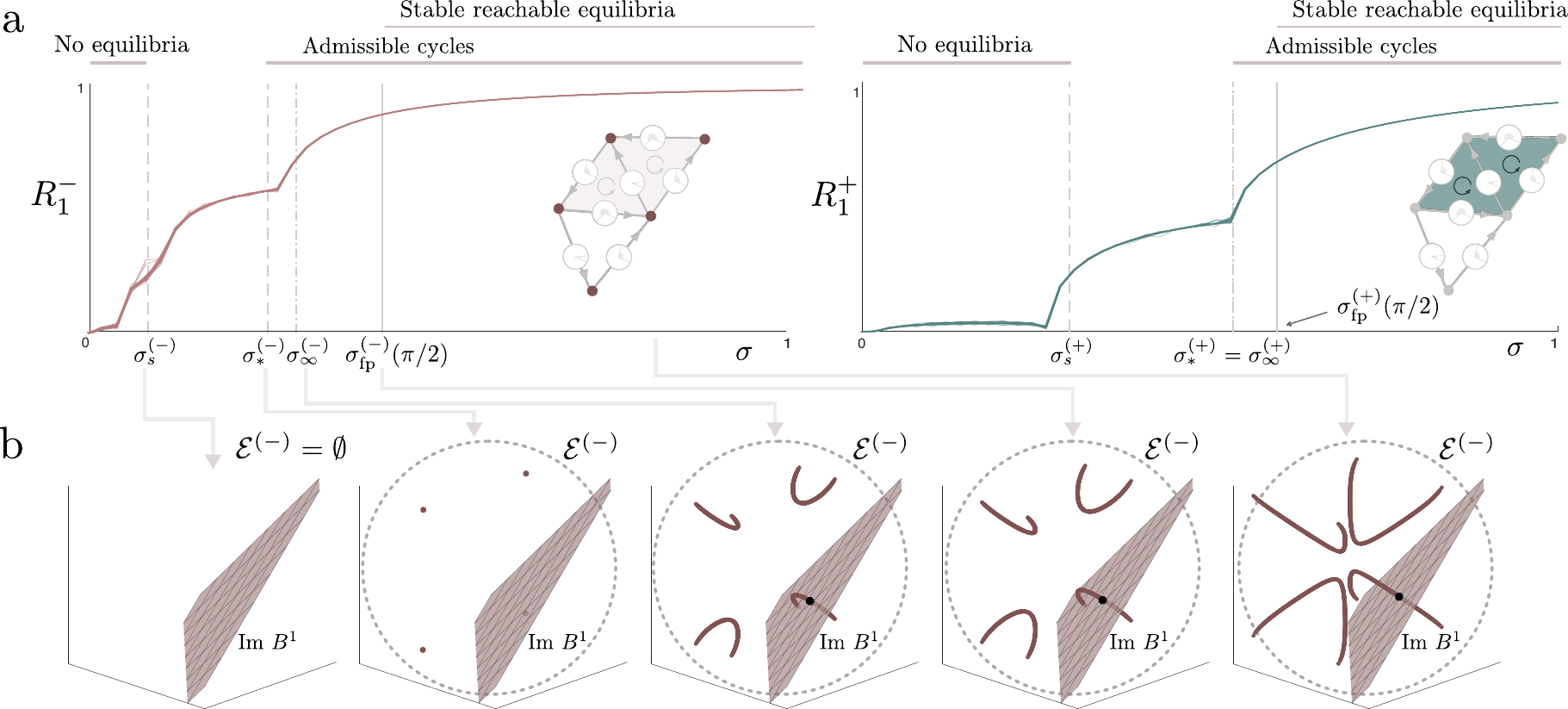}
    \caption{{\bf a.} Fixing the natural frequencies $\omega$, we simulate the edge simplicial Kuramoto model on a small simplicial complex with $20$ different initial phase configurations, for values of $\sigma\in [0,1]$, and compute the time-averaged partial order parameters $R_1^-$ (left), $R^+_1$ (right) from $t=0$ to $t=1000$. The vertical lines correspond to the values of $\sigma_s$ (\cref{prop:necessary_condition_2_ball}), $\sigma_*$ (\cref{theorem:critical_coupling}), $\sigma_\infty$ (\cref{eq:sigma_infinity}) and $\sigma_{\mathrm{fp}}$ (\cref{theorem:sigma_fp}). If we identify the last ``jump'' in the order with the emergence of reachable equilibria, then we see how the special values of $\sigma$ we derived actually bound its value from below and from above. As predicted by \Cref{theorem:critical_coupling_+_no_cycles}, the equilibrium transition value for the $(+)$ projection is exactly $\sigma^{(+)}_* = \sigma^{(+)}_\infty$. {\bf b.} The meaning of the different values of $\sigma$ is depicted by numerically computing the equilibrium set $\mathcal{E}^{(-)}\subset\R^3$ (\cref{def:eq_sets}) for the edge dynamics on a $2$-simplex. We see how the equilibrium set $\mathcal{E}^{(-)}$ is empty for $\sigma=\sigma_s^{(-)}$, it first appears as a discrete set of points for $\sigma = \sigma_*^{(-)}$ and grows, intersecting the plane $\Ima B^1$ for $\sigma \geq \sigma^{(-)}_{\mathrm{fp}}$ giving rise to a reachable equilibrium (\cref{def:reachable_equilibria}), marked here as a black dot.}
    \label{fig:bounds_compare}
\end{figure*}

It is now worth noting that the properties of the projections $\theta^{(+)} =D^k\theta$ and $\theta^{(-)}=B^k\theta$ are not entirely symmetrical, as the space of $(k+1)$-cycles can be trivial ($\ker B^{k+1} = \sset{0}$) and thus there are situations in which the space of admissible cycles is simply $\mathcal{A}^{(+)}=\sset{0}$. The same cannot be said for $\ker D^{k-1}$ resulting in the down projection $(-)$ being generally harder to treat. 
To shed more light on this, let us consider the $k$-th order dynamics on a simplicial complex $\mathcal{X}$ which has at least one $k$-simplex. From \cref{def:admissible_vectors}, the existence of equilibria for both of them depends on the presence or absence of admissible cycles which, respectively, must belong to $\ker B^{k+1} $ and $\ker D^{k-1}$. The asymmetry stems from the fact that $D^{k-1}$ cannot have a trivial kernel because
\begin{align}\label{eq:decomp_of_admissible}
\ker D^{k-1} = (\Ima B^k)^\perp \underbrace{=}_{\text{Hodge}} \Ima D^{k-2} \oplus \ker L^{k-1}\, , 
\end{align}
and thus
\begin{itemize}
    \item if $k = 1$ then $\ker D^{0} = \ker L^0$, which is non-trivial as there is at least one connected component;
    \item if $k>1$ then $\mathrm{dim} \ker D^{k-1} \geq \mathrm{dim} \Ima D^{k-2}$ which is nonzero because, by inclusion, there is a nonzero number of $(k-1)$-simplices and $D^{k-2}$ is not an all-zero matrix.
\end{itemize}
The same cannot be said for $B^{k+1}$ as, in general, there is no restriction on the number of $(k+1)$-cycles. In fact, on the same line of \cref{eq:decomp_of_admissible},
\begin{align*}
\ker B^{k+1} = (\Ima D^k)^\perp = \Ima B^{k+2}\oplus\ker L^{k+1}\, ,
\end{align*}
which is empty when there are no $(k+2)$-simplices and no $(k+1)$-holes. 
Therefore, the case of $ \ker B^{k+1} = \sset{0}$ deserves a special treatment.
\begin{theorem}[No higher-order cycles]\label{theorem:critical_coupling_+_no_cycles}
If there are no $(k+1)$-cycles (~$\ker B^{k+1}=\sset{0}$) then the following properties hold:
\begin{enumerate}
    \item if $\mathcal{A}^{(+)}\neq\emptyset$ then $\mathcal{A}^{(+)}=\sset{0}$;
    \item if $\mathcal{A}^{(+)}\neq\emptyset$ then the equilibrium set is a discrete set of points given by
    \begin{align}
    \mathcal{E}^{(+)} = \sset{ (-1)^{s_+}\odot\arcsin\left(\frac{\beta^{(+)}}{\sigma^\uparrow}\right) + s_+\pi: s\in\sset{0,1}^{n_{k+1}}}\, ;
    \end{align}
    \item $\sigma^{(+)}_* = \norm{\beta^{(+)}}_{\infty}$;
    \item All equilibria are reachable $\mathcal{E}^{(+)} = \mathcal{R}^{(+)}$.
\end{enumerate}
\end{theorem}

\begin{proof}
We prove each statement below.\\
1. It is trivial because $0$ is the only vector in $\ker B^{k+1}$.\\
2. Directly follows from \cref{eq:equilibrium_conditions_2.5} with $x^{(+)}=0$.\\
3. $0$ is the only vector in $\ker B^{k+1}$ so it will be admissible if and only if
\begin{align*}
\norm{\frac{\beta^{(+)}}{\sigma^\uparrow}}_\infty \leq 1\, .
\end{align*}
The smallest value of $\sigma^\uparrow$ for which this holds is $\sigma^\uparrow = \norm{\beta^{(+)}}_\infty$.
\item According to Definition \ref{def:reachable_equilibria}, an equilibrium is reachable for the $(+)$ projection if it belongs to $\Ima D^k$. In this case,
\begin{align*}
    \Ima D^k &= (\ker (D^k)^*)^\perp \\
     &= (\ker B^{k+1})^\perp = \sset{0}^\perp = \R^{n_{k+1}}\, ,
\end{align*}
hence the thesis.
\end{proof}

Notice how, in this case, the critical coupling $\sigma^{(+)}_*$ is both the transition value for the existence of admissible cycles (by \cref{def:admissible_vectors}) and the existence of the equilibria of the divergence-free component (by \cref{prop:curldiv-projected_equivalence}).
This means that $\sigma^\uparrow \geq \sigma^{(+)}_*$ is a necessary and sufficient condition for the existence of phase-locked configuration. 
From this general result, we can obtain \emph{for free} the well-known \cite{Stability_kuramoto_model} exact equilibrium transition for the node Kuramoto on trees as, by definition, they have no $1$-cycles.

\subsection{Sufficient condition for phase locking}\label{subsection:sufficient}
Necessary conditions for equilibrium are useful in a setting where we are interested in pushing the system to a \emph{non-equilibrium} state. If, in fact, we are able to tune the coupling strength below one of the bounds derived above ($\sigma_s$ or $\sigma_*$), we are guaranteed that the system will not reach a phase-locked configuration. If, however, we want the system to be phase-locked, we need \textit{sufficient} conditions that can ensure the existence of such equilibria.

An elegant bound on $\sigma^\updownarrow$, that both ensures the existence of equilibria and that is easy to compute, can be found  generalizing one of the results proven in \cite[Theorem 4.7]{dorfler2012exploring} by using the proof technique first introduced in~\cite{Stability_kuramoto_model} for the node Kuramoto.
\begin{theorem}[Sufficient condition for the existence of stable reachable equilibria]\label{theorem:sigma_fp}
For any $\gamma\in (0,\pi/2)$, if
\begin{align}\label{sigma_fixed_point}
    \sigma^\updownarrow \geq \sigma_{\mathrm{fp}}^{(\pm)}(\gamma) \defeq \frac{\sqrt{\max_i w^{k\pm 1}_i}}{\sin(\gamma)}\norm{\beta^{(\pm)}}_{w^{k\pm 1}}\, ,
\end{align}
there exists an asymptotically stable reachable equilibrium for the $(\pm)$ dynamics such that
\begin{align}
\norm{\theta^{(\pm)}}_{\infty} \leq \gamma\, .
\end{align}
\end{theorem}
\begin{proof}
The proof directly follows the constructions in ~\cite[Theorem 2]{Stability_kuramoto_model} by rewriting the equilibrium equation for the projection dynamics as a fixed point equation (hence the subscript \emph{fp} in $\sigma$) $x=f(x)$ and finding $\sigma^\updownarrow$ such that $f$ is a continuous function from a convex compact set to itself. Brouwer's fixed point theorem then provides the existence of a fixed point (a reachable equilibrium) in this set. The full proof can be found in Appendix~\ref{appendix:proof}.
\end{proof}
Four important observations should be highlighted from this result: 
\begin{enumerate}
    \item it is \textit{always} possible to tune the coupling strengths in order for the curl-free and divergence-free components to independently reach equilibrium; 
    \item after a certain value of the coupling strength, these equilibrium configurations always exist and at least one of them is close to the origin;
    \item increasing the coupling will also increase the closeness of the equilibrium to the origin;
    \item we see from the definition of the simplicial order parameter \cref{eq:SOP} that, if each component of the projection is close to $0$, then the configuration will be such that $R_k \approx 1$ i.e. phase synchronized (\cref{subsection:SOP}).
\end{enumerate}
We also highlight that, when the complex is unweighted, the expression of the bound becomes
\begin{align}
    \sigma^{(\pm)}_{\mathrm{fp}}(\gamma) = \frac{1}{\sin(\gamma)}\norm{\beta^{(\pm)}}_2\, .
\end{align}
Tor the node Kuramoto and for $\gamma = \frac{\pi}{2}$, it reduces to
\begin{align}
    \sigma^{(+)}_{\mathrm{fp}} = \norm{B_1^\dagger\omega}_2 = \norm{(L^0)^\dagger B_1^\top \omega}_2\, , 
\end{align}
which, when approximated, gives the well-known bound
\begin{align}
\sigma \geq \frac{1}{\lambda_2(L^0)}\norm{B_1^\top\omega}_2\, ,
\end{align}
where $\lambda_2(L^0)$ is the Fiedler eigenvalue of the network.
Another interesting observation is that
\begin{align*}
\norm{\beta^{(+)}}^2_{w^{k+1}} &= \inner{(B^{k+1})^\dagger\omega}{(B^{k+1})^\dagger\omega}_{w^{k+1}} \\
&= \inner{\omega}{(B^{k+1}D^{k})^\dagger\omega}_{w^k}\\ 
&= \inner{\omega}{(L^k_{\uparrow})^\dagger\omega}_{w^k}\, ,
\end{align*}
which is exactly the effective resistance of $\omega$ as defined in \cite{resistance}. 
In other words, to have equilibrium, the coupling must overcome the ``structural'' resistance of the simplicial complex, encoded in both the incidence structure ($L^k_{\updownarrow}$) and the natural frequencies. 
This is a powerful observation because it means that it might be possible to define pairs of structures and frequencies to reach particular types of dynamics or control the frequencies to move across regimes. 

Finally, while \cref{theorem:sigma_fp} ensures the existence of reachable equilibria, in practice its value $\sigma^{(\pm)}_{\mathrm{fp}}$ tends to be conservative and to overestimate the minimum value of $\sigma$ for which stable reachable equilibria exist. In perfect analogy with the node Kuramoto literature~\cite{synch_survey}, it is often seen in practice that 
\begin{align}\label{eq:sigma_infinity}
\sigma^{(\pm)}_\infty \defeq \norm{\beta^{(\pm)}}_\infty
\end{align}
is closer to the true reachability threshold, and thus provides a sharper bound. This value, moreover, exactly coincides with the reachability transition in some special cases, such as in Thm.~\ref{theorem:critical_coupling_+_no_cycles}. 
The different bounds on $\sigma$ found in \cref{subsection:necessary,subsection:sufficient} are shown in \cref{fig:bounds_compare}, where they are related to the partial order parameters on a small simplicial complex. 
We see how $\sigma_*$ and $\sigma_{\mathrm{fp}}$ actually bound the point of the last jump, corresponding to the transition value after which the dynamics admits reachable equilibria.

\section{Coupling the Hodge components}\label{coupling}

The simplicial Kuramoto model of \cref{eq:simplicial_kuramoto} provides a natural way to formulate synchronization dynamics of topological signals interacting on a simplicial complex~\cite{ghorbanchian2021higher}.
Building upon its form, many different variants with interesting behaviors can be formulated. The first models we consider are those for which the Hodge decomposition of the dynamics does not lead to decoupled equations.
\begin{itemize}
    \item In \cref{subsection:explosive}, we review the explosive model, proposed in \cite{millan2020explosive}, which couples the Hodge components through the order parameters. We state the model and propose a similar variant, obtained as a gradient flow, which lends itself to an easier analytical treatment.
    \item In \cref{subsection:frustration}, we consider Sakaguchi-Kuramoto type models, where the dynamics is frustrated by an external parameter. This classical variant is extended to the simplicial case in two ways: the first directly follows from the node Kuramoto and the second, proposed in~\cite{arnaudon2022connecting}, adds frustration in an orientation-independent fashion.
\end{itemize} 

\subsection{Explosive simplicial Kuramoto}\label{subsection:explosive}

One of the early works on the simplicial Kuramoto model proposed to couple the different Hodge components of the dynamics with factors depending on the partial order parameters~\cite{millan2020explosive}. 
In that work, with the partial order parameters
\begin{align}\label{eq:nonnegative_order}    
\begin{split}
R_k^{[+]}(\theta) &= \frac{1}{n_{k+1}}\abs{\sum_{\alpha=1}^{n_{k+1}} e^{i(D^k\theta)_\alpha}}\\
R_k^{[-]}(\theta) &= \frac{1}{n_{k-1}}\abs{\sum_{\alpha=1}^{n_{k-1}} e^{i(B^k\theta)_\alpha}}\, , 
\end{split}
\end{align}
for a $k$-cochain $\theta$, the following dynamical system is defined
\begin{align}\label{eq:explosive_millan}
\Dot{\theta} = \omega &- \sigma^\downarrow R_k^{[+]}(\theta)D^{k-1}\sin\left(B^k\theta\right)\nonumber \\
&- \sigma^\uparrow R_k^{[-]}(\theta) B^{k+1}\sin\left(D^k\theta\right)\, ,
\end{align}
which was shown to display explosive transitions in the order parameters $R_k^{[\pm]}$ when varying $\sigma$. 

The partial order parameters used by~\cite{millan2020explosive} (\cref{eq:nonnegative_order}) are different from the ones defined here in \cref{def:partial_order_parameters_1,def:partial_order_parameters_2}.
Indeed, our formulation allows for negative values (\cref{subsection:SOP}) as well as a derivation of the simplicial Kuramoto dynamics as a gradient flow. 
We show here that a nonlinearity introduced into the potential allows us to formulate an explosive model analogous to~\cite{millan2020explosive}.
From the two partial order parameters defined in~\cref{def:partial_order_parameters_1,def:partial_order_parameters_2}, we can consider their product to define the \emph{explosive simplicial Kuramoto} model as
\begin{align}
    \Dot{\theta} = C^{+}_k C^{-}_k W_{k}\nabla_{\theta} (R^+_k R^-_k)\, , 
\end{align}
whose explicit dynamics is
\begin{align}\label{eq:adaptively_coupled}
\begin{split}
\Dot{\theta} = \omega -\sigma^{\downarrow} R^+_k(\theta)D^{k-1}\sin\left(B^k\theta\right)\\
-\sigma^{\uparrow}R^-_k(\theta)B^{k+1}\sin\left(D^k\theta\right)\, .
\end{split}
\end{align}
where we have introduced coupling strengths and natural frequencies for generality.
The projected dynamics and the Hodge components are now coupled because the interaction term from below depends only on $\theta^{(-)}$ but $R^+_k$ depends on $\theta^{(+)}$, and vice versa for the interaction from above. 
This nonlinear gradient flow dynamics is different from \cref{eq:explosive_millan}, but still displays an explosive phase transition in the order parameter, even for small simplicial complexes (\cref{fig:model_comparison}a). 
The possibility of having a negative order parameter in front of the interaction terms, moreover, can make the model behave in such a way as to maximize the phase difference between interacting oscillators~\cite{StrogatzContrarian}, giving rise to new dynamical phenomena. 
As shown in~\cref{fig:model_comparison}a, the model shows a second phase transition in $\sigma$ after which the dynamics is bistable and can converge to both phase ($R_k\approx 1$) and anti-phase ($R_k\approx -1$) synchronized configurations.

\begin{figure*}[htp]
    \centering
    \includegraphics[width=0.95\linewidth]{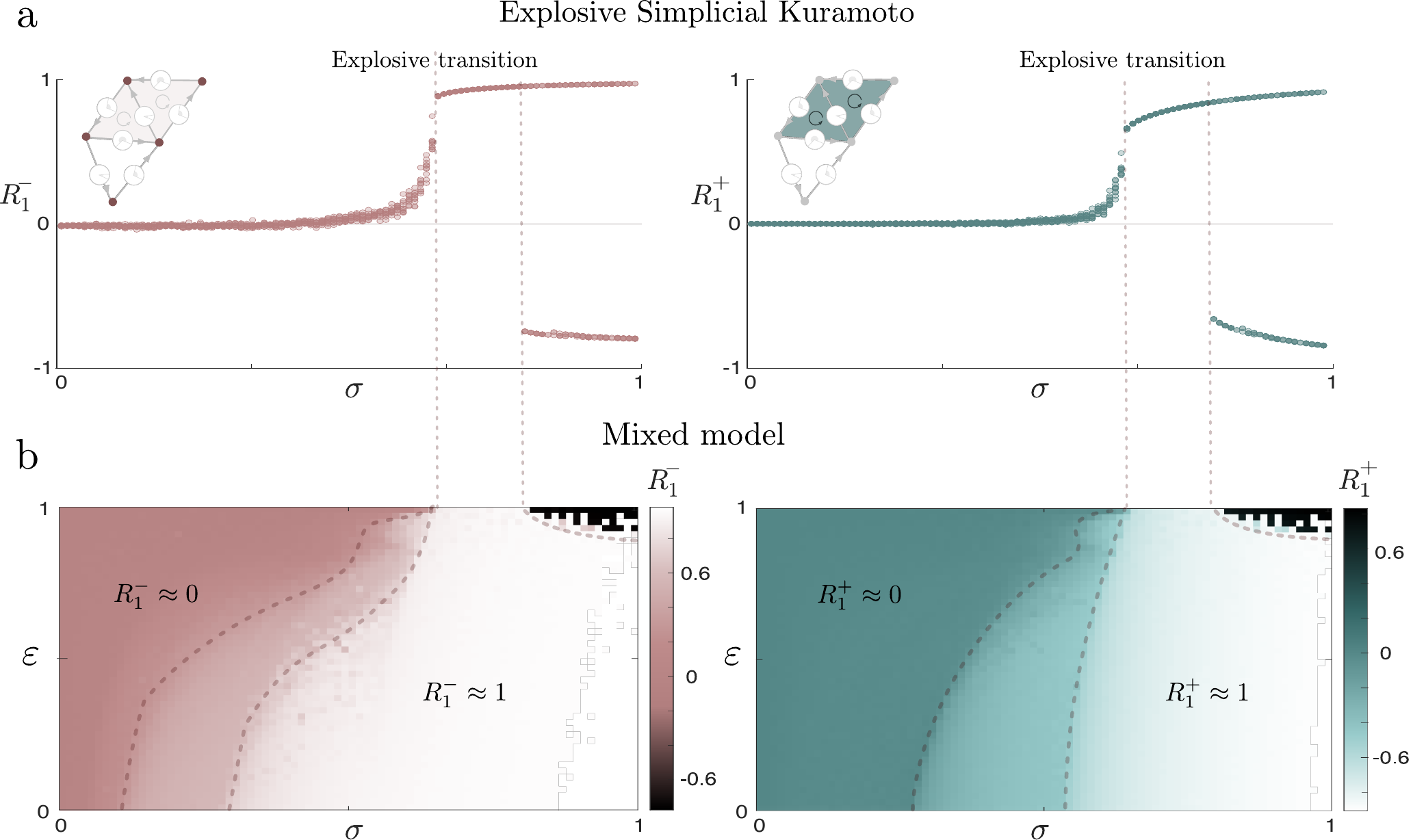}
    \caption{
    Fixing the natural frequencies and the frustrations, we run variants of the simplicial Kuramoto model on a small simplicial complex for different values of $\sigma\in [0,1]$ and compute the time-averaged partial order parameters (red: lower order parameter, blue: upper order parameter). {\bf a.} The explosive model \cref{eq:adaptively_coupled} shows an explosive transition in both the down- (left) and up- (right) partial order parameters. After a certain value of $\sigma$, moreover, some of the trajectories converge to anti-phase synchronized configurations characterized by the order parameter being close to $-1$. 
    {\bf b.} The phase diagram of mixed model \cref{eq:mixed_model} is depicted for $\sigma\in[0,1],\, \varepsilon\in [0,1]$. For any given $\varepsilon$, the dashed lines show the $\sigma$ corresponding to the first and last jump in the order. As we can see, they converge to a single point when $\varepsilon = 1$, signaling the explosiveness of the model. The dashed line on the right encircles the region in which the system is bistable, and the trajectories can converge to both phase and anti-phase synchronized configurations.}
    \label{fig:model_comparison}
\end{figure*}

Interestingly, for this model, we can find a sufficient condition for the existence of a phase-locked configuration analogous to \cref{theorem:sigma_fp}. In this case, as expected, the bounds related to the $(+)$ and $(-)$ projections are coupled. 
\begin{theorem}[Sufficient condition for the existence of reachable equilibria]\label{theorem:sigma_fp_explosive}
For any $\gamma^{(+)},\gamma^{(-)}\in (0,\pi/2)$, if
\begin{align}\label{sigma_fixed_point_explosive}
\begin{cases}
    \sigma^\uparrow \geq \frac{\sqrt{\max_i w^{k+1}_i}}{\sin(\gamma^{(+)})\cos(\gamma^{(-)})}\norm{\beta^{(+)}}_{w^{k+1}}\\[10pt]
    \sigma^\downarrow \geq \frac{\sqrt{\max_i w^{k-1}_i}}{\sin(\gamma^{(-)})\cos(\gamma^{(+)})}\norm{\beta^{(-)}}_{w^{k-1}}
\end{cases}\, , 
\end{align}
then both the projections $\theta^{(+)},\theta^{(-)}$ of the explosive simplicial Kuramoto model \cref{eq:adaptively_coupled} admit reachable equilibria such that
\begin{align}
\norm{\theta^{(+)}}_\infty \leq \gamma^{(+)},\ \norm{\theta^{(-)}}_\infty \leq \gamma^{(-)}\, .
\end{align}
\end{theorem}
\begin{proof}
The proof is similar to the one of \cref{theorem:sigma_fp} and can be found in \cref{appendix:proof_explosive}.
\end{proof}
It is interesting to see that, in the bound above, there is a tradeoff between the coupling strengths of the two projections. If we want a low bound on the coupling for the $(+)$ projection, then we need $\gamma^{(+)}$ to be high and $\gamma^{(-)}$ to be low. By doing so, however, we will result in a high value of the bound for the $(-)$ component.
Notice, moreover, how this result is only concerned with stating the presence of a phase-locked configuration whose projections can independently be made arbitrarily close to the origin (and thus with high values of the order parameters) by tuning the couplings, but states nothing about whether the dynamics will actually converge to them. In addition, stability analysis is challenging for this system and is left for future work.

This gradient flow approach suggests a more general way to build variants of the simplicial Kuramoto model which couple the dynamics across Hodge subspaces. 
For this, we can consider a general function $f$ of the partial order parameters and take its gradient
\begin{align}
\label{eq:explosive-gradient-flow}
   \dot \theta  = W_{k}\nabla_{\theta} f(R^+_k,R^-_k)\, ,
\end{align}
which, modulo normalization constants, reduces to the standard simplicial Kuramoto \cref{eq:simplicial_kuramoto} for $f(x,y) = x + y$.
As an example, if we consider a linear interpolation between the standard potential and the explosive one, $f_\varepsilon(x,y) = (1-\varepsilon)(x+y)+\varepsilon xy$, parametrized by $\varepsilon\in [0,1]$, we have what we call \emph{mixed model}
\begin{align*}
    \dot\theta = W_{k} \nabla_{\theta} ((1-\varepsilon)C_kR_k + \varepsilon C_k^+C_k^- R_k^+ R_k^-)\, ,
\end{align*}
which explicitly reads
\begin{align}\label{eq:mixed_model}
\begin{split}
    \Dot \theta 
     &= - (1-\varepsilon+\varepsilon R^-_k(\theta)) B^{k+1}\sin(D^k\theta)\\
     &- (1-\varepsilon + \varepsilon R^+_k(\theta)) D^{k-1}\sin(B^k\theta)\, .
    \end{split}
\end{align}
For $\varepsilon=0$, we recover the standard simplicial Kuramoto model, and for $\varepsilon = 1$ we get \cref{eq:adaptively_coupled}. The phase diagram of this dynamics is shown in \cref{fig:model_comparison}b, where the explosive transition for $\varepsilon = 1$ is evident and a region of bistability appears for high values of $\sigma$. 
Notice that, although the potential is linear in $\varepsilon$, the stationary dynamic is not, and, through an analogous proof, it is possible to derive a result equivalent to \cref{theorem:sigma_fp_explosive} to get sufficient conditions for the existence of reachable equilibria.

\subsection{Simplicial Sakaguchi-Kuramoto}\label{subsection:frustration}
The Sakaguchi-Kuramoto model~\cite{SakaguchiKuramoto} is a well-known extension of the Kuramoto model, which modifies the interaction function by including a phase lag parameter. 
Given a frustration vector $\alpha$ on the edges, we can write it as a modification of \cref{eq:standard_kuramoto}
\begin{align}\label{eq:kuramoto_sakaguchi_nodes}
\Dot{\theta}_i = \omega_i - \sigma\sum_j A_{ij}\sin(\theta_i - \theta_j + \alpha_{ij})\, .
\end{align}
We can extend it to the simplicial case in a simple manner by considering two frustration cochains $\alpha_{k-1}\in C^{k-1}$ and writing the simplicial Kuramoto model 
\begin{align}\label{eq:sakaguchi_kuramoto_dependent}
\Dot{\theta} &=\omega - \sigma^\uparrow B^{k+1}\sin\left(D^k\theta + \alpha_{k+1}\right) \nonumber \\ 
&- \sigma^\downarrow D^{k-1}\sin\left(B^k\theta + \alpha_{k-1}\right)\, ,
\end{align}
where $\alpha$ is the effect of an external field on each interaction simplex.
As this model does not couple the Hodge subspaces, we refer to it as the \textit{simple} frustrated model.
In addition, while it has a simple form, it can be proven that it does not reduce to the Sakaguchi-Kuramoto model of \cref{eq:kuramoto_sakaguchi_nodes} for $k = 0$ \cite{arnaudon2022connecting}.
Before considering how to include frustrations in a more meaningful way, we notice that this simple frustration has the surprising property that $\alpha$ can be used to control the system by making any projected configuration reachable and stable.

\begin{theorem}[Control of reachable equilibrium in simple model]\label{theorem:make_eq_reachable}
If $\sigma^\updownarrow > \sigma^{(\pm)}_\infty = \norm{\beta^{(\pm)}}_\infty$, then, for any chosen projected configuration $\theta^{(\pm)}_{*}$, i.e. $\theta_{*}^{(+)}\in \Ima D^k$, $\theta^{(-)}_{*}\in\Ima B^k$, if
\begin{align}
\alpha_{k\pm 1} = \arcsin\left(\frac{\beta^{(\pm)}}{\sigma^\updownarrow}\right) -\theta^{(\pm)}_{*}\, , 
\end{align}
then $\theta_{*}^{(\pm)}$ is an asymptotically stable, reachable equilibrium for the $(\pm)$ projection of the simple frustrated dynamics \cref{eq:sakaguchi_kuramoto_dependent}.
\end{theorem}
\begin{proof}
We prove it for the $(+)$ projection, as the $(-)$ case is analogous.
We see from \Cref{eq:sakaguchi_kuramoto_dependent} that an equilibrium $\theta^{(+)}_{eq}$ of the $(+)$ projection of the frustrated dynamics will satisfy
\begin{align*}
\sin\left(\theta^{(+)}_{eq}+\alpha_{k+1}\right) = \frac{\beta^{(+)}}{\sigma^\uparrow} + x^{(+)}\, .
\end{align*}
As $\sigma^\uparrow \geq \norm{\beta^{(+)}}_\infty
$, $x^{(+)}=0$ is an admissible cycle (Def. \ref{def:admissible_vectors}) and thus
\begin{align*}
\theta^{(+)}_{eq} = \arcsin\left(\frac{\beta^{(+)}}{\sigma^\uparrow}\right)-\alpha_{k+1} = \theta^{(+)}_*\in \Ima D^k\, ,
\end{align*}
is a reachable equilibrium. The proof of stability can be found in \cref{appendix:make_eq_reachable}.
\end{proof}
We can visualize this result by looking at any panel of \cref{fig:bounds_compare}b and noticing that the action of a linear frustration corresponds to a translation of $\mathcal{E}^{(-)}$, resulting in a different intersection with the reachable subspace.
The strength of \cref{theorem:make_eq_reachable} is in the fact that, with a fine-tuned frustration, it is possible to have equilibrium configurations as ordered as we want while keeping the coupling strengths comparatively low. By exploiting this idea, we can get the following corollary.
\begin{corollary}\label{cor:0reachable}
Under the hypotheses of \cref{theorem:make_eq_reachable}, if
\begin{align}
\alpha_{k\pm 1} = \arcsin\left(\frac{\beta^{(\pm)}}{\sigma^{\updownarrow}}\right)\, , 
\end{align}
then $0\in\R^{n_{k\pm1}}$ is a stable, reachable equilibrium for the $(\pm)$ projection and thus there is a stable equilibrium configuration of the frustrated dynamics \cref{eq:sakaguchi_kuramoto_dependent} with partial order parameter $R^\pm_k(\theta) = 1$.
\end{corollary}
\begin{proof}
Simply follows by applying \Cref{theorem:make_eq_reachable} with $\theta^{(\pm)}_* = 0$ and using the definition of simplicial order parameter \cref{eq:SOP}.
\end{proof}

\begin{figure}
    \centering
    \includegraphics[width=0.85\linewidth]{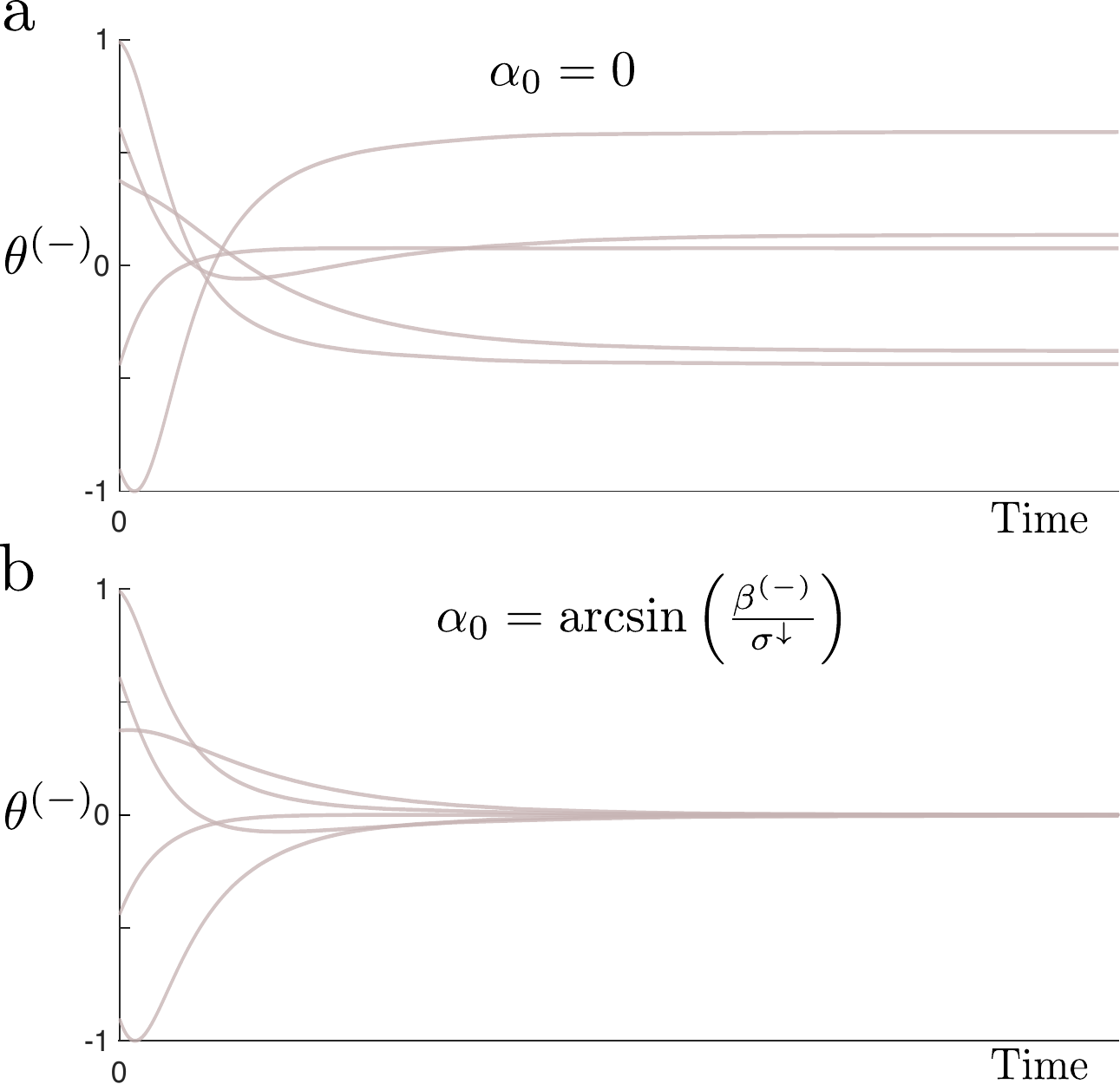}
    \caption{
    Application of \cref{cor:0reachable} to the $(-)$ projection of the dynamics on a small simplicial complex. Tuning the frustration cochain it is possible to have a stable equilibrium configuration such that $\theta^{(-)}_{eq}=0$, as shown by the bottom panel.}
    \label{fig:makeitreachable}
\end{figure}

An application of this corollary can be seen in \Cref{fig:makeitreachable}, where the configuration where all phases differences $\theta^{(-)}$ are equal to $0$ is made reachable by tuning the frustration.
The problem with this simple frustration formulation comes from the oriented nature of the simplices. Intuitively, two oscillating simplices with a common subface $a$ will see the frustration on $a$ with a different sign, depending on their relative orientations. In~\cite{arnaudon2022connecting}, this issue is addressed by lifting the phases cochains, similarly to~\cite{SchaubSirev}, into another space where both orientations are present, and by then projecting back to obtain a model which is independent on the orientation of $(k+1)$-simplices. 
We write here the resulting equation, slightly generalized from~\cite{arnaudon2022connecting} to include orientation-independent frustrations on $(k-1)$-simplices too.
We define the lift operators
\begin{align}
V^k = 
\begin{pmatrix}
I_{n_k} \\
-I_{n_k}
\end{pmatrix},\ \ 
U^k = 
\begin{pmatrix}
I_{n_k} \\
I_{n_k}
\end{pmatrix}\, , 
\end{align}
and indicate with $(A)^\pm \defeq (A\pm |A|)/2$ the projection of a matrix onto its positive or negative components.
Using these definitions, we can write the \emph{orientation-independent simplicial Sakaguchi-Kuramoto model}
\begin{align}\label{eq:sakaguchi_kuramoto_independent}
\dot{\theta} = \omega - \sigma^\uparrow \left(B^{k+1}(V^{k+1})^\top\right)^- \sin\left(V^{k+1}D^k\theta + U^{k+1}\alpha_{k+1}\right) \nonumber\\
-\sigma^\downarrow\left(D^{k-1}(V^{k-1})^\top\right)^- \sin\left(V^{k-1}B^k\theta + U^{k-1}\alpha_{k-1}\right)\, . 
\end{align}

Note that the projection of the external operator onto its negative components is nonlinear, and thus changes its image and kernel. This means that the Hodge decomposition of \Cref{eq:sakaguchi_kuramoto_independent} will lead to components that are not evolving independently but are coupled. 
\begin{proposition}\label{prop:orientation_independence} \cite{arnaudon2022connecting}
It holds that \cref{eq:sakaguchi_kuramoto_independent} is independent on the orientation of the simplices of order $k-1$ and $k+1$.
\end{proposition}
\begin{proof}
Let us focus on the interaction from below, as the other case is completely symmetrical.
A change of orientation of a $(k-1)$-simplex indexed by $i$ can be encoded in the action of a diagonal matrix $P$ such that
$P_{jj} = 1$ if $j\neq i$ and $P_{ii} = -1$. The boundary and coboundary operators in the new simplicial complex with the orientation of $i$ flipped are
$\widetilde{D}^{k-1} = D^{k-1}P,\ \widetilde{B}^k = P B^k$.
We see that the change of orientation matrix related to simplex $i$ acts on the lift matrix from the right by swapping rows $i$ and $2i$
$V^{k-1}P = \widetilde{P}V^{k-1}$, where $\widetilde{P}$ is the corresponding permutation matrix.
The interaction term from below will then become
\begin{align*} 
&\left(\widetilde{D}^{k-1}(V^{k-1})^\top\right)^- \sin\left(V^{k-1}\widetilde{B}^k\theta + U^{k-1}\alpha_{k-1}\right) \\
&= \left(D^{k-1}(V^{k-1})^\top\right)^- \widetilde{P} \sin\left(\widetilde{P} V^{k-1}B^k\theta + \widetilde{P}U^{k-1}\alpha_{k-1}\right)\\
&= \left(D^{k-1}(V^{k-1})^\top\right)^- \sin\left(V^{k-1}B^k\theta + U^{k-1}\alpha_{k-1}\right)\, ,
\end{align*}
as $\widetilde{P}^2 = I$ and $\widetilde PU^{k-1}= U^{k-1}$, being $\widetilde{P}$ a row-swap operation between two equal rows.
\end{proof}

The orientation-independent model, and its associated  \Cref{prop:orientation_independence} (proven in \cite{arnaudon2022connecting}) can be better understood by making the following observations. 
If we consider arbitrary $\widetilde\alpha_{k\pm 1} = (\underline{\alpha}_{k\pm 1},\overline{\alpha}_{k\pm 1}) \in \R^{2n_{k\pm 1}}$ which are not necessarily of the form $U^{k+1}\alpha_{k\pm1} = (\alpha_{k\pm 1},\alpha_{k\pm 1})$, we can define the more general 
\emph{orientation-selective simplicial Sakaguchi-Kuramoto model}
\begin{align}\label{eq:sakaguchi_kuramoto_selective}
\dot{\theta} = \omega - \sigma^\uparrow \left(B^{k+1}(V^{k+1})^\top\right)^- \sin\left(V^{k+1}D^k\theta + \widetilde\alpha_{k+1}\right) \nonumber\\
-\sigma^\downarrow\left(D^{k-1}(V^{k-1})^\top\right)^- \sin\left(V^{k-1}B^k\theta + \widetilde\alpha_{k-1}\right)\, . 
\end{align}
In this case, we can see that a different frustration will act on $k$-simplices depending on their relative orientation. 
In particular, the elements of $\underline{\alpha}_{k\pm 1}$ represent frustrations on $(k\pm 1)$ simplices acting only on the $k$-simplices which are incoherently oriented with them, while the last components $\overline{\alpha}_{k\pm 1}$ will act only on coherently oriented simplices. 
Hence, if these two coincide, we have orientation independence. 
Indeed, it is enough to expand the lift matrices and projection operators to see that, for example, the term of the interaction from above can be rewritten as
\begin{align}
(B^{k+1})^-\sin(D^k\theta + \underline{\alpha}_{k+1}) + (B^{k+1})^+\sin(D^k\theta - \overline{\alpha}_{k-1})\, , 
\end{align}
and notice that the nonzero elements of $(B^{k+1})^-$ contain the adjacencies between incoherently oriented $k$ and $(k+1)$-simplices, while $(B^{k+1})^+$ contains only the coherently-oriented adjacencies. 
Moreover, when $\underline{\alpha}_{k+1} = - \overline{\alpha}_{k+1}$, then we can compact the two matrices $(B^{k+1})^\pm$ and get back the simple frustration of \cref{eq:sakaguchi_kuramoto_dependent}, which can now be interpreted as inducing opposite frustrations on coherently or incoherently oriented simplices. 
Finally, it should be possible to also control the equilibrium solution via $\widetilde \alpha_{k\pm1}$, but as this system is now coupled, both projections will have to be controlled together to obtain a consistent system.

\section{Coupling the different orders with the Dirac operator}\label{section:coupling_orders}

Up to this point, we have considered topological signals of a fixed order, on nodes, edges faces, and so on.
This approach gives rise to interesting types of interactions.
However, it does not fully exploit the multi-order nature of simplicial complexes, because it involves only $k$-simplices and their upper/lower adjacencies. 
This is a direct consequence of the fact that $B^{k}B^{k+1} =0$: coupling signals between, for example, nodes and faces, cannot be done with a simple concatenation of boundary operators.
Instead, we can generalize the simplicial Sakaguchi-Kuramoto models by letting the frustration vector be the signal of a lower/higher order on the same simplicial complex.
This can be formalized through the discrete Dirac operator (also known as Gauss-Bonnet operator~\cite{anne2015gauss}), first introduced in~\cite{lloyd2016quantum} in the context of simplicial complexes, and later used for synchronization\cite{calmon2022dirac} and signal processing ~\cite{calmon2023dirac}. 

\subsection{Discrete Dirac formalism}

For a simplicial complex with simplices up to order $K$, we can gather the phases into a single vector $\Theta = (\theta_{(0)},\theta_{(1)}, \ldots, \theta_{(K)})$ and define the Dirac operator on simplicial complexes~\cite{Bianconi_2021_Dirac,lloyd2016quantum,wee2023persistent} as the square, block tridiagonal matrix 
\begin{align}
\mathbf{D} \defeq \mathrm{tridiag}([D^0,\dots, D^{K-1}], [0,\dots, 0], [B^1,\dots, B^K])\, ,
\end{align}
where $0$ indicates the matrix of the right size with all zero elements.
The Dirac operator contains all the adjacency structure of the simplicial complex and it is, by construction, the ``square root'' of the Laplacian matrix of the complex, in the sense that its square is the block diagonal matrix of the Hodge Laplacians
\begin{align}
\mathbf{L} \defeq \mathbf{D}^2 = \diagm(L^0,\dots,L^{K})\, .
\end{align}
In~\cite{calmon2022dirac} it is shown how, on a network ($K=1$), we can elegantly write the evolution of the phases of oscillating nodes and edges under the simplicial Kuramoto dynamics with the Dirac operator as
\begin{align}\label{eq:Kuramoto_dirac}   
\Dot{\Theta} = \Omega - \sigma \mathbf{D}\sin(\mathbf{D}\Theta)\, ,
\end{align}
where $\Omega = (\omega_{(0)},\omega_{(1)})$ contains the natural frequencies. 
\Cref{eq:Kuramoto_dirac}, however, only corresponds to the simplicial Kuramoto model for phases on edges of a network and does not generalize to simplicial complexes of arbitrary order. 
Indeed, for a simplicial complex with nodes, edges, and triangles ($K = 2)$, the Dirac operator is
\begin{align*}
\mathbf{D} =
\begin{pmatrix}
0 & B^1 & 0\\
D^0 & 0 & B^2\\
0 & D^1 & 0
\end{pmatrix}\, ,
\end{align*}
and the corresponding Kuramoto model becomes
\begin{align*}
\mathbf{D}\sin(\mathbf{D}\Theta) = 
\begin{pmatrix}
B^1\sin\left(D^0\theta_{(0)}+B^2\theta_{(2)}\right)\\[2ex]
D^0\sin\left(B^1\theta_{(1)}\right)+B^2\sin\left(D^1\theta_{(1)}\right)\\[2ex]
D^1\sin\left(D^0\theta_{(0)} + B^2\theta_{(2)}\right)\, ,
\end{pmatrix}
\end{align*}
which does not correspond to three uncoupled simplicial Kuramoto models on the nodes, edges, and triangles.
It is nevertheless possible to write the simplicial Kuramoto models on all orders with the Dirac operator $\mathbf{D}$ by considering its decomposition into the sum of its upper and lower block triangular matrices. 
Indeed, instead of splitting $\mathbf{D}$ by order~\cite{calmon2023dirac}, we split  by \textit{type} of interaction to obtain a direct generalization for the boundary operators appropriate for the Dirac case
\begin{align}
\mathbf{D} = \boldsymbol{d} + \boldsymbol{\delta}\, ,
\end{align}
where
\begin{align}
   \boldsymbol{\delta} = \mathrm{tridiag}([0,\dots, 0], [0,\dots, 0], [B^1,\dots, B^K])\, ,
\end{align}
and 
\begin{align}
\boldsymbol{d} = \mathrm{tridiag}([D^0,\dots, D^{K-1}], [0,\dots, 0], [0,\dots, 0])\, .
\end{align}
If seen as operators from the direct sum of the cochain spaces $C^0(\mathcal{X})\oplus \dots\oplus C^K(\mathcal{X}) \cong \R^{n_0+\dots + n_K}$ (whose inner product is given in matrix form by $\mathbf{W}^{-1} = \mathrm{
diag
}(W^{-1}_0,\dots, W^{-1}_K)$) to itself, then one is the adjoint of the other, $\boldsymbol{d} = \boldsymbol{\delta}^*$ i.e.
\begin{align*}
\boldsymbol{d}= \mathbf{W}\boldsymbol{\delta}^\top \mathbf{W}^{-1}\, .
\end{align*}
It also follows from $B^{k}B^{k+1} = 0$ that these operators are nilpotent
\begin{align}
\boldsymbol{d}^2 & = \boldsymbol{\delta}^2 = 0\, , 
\end{align}
and their products give the block diagonal matrices of up and down Laplacians
\begin{align}
\mathbf{L}_\downarrow  \defeq \boldsymbol{d}\boldsymbol{\delta}\qquad\mathrm{and}\qquad \mathbf{L}_\uparrow  \defeq \boldsymbol{\delta}\boldsymbol{d}\, .
\end{align}

As an example, for $K=2$, we have
\begin{align*}
\boldsymbol{\delta}
+ \boldsymbol{d} = 
\begin{pmatrix}
0 & B^1 & 0\\
0 & 0 & B^2\\
0 & 0 & 0
\end{pmatrix} +
\begin{pmatrix}
0 & 0 & 0\\
D^0 & 0 & 0\\
0 & D^1 & 0
\end{pmatrix}\, ,
\end{align*}
hence the two Laplacians are
\begin{align*}
\boldsymbol{L}_\uparrow = 
\begin{pmatrix}
    L^0_{\uparrow} & 0 & 0\\
    0 & L^1_{\uparrow} & 0\\
    0 & 0 & 0
\end{pmatrix},\, \boldsymbol{L}_\downarrow = 
\begin{pmatrix}
    0 & 0 & 0\\
    0 & L^1_{\downarrow} & 0\\
    0 & 0 & L^2_{\downarrow}
\end{pmatrix}\, .
\end{align*}
Moreover, we have that
\begin{align}
\mathbf{L} = \mathbf{D}^2 = (\boldsymbol{d} + \boldsymbol{\delta})^2 = \boldsymbol{d}\boldsymbol{\delta} + \boldsymbol{\delta}\boldsymbol{d}\, , 
\end{align}
which suggests an elegant way to write the evolution of the phases of all simplices in the complex under the simplicial Kuramoto dynamics as
\begin{align}\label{eq:dirac_synch}
    \dot \Theta = \Omega - \sigma^\uparrow\boldsymbol{\delta} \sin\left (\boldsymbol{d}\Theta\right) - \sigma^\downarrow\boldsymbol{d} \sin\left (\boldsymbol{\delta}\Theta\right)\, , 
\end{align}
where $\boldsymbol{\delta}\sin(\boldsymbol{d}\Theta)$ contains all the interaction terms from above and $\boldsymbol{d}\sin(\boldsymbol{\delta}\Theta)$ all the ones from below. 
It is easy to check that on a network ($K = 1$) we recover $\boldsymbol{D}\sin(\boldsymbol{D}\Theta)$ as an interaction term. 
\Cref{eq:dirac_synch}, moreover, has the same form of the simplicial Kuramoto model of \cref{eq:simplicial_kuramoto}, and thus can be written as a gradient flow
\begin{align}  \label{eq:super-gradient-flow}
\Dot{\Theta} = \Omega + C\mathbf{W}\nabla_\Theta \mathbf{R}(\Theta)\, ,
\end{align}
with the Dirac order parameter defined as
\begin{align}
\mathbf{R}(\Theta) = \frac{1}{C}\left(\ones^\top\mathbf{W}^{-1}\cos(\boldsymbol{d}\Theta) +  \ones^\top\mathbf{W}^{-1}\cos(\boldsymbol{\delta}\Theta)\right)\, , 
\end{align}
with normalization constant $C = \ones^\top \mathbf{W}^{-1}\ones$.
It can also be written as 
\begin{align}
    \boldsymbol{R}(\Theta) = \frac{1}{C}\sum_{k=1}^K C_kR_k(\theta_{(k)})\, , 
\end{align}
or in terms of the partial Dirac order parameters
\begin{subequations}
\begin{align} \label{def:dirac_partial_order_parameters_2}
\boldsymbol R^-(\theta) &\defeq \frac{1}{C^{-}}\ones^\top W^{-1}\cos\left(\boldsymbol{\delta}\Theta\right)\\ 
\qquad\boldsymbol R^+(\Theta) &\defeq \frac{1}{C^{+}}\ones^\top W^{-1}\cos\left(\boldsymbol{d}\Theta\right)\, ,
\end{align}
\end{subequations}
where the normalization constants $C^{\pm} = \ones^\top \boldsymbol{W}^{-1}\ones$ as
\begin{align}
C \boldsymbol{R}(\Theta) = C^+ \boldsymbol{R}^+(\Theta) + C^- \boldsymbol{R}^-(\Theta) \, . 
\end{align}

Naturally, we also have the Hodge decomposition \cref{eq:hodge_decomposition} on all orders
\begin{align}
    \bigoplus_{k=1}^{K} C^k(\mathcal X) = \mathrm{Im} \, \boldsymbol{\delta} \oplus \mathrm{ker}\, \boldsymbol{L} \oplus \mathrm{Im}\, \boldsymbol{d}\, .
\end{align}
In \cref{eq:dirac_synch}, however, the phases of the simplices of different orders evolve independently of one another, as \Cref{eq:dirac_synch} is just a formal reformulation to include all possible simplicial Kuramoto models that exist on a simplicial complex of order $K$ into a single formula. 
The advantage of this formulation is that it provides a general mathematical framework to couple the dynamics across different orders.

\subsection{Explosive Dirac Kuramoto dynamics}

To couple the dynamics across orders, \cite{ghorbanchian2021higher} proposed to multiply the interaction with the factor depending on the order parameters \cref{eq:nonnegative_order} of the dynamics above and below, as a Dirac generalization of the earlier explosive model of~\cite{millan2020explosive}. 
For a network ($K=1$), this coupling was made into  a so-called \emph{Nodes-Links} (NL) model~\cite{ghorbanchian2021higher}
\begin{align}
\begin{cases}
\Dot{\theta}_{(0)} = \omega_{(0)} - R^{[-]}_1(\theta_{(1)}) B^1\sin(D^0 \theta_{(0)}) \\
\Dot{\theta}_{(1)} = \omega_{(1)} - R_0^{[+]}(\theta_{(0)})D^0\sin(B^1\theta_{(1)})\, , 
\end{cases}       
\label{eq:dirac_NL}
\end{align}
and, with non-oscillating triangles, into the \emph{Nodes-Links-Triangles} (NLT) model~\cite{ghorbanchian2021higher}
\begin{align}
\begin{cases}
\Dot{\theta}_{(0)} = \omega_{(0)} - R^{[-]}_1(\theta_{(1)}) B^1\sin(D^0 \theta_{(0)}) \\
\begin{aligned}
\Dot{\theta}_{(1)} = \omega_{(1)} &- R_0^{[+]}(\theta_{(0)})R_1^{[+]}(\theta_{(1)})D^0\sin(B^1\theta_{(1)})\\[5pt]
& - R_1^{[-]}(\theta_{(1)})B^2\sin(D^1\theta_{(1)})\, .
\end{aligned}
\end{cases}   
\label{eq:dirac_NLT}
\end{align}

Inspired by these formulations, we can write order-coupled models in the Dirac formalism. As we did in \cref{subsection:explosive} with the explosive model, we can write the following nonlinear gradient flow
\begin{align} \label{eq:nonlinear-super-gradient-flow}
\Dot{\Theta} = \Omega + C^+C^-\mathbf{W}\nabla_\Theta\left ( \mathbf{R}^+\mathbf{R}^-\right)\, ,
\end{align}
where the coupling is global as it depends on all simplices of all orders. 
Alternatively, we can have a coupling across only adjacent orders with
\begin{align} \label{eq:local-nonlinear-super-gradient-flow}
\Dot{\Theta} = \Omega + \mathbf{W}\nabla_\Theta \left(\sum_{k=1}^K C^+_{k-1}C^-_{k} R_{k-1}^+R_k^-\right)\, ,
\end{align}
which is such that, for all orders $k$, the interaction term from below of order $k$ will depend on the order parameter from above of order $k-1$ and vice versa.
We leave the analysis of these models for future works.

\begin{figure*}
    \centering    
    \includegraphics[width=0.7\textwidth]{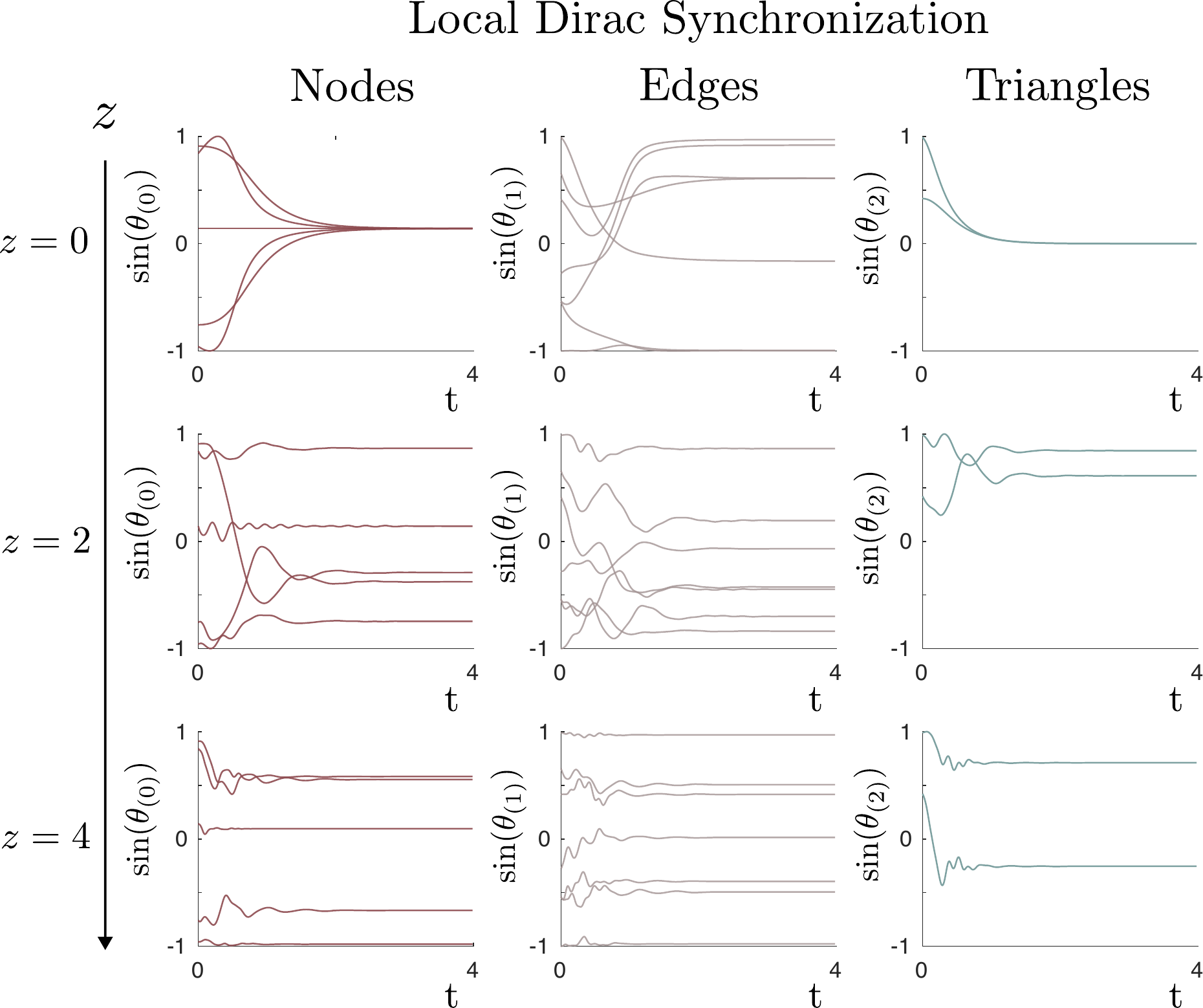}
    \caption{Local Dirac Synchronization (\cref{eq:dirac_nlt}) on nodes, edges, and triangles of the small simplicial complex depicted in \cref{fig:model_comparison}a, with $\Omega = 0$. We simulate the dynamics for different values of $z$ and see how the coupling it induces disrupts synchronization and results in the emergence of damped oscillations.
    }
    \label{fig:lds}
\end{figure*}

\subsection{Frustrated Dirac Kuramoto model}
Just as in~\cite{calmon2022dirac,calmon2023local}, one can instead consider a local coupling with the half super-Laplacian matrices $\mathbf{L}_\uparrow,\mathbf{L}_\downarrow$, by writing
\begin{align}\label{eq:LDS}
\begin{split} 
 \Dot{\Theta} = \Omega &- \sigma^\downarrow\boldsymbol{d}\sin\left(\boldsymbol{\delta}\Theta-z\gamma \mathbf{L}_\uparrow\Theta\right) \\
 &-\sigma^\uparrow\boldsymbol{\delta}\sin\left(\boldsymbol{d}\Theta-z\gamma \mathbf{L}_\downarrow\Theta\right)\, ,
 \end{split}
\end{align}
where $z>0$ regulates the strength of the local cross-order coupling and $\gamma$ is the block-diagonal matrix~\cite{wee2023persistent}
\begin{align}
\gamma = \mathrm{diag}(I_{n_0},-I_{n_1},\dots,(-1)^K I_{n_K})\, ,
\end{align}
which anticommutes with both $\boldsymbol{d}$ and $\boldsymbol{\delta}$. 
This choice of $\gamma$ comes from the fact that for the linearized dynamics (with unit coupling strengths)
$\Dot{\Theta} = \Omega - (\mathbf{D}^2 + \gamma\mathbf{D}^3)\Theta$, 
the matrix $-(\mathbf{D}^2 + \gamma\mathbf{D}^3)$ can be shown to have complex eigenvalues with non-positive real part, resulting in the emergence of damped oscillations (see~\cite[Appendix A]{calmon2023local}), as depicted in \cref{fig:lds}.
\Cref{eq:LDS}, named \emph{local Dirac synchronization}~\cite{calmon2023local}, displays explosive synchronization transitions and stable hysteresis loops.
As an example, for $K=2$, we can write \cref{eq:LDS} explicitly
\begin{align}\label{eq:dirac_nlt}
\begin{cases}
\Dot{\theta}_{(0)} = \omega_{(0)} - \sigma^\uparrow B^1\sin(D^0\theta_{(0)}+L^1_\downarrow\theta_{(1)})\\[6pt]
\begin{aligned}
\Dot{\theta}_{(1)} = \omega_{(1)} - \sigma^\downarrow D^0\sin(B^1\theta_{(1)}-L^0_\uparrow\theta_{(0)})  \\[4pt]
-\sigma^\uparrow B^2\sin(D^1\theta_{(1)}-L^2_\downarrow \theta_{(2)})
\end{aligned}\\[10pt]
\Dot{\theta}_{(2)} = \omega_{(2)} - \sigma^\downarrow D^1\sin(B^2\theta_{(2)} +L^1_\uparrow\theta_{(1)})
\end{cases}\, .
\end{align}
From a more general point of view, if we now define a frustration $A_\updownarrow=(\alpha^\updownarrow_{(0)}, \alpha^\updownarrow_{(1)}, \ldots, \alpha^\updownarrow_{(K)})$ (possibly dependent on $\Theta$) we can generalize the single-order frustrated simplicial Kuramoto model (\cref{eq:sakaguchi_kuramoto_dependent}) as
\begin{align}
\Dot{\Theta} &= \Omega - \boldsymbol{d}\sin(\boldsymbol{\delta}\Theta + A_\uparrow) - \boldsymbol{\delta}\sin(\boldsymbol{d}\Theta + A_\downarrow)\, .
\label{eq:dirac_gen_non_OI}
\end{align}
In addition, if we introduce the \emph{total} lift operators $\mathbf{V} = \mathrm{diag}(V^0, V^1, \ldots V^K)$, $\mathbf{U} = \mathrm{diag}(U^0, U^1, \ldots U^K)$, we have an orientation independent version, akin to \cref{eq:sakaguchi_kuramoto_independent} but in the Dirac framework as
\begin{align}
\Dot{\Theta} = \Omega &- \left(\boldsymbol{d}\mathbf{V}^\top\right)^-\sin(\mathbf{V}\boldsymbol{\delta}\Theta + \mathbf{U}A_\uparrow) \nonumber\\
 &\, - \left(\boldsymbol{\delta}\mathbf{V}^\top\right)^-\sin(\mathbf{V}\boldsymbol{d}\Theta + \mathbf{U}A_\downarrow)\, .
 \label{eq:dirac_gen_OI}
\end{align}
In contrast to \cref{eq:dirac_gen_non_OI}, being orientation independent, this system couples both the simplicial orders and the Hodge subspaces.

It follows that the local Dirac synchronization dynamics of \cref{eq:LDS} is the application of the simple $\Theta$-dependent frustrations $A_\uparrow = -z\gamma \mathbf{L}_\uparrow\Theta,\, A_\downarrow = -z\gamma\mathbf{L}_\downarrow\Theta$. This fact, together with the gradient flow formulation of the Dirac model, gives us a common framework to build and study multiple variants of the model. It would be natural, for example, to consider an analogous model where the local frustration is introduced in an orientation-independent fashion.
Finally, we did not consider here possible extensions of the results on equilibrium solutions but left it for future works.
It should be possible to extend some of the theorems of this work due to the similar structure between a single simplicial Kuramoto model and the Dirac-based formulation.

\section{Application to functional connectivity reconstruction} \label{section:application}

Oscillator models have been extensively used in neuroscience as they offer a powerful and flexible framework for studying simplified versions of the dynamics of neuronal or brain networks~\cite{cumin2007generalising,breakspear2010generative,schmidt2015kuramoto,pope2021modular,pope2023co}. 
By treating neurons or brain regions as oscillators that interact with each other, these models can capture significant features of brain activity observed in experiments, such as the presence of rhythms and oscillations. 
While oscillator models have been widely used to study brain dynamics, it is important to note that most of these have focused on pairwise interactions between neurons or brain regions. 
This is due in part to the fact that pairwise interactions are simpler to model or analyze and that anatomically it is more realistic to consider the dynamics taking place on networks rather than higher-order systems. 
However, recent studies have suggested that higher-order interactions may also play an important role in brain dynamics, both functionally~\cite{schneidman2006weak,yu2011higher} and structurally~\cite{sun2023dynamic,ghorbanchian2021higher}. 
These interactions involve three or more elements and can give rise to emergent phenomena that cannot be explained by pairwise interactions alone.
Given the potential importance of these higher-order interactions, it is natural to apply models of higher-order synchronization to brain data. 
These models might offer a more comprehensive framework for studying the dynamics of large-scale brain networks and have the potential to give us new insights into the mechanisms underlying cognition and behavior.

To test this hypothesis, we study how well simplicial Kuramoto models of various orders could reproduce brain correlation patterns. 
Following the methodology proposed in~\cite{pope2021modular,pope2023co}, we run simulations of $5$ different variants of the simplicial Kuramoto model on a real structural connectome, the network that describes the connectivity structures between regions of the human cerebral cortex, and we investigate how well each model can reproduce the resting-state functional activity experimentally measured. 

\begin{figure*}[htp]
    \centering
    \includegraphics[width=\linewidth]{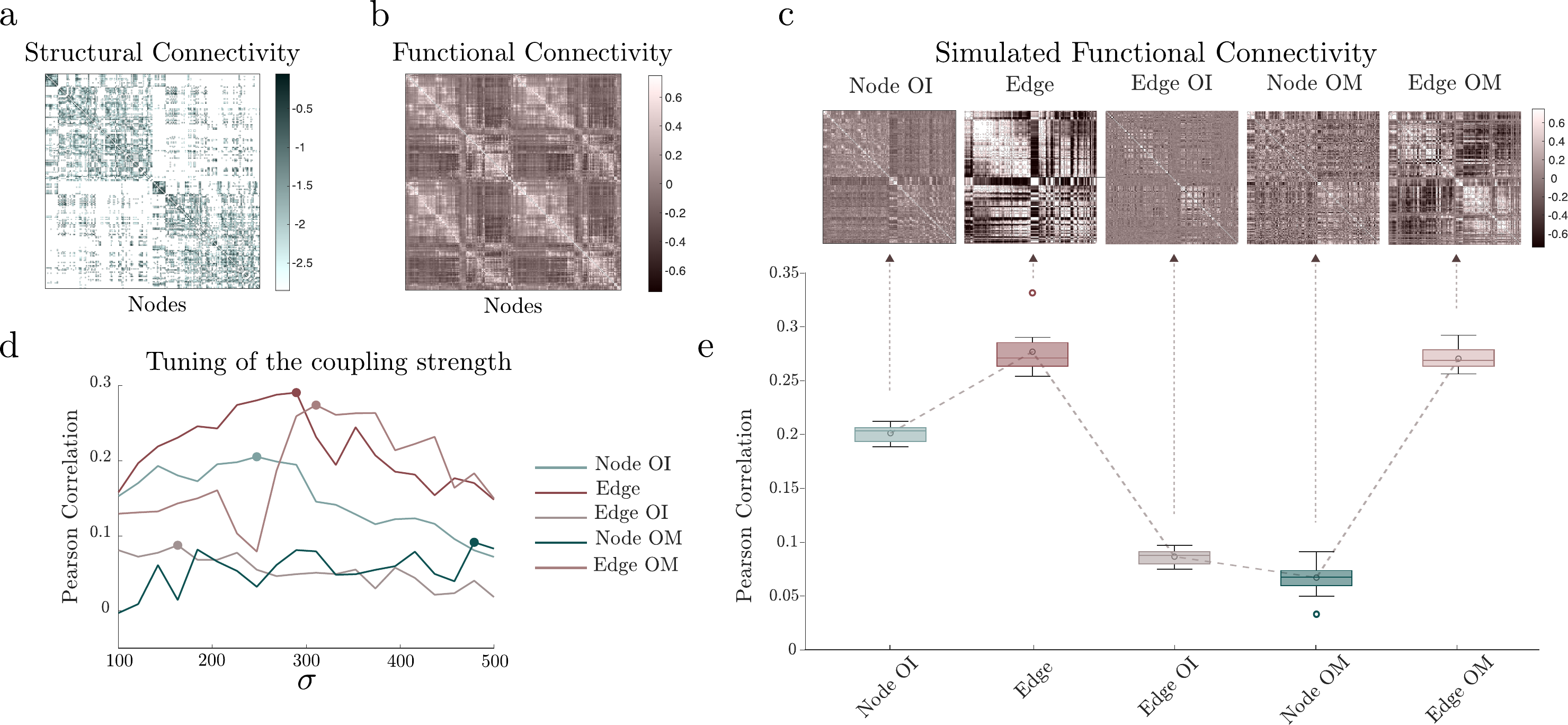}
    \caption{
    {\bf a.} Structural connectivity matrix representing the weighted network onto which we simulate the dynamics. The color represents the logarithm of the weight. 
    {\bf b.} The empirical functional connectivity matrix.
    {\bf c.} We simulate $5$ different variants of the simplicial Kuramoto model and compute the correlation matrices of their post-processed trajectories as simulated FC matrices.  
    {\bf d.} The coupling strength $\sigma$ is tuned for each model by scanning $20$ values between $100$ and $500$. 
    {\bf e.} Pearson correlations between the empirical FC and the simulated FCs for the $6$ models, over $10$ simulations. 
    }
    \label{fig:application}
\end{figure*}

The structural connectome is encoded in a group-averaged weighted structural connectivity matrix (\cref{fig:application}a), obtained by diffusion imaging and tractography, by parcellating the brain into $N=200$ regions, which here take the role of nodes connected by $M=6040$ weighted edges. From the network adjacency matrix, we derive the incidence matrix $B_1\in\sset{-1,1}^{N\times M}$ by choosing randomly edges orientations.
To achieve consistency with~\cite{pope2021modular}, the connection weights $K_1,\dots,K_{M}$ are included, after being inverted, as weights on the edges $W_1 = \mathrm{diag}(\frac{1}{K_1},\dots,\frac{1}{K_{M}})$ and the tract lengths are encoded in an edge frustration vector $\alpha\in\R^{M}$. The natural frequencies for both node-based and edge-based models are sampled independently from a Gaussian distribution with a mean of $2\pi\, 40$ and a standard deviation of $2\pi\, 0.1$.
We compare the following five models:
\begin{enumerate}
    \item \emph{Orientation independent node Kuramoto-Sakaguchi model (Node OI)}. This is, by construction, the classical Kuramoto-Sakaguchi model of \cref{eq:kuramoto_sakaguchi_nodes}, used in \cite{pope2021modular}
    \begin{align}
    \Dot{\theta}_i = \omega_i - \sigma\sum_{j=1}^{200} K_{ij}\sin(\theta_i - \theta_j + \alpha_{ij})\, .
    \end{align}
    \item \emph{Edge Simplicial Kuramoto (Edge)}. The simplest possible simplicial Kuramoto model defined on the edges
    \begin{align}
    \Dot{\theta} = \omega - \sigma B_1^\top\sin(B_1 W_1^{-1}\theta)\, .
    \end{align}
    \item The \emph{Orientation Independent Edge Sakaguchi-Kuramoto (Edge OI)}
    \begin{align}
    \Dot{\theta} = \omega - \sigma \left(B_1^\top (V^0)^\top\right)^- \sin\left(V^0 B_1 W_1^{-1}\theta - U^0 B_1 W_1^{-1}\alpha\right)\, .
    \end{align}
 \end{enumerate}
The explosive simplicial Kuramoto model (\cref{eq:adaptively_coupled}) cannot be directly used as it requires nodes, edges, and triangles for its interaction terms to be nonzero. Triangles are not present in the structural connectivity network and thus, to avoid injecting arbitrary structure into the analysis, we will not use it. As a proxy for its behavior, however, we propose the similar \emph{order-modulated model} (OM), derived by multiplying $\sigma$ by the order parameter. In other words, the OM model is the gradient flow of the square order parameter. 
\begin{align}   
\Dot{\theta} = \omega + \frac{1}{2} C_kW_k\nabla_{\theta} R^2_k(\theta)\, .
\end{align}
We simulate two different OM models.
\begin{enumerate}
    \setcounter{enumi}{4}
    \item The \emph{Order-modulated node Kuramoto-Sakaguchi (Node OM)} 
    \begin{align}
    \Dot{\theta} = \omega - \sigma R_0(\theta) (B_1W_1^{-1} (V^1)^\top)^- \sin(V^1B_1^\top\theta - U^1\alpha)\, .
    \end{align}
    \item The \emph{Order-modulated edge simplicial Kuramoto (Edge OM)}
    \begin{align}
    \Dot{\theta} = \omega - \sigma R_1(\theta) B_1^\top \sin(B_1 W_{1}^{-1}\theta)\, .
    \end{align}
\end{enumerate}

Models 2, 3, and 5 are defined on the edges of the network. 
Given that we want to simulate a node-wise functional connectivity matrix, we consider the projections of their phases onto the nodes $\theta^{(-)}$ to get node-wise trajectories. For this reason, notice that it is not necessary to numerically solve all the $M$ equations on the edges, but it is enough to directly integrate the projected dynamics.

\subsection{Simulations} 
Following~\cite{pope2021modular}, the simulations are run for a total of $T=812$ seconds with a time resolution of $\delta t=1$ms (using MATLAB ode45), and the first $20$ seconds are discarded to allow the dynamics to reach stationarity. 
We then take the  trajectories, convert them into downsampled BOLD signals, filter them with a lowpass cutoff of $c=0.25$Hz, and use them to compute $N\times N$ pairwise Pearson correlation matrices. 
These simulated functional connectivity matrices (\cref{fig:application}b) are then compared to the experimental resting-state functional connectivity (FC) matrix (\cref{fig:application}c) using Pearson correlation (by correlating the vectorized upper triangular matrix). 
We repeat this process multiple times for each model by varying the coupling strength in order to tune it. 
We scan $20$ $\sigma$ values ranging from $100$ to $500$, and select the optimal one w.r.t Pearson correlation (\cref{fig:application}e). 
Given the optimal coupling strength for each model, we then perform $10$ simulations for each one of them with different random starting phases and natural frequencies and confront them with the empirical FC matrix. 
The results are shown in \cref{fig:application}d, where it is easy to see how the two non-frustrated edge-based models outperform the node ones, achieving an average of $r=0.27$ correlation against the $r=0.2$ of the standard node Sakaguchi-Kuramoto.
The result is statistically confirmed by an ANOVA test which achieves p-values lower than $10^{-3}$. The effect size against the node Kuramoto model is $0.0757$ for the edge model and $0.0692$ for the edge OM.

Our findings suggest that an edge-based description of the dynamics might provide a better fit to the experimental data, both outperforming the node-based models and without resorting to additional parameters or internal mechanisms, as for example edge flickering~\cite{pope2023co} (which was shown to obtain a slightly lower correlation than our edge Kuramoto model).  
In fact, arguably edge-based simplicial Kuramoto models might provide a better fit to the observed FC correlation structure exactly because the variables are defined on the connections that link different nodes together, rather than on the nodes themselves.
That is, the observed activity of brain regions might be better explained as the result of the information integration taking place via the structural fibers linking the regions, rather than by looking at the brain regions in themselves~\cite{gidon2020dendritic}, and display interesting parallels with neural frequency mixing behaviours~\cite{haufler2019detection,luff2023neuron}. 
Naturally, these results are preliminary and intended as a simple demonstration of the potential of simplicial (and more generally, higher-order) oscillator models in the context of computational neurobiological models. 

\begin{table*}[htpb]
\begin{tabular}{lcccccccc}
    \toprule
    & \multicolumn{2}{c}{Standard} & \multicolumn{2}{c}{Frustrated} & \multicolumn{2}{c}{OI Frustrated}& \multicolumn{2}{c}{Explosive}  \\ \cmidrule(lr){2-3}\cmidrule(lr){4-5}\cmidrule(lr){6-7}\cmidrule(lr){8-9}
    & Single order & Dirac & Single order & Dirac & Single order & Dirac & Single order & Dirac \\ \midrule
    Equation & \eqref{eq:simplicial_kuramoto} & \eqref{eq:dirac_synch} & \eqref{eq:sakaguchi_kuramoto_dependent}  & \eqref{eq:dirac_gen_non_OI}/\eqref{eq:LDS} & \eqref{eq:sakaguchi_kuramoto_independent} & \eqref{eq:dirac_gen_OI} & \eqref{eq:adaptively_coupled} &(\ref{eq:nonlinear-super-gradient-flow},\ref{eq:local-nonlinear-super-gradient-flow}) \\
    Hodge coupling & no & no & no & no/yes & yes & yes & yes & yes \\
    Order Coupling & - & no & - & no/yes & - & no & - & yes \\
    \bottomrule
\end{tabular}
\caption{Taxonomy of the simplicial Kuramoto models presented in this work. OI stands for orientation independent.}
\label{tab:taxonomy}
\end{table*}

\section{Summary and Outlook}

Simplicial Kuramoto models, where oscillators are defined on simplices rather than on nodes, have grown in numbers, yielding a wide variety of different and interesting dynamics. 
Here, we have attempted to provide a more unified view, akin to a taxonomy, of this simplicial Kuramoto zoo. 
Our description has relied heavily on topology and discrete differential geometry because the simplicial structure of these models naturally lends itself to a topological and geometrical language, including boundary operators and the Hodge Laplacian. 

We have shown that these models can be divided into three main categories:
\begin{itemize}
    \item ``simple'' models, which can all be rewritten in a single framework: that of gradient flows, encoded in \cref{eq:single-order-gradient-flow} for a single order, and \cref{eq:super-gradient-flow} for all orders. These models do not have couplings across orders or frustration.
    \item ``Hodge-coupled'' models, in which different Hodge components of the dynamics are coupled. 
    These include explosive models, which can be rewritten in a similar gradient flow framework (\cref{eq:explosive-gradient-flow}) as the simple models, but also include other models that require additional ingredients, in our case two flavors of frustration (\cref{eq:sakaguchi_kuramoto_independent}). 
    \item ``order-coupled'' (Dirac) models, in which oscillators are coupled across different orders with or without frustrations (\cref{eq:LDS}).
\end{itemize}
This unified view in terms of just two ingredients---gradient flows and frustrations---compresses this model taxonomy to a lower-dimensional space of models and has allowed us to describe the general properties of these models. 

A first example is the possibility to derive a set of bounds on the value of the coupling strength that are necessary or sufficient to obtain synchronization for ``simple'' models, thanks to their simplicity.
This gave us a general description of the space of equilibria and their relative degree of reachability as a function of the coupling strength. 
Additionally, using this taxonomy, it is possible to investigate when two models are genuinely different or not: we demonstrated that the simple simplicial Kuramoto model is strictly equivalent to the standard Kuramoto model on (pairwise) networks if the underlying simplicial complex structure is manifold-like (\Cref{thm:simplicial-kuramoto-manifold}). 
More specifically, by mapping the oscillators defined on simplices of order $k$ to nodes in an effective (pairwise) network, the effective dynamics reduce to a standard Kuramoto model. 
This is a powerful result that bridges the simplicial models and the well-known standard Kuramoto model and shows that the simplicial models are of most interest on non-manifold-like simplicial complexes.

More generally, we showed that the simplicial models can be related to another important class of higher-order Kuramoto models: those where oscillators are defined only on nodes and interact across hypergraphs~\cite{skardal2019abrupt,lucas2020multiorder}. 
Indeed, a simplicial model on a generic complex can be rewritten as a node-Kuramoto model with group interactions occurring on an effective dual hypergraph. 
There is one important difference, however: the models obtained this way do not have the properties usually desired for models defined on nodes, that is, the coupling functions do not vanish when all phases are equal.  
As a consequence, contrary to these other models, the standard 1-cluster synchronization solution is not guaranteed to exist and the equations are not invariant under a uniform phase shift. 
Although an equivalent effective hypergraph can be found to define oscillators on nodes, the simplicial models are more naturally described in the formalism of discrete differential geometry by the boundary operators of the simplicial complex. 
It is an interesting future direction to investigate to what degree and exactly under what conditions these two classes of models can be related to each other. 
Furthermore, the formalism and results presented here, of course, refer to the case of synchronization, but we expect them to be rather straight-forwardly generalizable to more general dynamics, such as consensus~\cite{neuhauser2021multibody,neuhauser2021consensus} or diffusion~\cite{SchaubSirev,schaub2018flow}, and structures, such as cell complexes~\cite{carletti2023global}. 

With respect to applications, we provided a simple example of application to the reconstruction of brain functional connectivity from a structural connectome, a common and still open task in computational neuroscience~\cite{einevoll2019scientific,deco2011emerging}, finding that vanilla models of simplicial edge Kuramoto models are competitive or even outperform more complex node-based models~\cite{pope2023co}. 
We suspect that this might be related to the fact that, when edge phases are projected down to node dynamics, they behave akin to time-evolving temporal delays across node signals, an element that has been recognized as crucial in brain dynamical simulations~\cite{petkoski2019transmission}.
Similar considerations however are relevant also for many other types of real-world complex systems, such as network traffic~\cite{petri2013entangled, lo2001dynamic, levin2016paradoxes} and power grid balancing~\cite{taher2019enhancing,he2009design}. 

Overall, we believe that the proposed framework provides a starting point to shed new light and further research on a number of interlaced theoretical and practical topics across the broader community of complex dynamical systems. 

\subsection*{Acknowledgement}
A.A. was supported by funding to the Blue Brain Project, a research center of the École polytechnique fédérale de Lausanne (EPFL), from the Swiss government’s ETH Board of the Swiss Federal Institutes of Technology. R.P. acknowledges funding from the Deutsche Forschungsgemeinschaft (DFG, German Research Foundation) Project-ID 424778381-TRR 295. \\

\subsection*{Code availability}
An open-source code to numerically solve the presented models is available at \url{https://github.com/arnaudon/simplicial-kuramoto}.

\bibliography{references}

\clearpage

\onecolumngrid

\appendix
\input{SI}

\end{document}

%% file: SI.tex
\section{Kuramoto model expressed with the boundary matrices}\label{section:kuramoto_with_boundary}
We prove here how the Kuramoto model on a graph with $N$ nodes and set of edges $\mathcal{E}$,
$$
\Dot{\theta}_i = \omega_i - \sigma\sum_{j=1}^N A_{ij}\sin(\theta_i-\theta_j) ,
$$
can be rewritten using the boundary matrices as
$$
\Dot{\theta} = \omega - \sigma B_1\sin(B_1^\top\theta) .
$$
First, compute the action of $B_1^\top$ on the phases vector. For any edge $\epsilon$
$$
(B_1^\top\theta)_{\epsilon}= \sum_{i=1}^N (B_1^\top)_{\epsilon i}\theta_i = \sum_{i=1}^N (B_1)_{i\epsilon}\theta_i = \theta_{\mathbf{h}(\epsilon)} - \theta_{\mathbf{t}(\epsilon)},
$$
where $\mathbf{h}(\epsilon),\mathbf{t}(\epsilon)$ give respectively the head and tail node of edge $\epsilon$. It follows that, for any node $i$,
\begin{align*}
\left[B_1\sin(B_1^\top\theta)\right]_i &= \sum_{\epsilon\in\mathcal{E}} (B_1)_{i\epsilon} \sin(B_1^\top\theta)_\epsilon = \sum_{\epsilon\in\mathcal{E}} (B_1)_{i\epsilon}\sin(\theta_{\mathbf{h}(\epsilon)} - \theta_{\mathbf{t}(\epsilon)}) \\
&= \sum_{\epsilon:\mathbf{h}(\epsilon) = i} \sin(\theta_i - \theta_{\mathbf{t}(\epsilon)}) - \sum_{\epsilon:\mathbf{t}(\epsilon) = i}\sin(\theta_{\mathbf{h}(\epsilon)}-\theta_i)\\
&= \sum_{\epsilon:\mathbf{h}(\epsilon) = i} \sin(\theta_i - \theta_{\mathbf{t}(\epsilon)}) + \sum_{\epsilon:\mathbf{t}(\epsilon) = i}\sin(\theta_i-\theta_{\mathbf{h}(\epsilon)}) \\
&= \sum_{j\in\mathcal{N}(i)} \sin(\theta_i-\theta_j) = \sum_{j=1}^n A_{ij}\sin(\theta_i-\theta_j),
\end{align*}
where $\mathcal{N}(i)$ is the neighborhood of node $i$.

\section{Proof of Theorem \ref{theorem:sigma_fp}}\label{appendix:proof}
For simplicity, we prove the result only for the $(-)$ projection. The $(+)$ case can be easily recovered by replacing $L^{k-1}_{\uparrow}$ with $L^{k+1}_{\downarrow}$, $\omega^{(-)}$ with $\omega^{(+)}$ and $\sigma^\downarrow$ with $\sigma^\uparrow$.\\
The idea of the proof inspired by \cite{Stability_kuramoto_model} is to find the dynamics of the coefficients of $\theta^{(-)}$ w.r.t to a basis of the subspace $\Ima L^{k-1}_{\uparrow}$, rewrite its equilibrium equation as a fixed-point equation and then find conditions to apply Brouwer's fixed-point theorem.

First, as $L^{k-1}_{\uparrow}$ is the matrix representation of a self-adjoint, positive semidefinite operator, we can consider its eigendecomposition
\begin{align}
    L^{k-1}_{\uparrow} = V\Lambda V^*\, ,
\end{align}
where $V$ is a unitary matrix ($V^*V = V^*V = I_{n_{k-1}}$) and $\Lambda$ is diagonal with non-negative elements. 
Recall that $V^* = W_{k-1} V W^{-1}_{k-1}$, which will later ensure that the inner product on the eigenspace is compatible with the original inner product from $W_{k-1}$.
Let us assume that the zero eigenvalues of $\Lambda$ are the last ones in the diagonal, so that
\begin{align*}
    \Lambda = \mathrm{diag}(\lambda_1,  \dots, \lambda_r, 0, \dots,  0)\, , 
\end{align*}
where $r = \rank(L^{k-1}_{\uparrow})$. 

The columns of $V$ provide a basis of $\R^{n_{k-1}}$. We want however to restrict ourselves to the subspace $\Ima L^{k-1}_{\uparrow}$. To do that, we drop the columns associated with zero eigenvalues (which span $\ker L^{k-1}_{\uparrow}$) and consider the compact eigendecomposition
\begin{align}
    L^{k-1}_{\uparrow} = \widetilde{V}\widetilde{\Lambda}\widetilde{V^*}\, ,
\end{align}
$\widetilde{V}$ consists of the first  $r$ columns of $V$, $\widetilde{V^*}$ is made by the first $r$ rows of $V^*$ and $\widetilde{\Lambda} = \diagm(\lambda_1,\dots,\lambda_r)$. The columns of $\widetilde{V}$ are a basis of the reachable subspace $\Ima L_{\uparrow}^{k-1}$. 

In order to carry out the proof, we need to show that $\widetilde{V^*}$ is the adjoint matrix to $\widetilde{V}$ with respect to a particular choice of ``natural'' inner product on the space of coefficients $\R^r$. To do that we need some preliminary definitions and results. We first define the \textit{truncation matrix}
\begin{equation*}
\begin{cases}
I_{a,b} = \begin{pmatrix} I_b \\ 0_{a-b,b}\end{pmatrix} \text{ if } a>b\\[3ex]
I_{a,b} = I_{b,a}^\top \text{ if } a<b\\[1ex]
I_{a,b} = I_a \text{ if } a=b
\end{cases}
\end{equation*}
which, truncates the columns or rows of a matrix when multiplied respectively on the right or left. It follows that
\begin{align}
    \widetilde{V} = VI_{n_{k-1},r}\qquad \mathrm{and} \qquad \ \widetilde{V^*} = I_{r,n_{k-1}}V^*\, ,
\end{align}
from which we have
\begin{Lemma}
$\widetilde{V^*}\widetilde{V} = I_r$.
\end{Lemma}
\begin{proof}
\begin{align*}
\widetilde{V^*}\widetilde{V} = I_{r,n_{k-1}} V^* V I_{n_{k-1},r} = I_{r,n_{k-1}}I_{n_{k-1},r} = I_r
\end{align*}
as $r<n_{k-1}$.
\end{proof}
Moreover, one can see that, if $A\in\R^{n_{k-1},n_{k-1}}$ is diagonal, then
\begin{equation}\label{eq:diagonal}  
I_{r,n_{k-1}}A = \widetilde{A} I_{r,n_{k-1}},
\end{equation}
where 
$$
\widetilde{A} = \diagm(a_1,\dots,a_r).
$$
It is now simple to prove that the truncation of the inverse weight matrix $\widetilde{W}_{k-1}^{-1} \defeq \diagm\left(\frac{1}{w^{k-1}_1},\dots,\frac{1}{w^{k-1}_r}\right)$ is the natural inner product of the coefficients space.
\begin{Lemma}
$\widetilde{V^*}$ is the adjoint matrix to $\widetilde{V}$ w.r.t to the inner product $\widetilde{W}^{-1}_{k-1}$ i.e. $\widetilde{V^*} = \widetilde{V}^*$.
\end{Lemma}
\begin{proof}
$$
\widetilde{V^*} = I_{r,n_{k-1}}V^* =  I_{r,n_{k-1}}W_{k-1}V^\top W^{-1}_{k-1} = \widetilde{W}_{k-1} I_{r,n_{k-1}}V^\top W^{-1}_{k-1} = \widetilde{W}_{k-1}\widetilde{V}^\top W^{-1}_{k-1} = \widetilde{V}^* .
$$
because of Eq.~\eqref{eq:diagonal}. 
\end{proof}
With a slight abuse of notation, in the following we will denote the norm on the coefficient space with $\norm{c}_{w^{k-1}}$, keeping in mind the fact that
\begin{equation}
\norm{c}_{w^{k-1}}^2 = \inner{c}{c}_{w^{k-1}} = \inner{\widetilde{V}^*\widetilde{V}c}{c}_{w^{k-1}} = \inner{\widetilde{V}c}{\widetilde{V}c}_{w^{k-1}} = \norm{\widetilde{V}c}_{w^{k-1}}^2.
\end{equation}

We can now rewrite the simplicial Kuramoto dynamics of the $(-)$ projection in the basis $\widetilde{V}$, $\theta^{(-)}=\widetilde{V}c$:
\begin{align}
    \dfrac{d}{dt}\widetilde{V}c = \omega^{(-)} - \sigma^\downarrow L^{k-1}_{\uparrow}\sin(\widetilde{V}c) = \omega^{(-)} - \sigma^\downarrow \widetilde{V}\widetilde{\Lambda}\widetilde{V}^*\sin(\widetilde{V}c)\, .
\end{align}
With this formulation, we are naturally restricting $\theta^{(-)}$ to lie in the reachable subspace. 
We find the dynamics of the coefficients $c$ by left multiplying by $\widetilde{V}^*$ and using $\widetilde{V}^*\widetilde{V} = I$
\begin{equation}\label{eq:reachable_dynamics}
\Dot{c} = \widetilde{V}^*\omega^{(-)} - \sigma^\downarrow \widetilde{\Lambda} \widetilde{V}^* \sin(\widetilde{V}c),
\end{equation}

The coefficients $c$ are associated to a reachable equilibrium configuration if and only if $\Dot{c}=0$, i.e.
\begin{align}
\widetilde{V}^*\frac{\omega^{(-)}}{\sigma^\downarrow} = \widetilde{\Lambda} \widetilde{V}^* \sin(\widetilde{V}c)\, .
\end{align}
We want to reduce this equation to a fixed point equation, of the form $f(c) = c$ for some function $f$. First, we write
\begin{align*}
\widetilde{V}^*\frac{\omega^{(-)}}{\sigma^\downarrow} &= \widetilde{\Lambda} \widetilde{V}^* \sin(\widetilde{V}c) \iff \widetilde{\Lambda}^{-1} \widetilde{V}^*\frac{\omega^{(-)}}{\sigma^\downarrow} = \widetilde{V}^* S(c)\widetilde{V}c
\end{align*}
where we defined $S(c) \defeq \diagm(\sinc(\widetilde{V}c))$ with $\sinc(x) = \sin(x)/x$. 
We then have the fixed point equation
\begin{equation}\label{eq::fixed_point_equation}
c = (\widetilde{V}^*S(c)\widetilde{V})^{-1} \widetilde{\Lambda}^{-1}\widetilde{V}^*\frac{\omega^{(-)}}{\sigma^\downarrow} \defeq f(c)\, ,
\end{equation}
which make sense only if  the matrix $\widetilde{V}^*S(c)\widetilde{V}$ is invertible.

\begin{Lemma}[Invertibility of $\widetilde{V}^*S(c)\widetilde{V}$]
If $S(c)$ has strictly positive elements, then $\widetilde{V}^*S(c)\widetilde{V}$ is invertible. 
\end{Lemma}
\begin{proof}
First, notice that $S$ is diagonal which, together with the inner product matrix being diagonal, means that $S(c)$ is a Hermitian matrix ($S(c)^* = S(c)$), and so is its square root. We get the following
$$
\widetilde{V}^*S(c)\widetilde{V} = (S^{\frac{1}{2}}(c)\widetilde{V})^*(S^{\frac{1}{2}}(c)\widetilde{V}) \defeq A^*A.
$$
Given that $\ker A^* = (\Ima A)^\perp$, we deduce that $A^*A$ is invertible if and only if $A = S^{\frac{1}{2}}(c)\widetilde{V}$ has trivial kernel. Moreover, we know that the columns of $\widetilde{V}$ are a basis and thus $\widetilde{V}c = 0 \iff c = 0$. If $S^{\frac{1}{2}}(c)$ is invertible, then, its kernel will be trivial and, by extension, the same will hold for $A$. 
\end{proof}
This result on the invertibility of $\widetilde{V}^*S(c)\widetilde{V}$ hence translates to a condition on $S(c)$.
\begin{Lemma}[Positive definiteness of $S$]\label{lemma:S_PD}
For any $\gamma\in (0,\pi/2)$, if the coefficients $c$ are such that
\begin{align}
 \norm{\widetilde{V}c}_{w^{k-1}} \leq \frac{\gamma}{\sqrt{\max_i(w^{k- 1}_i)}}\, ,
\end{align}
then $S(c)$ has positive diagonal elements.
\end{Lemma}
\begin{proof}
Under the hypothesis of the lemma it holds that
\begin{align}
\norm{\widetilde{V}c}_\infty \leq \norm{\widetilde{V}c}_2 = \sqrt{\sum_i (\widetilde{V}c)_i^2} = \sqrt{\sum_i w^{k- 1}_i \frac{1}{w^{k- 1}_i} (\widetilde{V}c)^2_i} \leq \sqrt{\max_i(w^{k- 1}_i)}\norm{\widetilde{V}c}_{w^{k-1}} \leq \gamma,
\end{align}
meaning that every component of $\widetilde{V}c$ will belong to the interval $[-\gamma,\gamma]$. The $\sinc$ function, which is applied component-wise to $\widetilde{V}c$, is strictly positive in $[-\gamma,\gamma]$ when $\gamma\in (0,\pi/2)$, hence the positive definiteness of $S(c) = \diagm(\sinc(\widetilde{V}c)))$.
\end{proof}

We now want to prove that the left-hand side of the equilibrium fixed point \cref{eq::fixed_point_equation} is a continuous map from the set 
\begin{align}
   \mathcal{B} = \sset{c:\norm{c}_{w^{k-1}} = \norm{\widetilde{V}c}_{w^{k-1}}\leq\frac{\gamma}{\sqrt{\max_i(w^{k- 1}_i)}}}
\end{align}
to itself.
First, one has the following inequality.
\begin{align}\label{eq::f(c)_bound}
\norm{f(c)}_{w^{k-1}} = \norm{(\widetilde{V}^* S(c) \widetilde{V})^{-1}\widetilde{\Lambda}^{-1}\widetilde{V}^*\frac{\omega^{(-)}}{\sigma^\downarrow}}_{w^{k-1}} \leq \frac{1}{\sigma^\downarrow}\norm{(\widetilde{V}^* S(c) \widetilde{V})^{-1}}_{w^{k-1}}\norm{\widetilde{\Lambda}^{-1}\widetilde{V}^*\omega^{(-)}}_{w^{k-1}},
\end{align}
where the first term is the matrix norm induced by the $w^{k-1}$ vector norm. Let us look at the two terms of \cref{eq::f(c)_bound} separately, starting from the right one.
\begin{Lemma}\label{lemma:right}
If $c\in\mathcal{B}$ then
\begin{align}
\norm{\widetilde{\Lambda}^{-1}\widetilde{V}^*\omega^{(-)}}_{w^{k-1}} = \norm{\beta^{(-)}}_{w^{k-1}}.
\end{align}
\end{Lemma}
\begin{proof}
\begin{align*}
\norm{\widetilde{\Lambda}^{-1}\widetilde{V}^*\omega^{(-)}}_{w^{k-1}} = \norm{\widetilde{V}^*\widetilde{V}\widetilde{\Lambda}^{-1}\widetilde{V}^*\omega^{(-)}}_{w^{k-1}} = \norm{\widetilde{V}^* (L^{k-1}_{\uparrow})^\dagger\omega^{(-)}}_{w^{k-1}} = \norm{\widetilde{V}^*\beta^{(-)}}_{w^{k-1}}.
\end{align*}
Moreover, by definition of the $(k-1)$ norm,
\begin{align*}
\norm{\widetilde{V}^*\beta^{(-)}}_{w^{k-1}}^2 = \inner{\widetilde{V}^*\beta^{(-)}}{\widetilde{V}^*\beta^{(-)}}_{w^{k-1}} = \inner{\widetilde{V}\widetilde{V}^*\beta^{(-)}}{\beta^{(-)}}_{w^{k-1}}\, ,
\end{align*}
but, as $\widetilde{V}\widetilde{V}^*$ is the orthogonal projection operator onto $Im(L^{k-1}_{\uparrow})$ and $\beta^{(-)}\in Im(L^{k-1}_{\uparrow})$, $\widetilde{V}\widetilde{V}^*\beta^{(-)} =\beta^{(-)}$, and we have the result.
\end{proof}
Let us now analyze the first term of Eq.~\cref{eq::f(c)_bound} and bound it from above.
\begin{Lemma}\label{lemma:left}
If $c\in\mathcal{B}$ then
\begin{align}
\norm{(\widetilde{V}^*S(c)\widetilde{V})^{-1}}_{w^{k-1}} \leq \frac{\gamma}{\sin(\gamma)} = \sinc^{-1}(\gamma).
\end{align}
\end{Lemma}

\begin{proof}
When $c\in\mathcal{B}$ then Lemma~\ref{lemma:S_PD} tells us that $S(c)$ is positive definite and therefore we can write
\begin{align*}
\norm{(\widetilde{V}^*S(c)\widetilde{V})^{-1}}_{w^{k-1}} = \norm{(A^*A)^{-1}}_{w^{k-1}},
\end{align*}
with $A = S^{\frac{1}{2}}(c)\widetilde{V}$,
for which it holds that 
\begin{align}\label{eq::proof_1}
\frac{1}{\norm{(A^*A)^{-1}}_{w^{k-1}}} = \min_{\norm{v}_{w^{k-1}}= 1} \norm{A^*A v}_{w^{k-1}}. 
\end{align}
We apply here the Cauchy-Schwarz inequality
$$
\abs{\inner{A^*Ac}{v}} \leq \norm{A^*Av}\norm{v}
$$
and find that the right-hand side of \cref{eq::proof_1} can be bounded from below

\begin{align*}
\min_{\norm{v}_{w^{k-1}}= 1} \norm{A^*Av}_{w^{k-1}} &\geq \min_{\norm{v}_{w^{k-1}}= 1} \abs{\inner{A^*Av}{v}_{w^{k-1}}} = \min_{\norm{v}_{w^{k-1}}= 1} \norm{Av}^2_{w^{k-1}} \\
&= \left(\min_{\norm{v}_{w^{k-1}}= 1} \norm{S^{\frac{1}{2}}(c)\widetilde{V}v}_{w^{k-1}}\right)^2 \\
&= \left(\min_{\norm{\widetilde{V}v}_{w^{k-1}}= 1} \norm{S^{\frac{1}{2}}(c)\widetilde{V}v}_{w^{k-1}}\right)^2 \\
&= \left(\min_{\norm{\theta}_{w^{k-1}}= 1, \theta\in Im(L^{k-1}_{\uparrow})} \norm{S^{\frac{1}{2}}(c)\theta}_{w^{k-1}}\right)^2\\
&\geq \left(\min_{\norm{\theta}_{w^{k-1}}= 1} \norm{S^{\frac{1}{2}}(c)\theta}_{w^{k-1}}\right)^2 = \left(\norm{S^{-\frac{1}{2}}(c)}^2_{w^{k-1}}\right)^{-1},
\end{align*}
where the last equality comes from Eq.~\eqref{eq::proof_1}.
We have proven that
$$
\frac{1}{\norm{(A^*A)^{-1}}_{w^{k-1}}} \geq \frac{1}{\norm{S^{-\frac{1}{2}}(c)}^2_{w^{k-1}}},
$$
or, equivalently,
\begin{equation}\label{eq::proof_2}
\norm{(A^*A)^{-1}}_{w^{k-1}} \leq \norm{S^{-\frac{1}{2}}(c)}^2_{w^{k-1}}.
\end{equation}
This term can be further rewritten as 
\begin{align*}
\norm{S^{-\frac{1}{2}}(c)}^2_{w^{k-1}} = \norm{(W_{k-1})^{-\frac{1}{2}}S^{-\frac{1}{2}}(c)(W_{k-1})^{\frac{1}{2}}}^2_2 = \norm{S^{-\frac{1}{2}}(c)}^2_2 = \max_i \sinc^{-1} (\widetilde{V}c)_i = \norm{\sinc^{-1} (\widetilde{V}c)}_\infty\, , 
\end{align*}
because $S(c)$ is diagonal with positive diagonal elements.
We now remove the dependency on $c$ by taking a maximum over $\mathcal{B}$
\begin{equation}
\norm{S^{-\frac{1}{2}}(c)}^2_{w^{k-1}} \leq \max_{c\in\mathcal{B}}\norm{\sinc^{-1} (\widetilde{V}c)}_\infty
\leq \max_{x\in [-\gamma,\gamma]} \sinc^{-1}(x) = \sinc^{-1}(\gamma) = \frac{\gamma}{\sin(\gamma)}.
\end{equation}
Thus we have that
\begin{align*}
\norm{(\widetilde{V}^*S(c)\widetilde{V})^{-1}}_{w^{k-1}}\leq \norm{S^{-\frac{1}{2}}(c)}^2_{w^{k-1}}\leq\sinc^{-1}(\gamma)\, .
\end{align*}
\end{proof}
Applying lemmas \ref{lemma:right} and \ref{lemma:left} to \cref{eq::f(c)_bound}, we have
\begin{align}
\norm{f(c)}_{w^{k-1}} = \norm{(\widetilde{V}^* S(c) \widetilde{V})^{-1}\widetilde{\Lambda}^{-1}\widetilde{V}^*\frac{\omega^{(-)}}{\sigma^\downarrow}}_{w^{k-1}} \leq \frac{1}{\sigma^\downarrow}\frac{\gamma}{\sin(\gamma)}\norm{\beta^{(-)}}_{w^{k-1}}\, ,
\end{align}
which means that $f(c)\in\mathcal{B}$ if and only if 
\begin{equation}  
\norm{f(c)}_{w^{k-1}} \leq \frac{\gamma}{\sqrt{\max_i (w^{k-1}_i)}}
\end{equation}
which holds if
\begin{equation}\label{eq:get_bound}
\frac{1}{\sigma^\downarrow}\frac{\gamma}{\sin(\gamma)}\norm{\beta^{(-)}}_{w^{k-1}} \leq \frac{\gamma}{\sqrt{\max_i(w^{k-1}_i)}} \iff \sigma^\downarrow \geq \frac{\sqrt{\max_i(w^{k- 1}_i)}}{\sin(\gamma)}\norm{\beta^{(-)}}_{w^{k-1}}.
\end{equation}

This proves that, under the condition of the theorem, $f(c)$ maps the closed ball $\mathcal{B}$ to itself and so Brouwer's theorem ensures the existence of a fixed point (i.e. a reachable equilibrium) $\theta^{(-)}_{eq} = \widetilde{V}c_{eq}$ with $c_{eq}\in \mathcal{B}$. 

The asymptotic stability of $\theta^{(-)}_{eq}$ can be seen by computing the Jacobian of the reachable dynamics \cref{eq:reachable_dynamics}
\begin{align}
J^{(-)}(c_{eq}) = -\sigma^\downarrow \widetilde{\Lambda}\widetilde{V}^* \diagm(\cos(\theta_{eq}^{(-)}))\widetilde{V}\, ,
\end{align}
which has the same nonzero eigenvalues as 
\begin{align*}
\widetilde{J}^{(-)}(c_{eq}) = -\sigma^\downarrow \widetilde{\Lambda}^{\frac{1}{2}}\widetilde{V}^* \diagm(\cos(\theta^{(-)}_{eq}))\widetilde{V} \widetilde{\Lambda}^{\frac{1}{2}}.
\end{align*}
If $c_{eq}\in\mathcal{B}$ and $\gamma\in(0,\pi/2)$, then
\begin{align*}
\norm{\theta_{eq}^{(-)}}_{w^{k-1}} \leq \frac{\gamma}{\sqrt{\max_i(w^{k- 1}_i)}} \implies \norm{\theta_{eq}^{(-)}}_\infty \leq \gamma \implies \cos\left(\theta_{eq}^{(-)}\right) > 0\, , 
\end{align*}
when $\gamma\in (0,\pi/2)$,
and thus
\begin{align*}
\widetilde{J}^{(-)}(c_{eq}) = -\sigma^\downarrow \left (\diagm\left (\cos\left(\theta_{eq}^{(-)}\right)\right)^\frac12\widetilde{V}\widetilde{\Lambda}^\frac12\right )^*\left (\diagm\left(\cos\left (\theta_{eq}^{(-)}\right )\right)^\frac12\widetilde{V}\widetilde{\Lambda}^\frac12\right) = -\sigma^\downarrow A^*A\, , 
\end{align*}
which is trivially negative definite as $A = \diagm\left (\cos(\theta^{(-)}_{eq})\right )^{\frac12}\widetilde{V}\widetilde{\Lambda}^{\frac12}$ has trivial kernel and $\sigma^\downarrow>0$.

\section{Proof of Theorem \ref{theorem:sigma_fp_explosive}}\label{appendix:proof_explosive}
The proof is a direct extension of the proof of \cref{theorem:sigma_fp} written in \Cref{appendix:proof}. 
Let us first write \cref{eq:adaptively_coupled} as
\begin{align}\label{eq:explosive_projection_dynamics} 
\begin{split}
\Dot{\theta}^{(+)} &= \omega^{(+)} - \sigma^\uparrow R_k^-(\theta^{(-)}) L^{k+1}_\downarrow\sin(\theta^{(+)}) \\
\Dot{\theta}^{(-)} &= \omega^{(-)} - \sigma^\downarrow R_k^+(\theta^{(+)})L^{k-1}_\uparrow \sin(\theta^{(-)})\, ,
\end{split}
\end{align}
for which we can write equilibrium conditions as fixed point equations \cref{eq::fixed_point_equation}:
\begin{align}
\begin{split}
c &= \left (\widetilde{V}^* S(c)\widetilde{V}\right)^{-1} \widetilde{\Lambda}^{-1} \widetilde{V}^*\frac{\omega^{(+)}}{\sigma^\uparrow R^-_k(V'c')} \defeq f^{(+)}(c,c')\\
c' &= \left(\widetilde{V'}^* S'(c')\widetilde{V'}\right)^{-1} \widetilde{\Lambda'}^{-1} \widetilde{V'}^*\frac{\omega^{(-)}}{\sigma^\downarrow R^+_k(Vc)} \defeq f^{(-)}(c,c')\, .
\end{split}
\end{align}
Here $c,c'$ are respectively the coefficients of $\theta^{(+)},\theta^{(-)}$ w.r.t the orthonormal bases $\widetilde{V},\widetilde{V'}$ of $\Ima L^{k+1}_\downarrow, L^{k-1}_\uparrow$. 
The configurations $\theta^{(+)} = Vc,\theta^{(-)}=V'c'$ will be reachable equilibria for the dynamics of \cref{eq:explosive_projection_dynamics} if and only if
\begin{align}
\mathbf{c} = \mathbf{f}(\mathbf{c})\, ,
\end{align}
where
\begin{align}
\mathbf{f}(\mathbf{c}) \defeq \begin{pmatrix}
    f^{(+)}(c,c')\\
    f^{(-)}(c,c')
\end{pmatrix}\qquad \mathrm{and} \qquad \mathbf{c} \defeq \begin{pmatrix}
    c\\
    c'
\end{pmatrix}\, .
\end{align}
Again we want to prove that $\mathbf{f}$ is a continuous function which maps the convex set $\mathcal{B}\times\mathcal{B}'$ to itself, with
\begin{align}
\mathcal{B} = \sset{c:\norm{\widetilde{V}c}_{w^{k+1}}\leq\frac{\gamma^{(+)}}{\sqrt{\max_i(w^{k+ 1}_i)}}},\ \ \mathcal{B}' = \sset{c':\norm{\widetilde{V}'c'}_{w^{k-1}}\leq\frac{\gamma^{(-)}}{\sqrt{\max_i(w^{k- 1}_i)}}}\, .
\end{align}
To prove this, we just need to show that
\begin{align*}
\mathbf{c}\in\mathcal{B}\times\mathcal{B}' \implies f^{(+)}\in\mathcal{B},\  f^{(-)}\in\mathcal{B}'\, .
\end{align*}
Repeating the same steps performed in \cref{appendix:proof} we get
\begin{align}
\begin{split}
    \norm{f^{(+)}(c,c')}_{w^{k+1}} &\leq \frac{1}{\sigma^\uparrow}\frac{C^{(-)}_k}{\abs{\ones^\top W^{-1}_{k-1}\cos(\widetilde{V}'c')}}\frac{\gamma^{(+)}}{\sin(\gamma^{(+)})}\norm{\beta^{(+)}}_{w^{k+1}}\\
    \norm{f^{(-)}(c,c')}_{w^{k-1}} &\leq \frac{1}{\sigma^\downarrow}\frac{C^{(+)}_k}{\abs{\ones^\top W^{-1}_{k+1}\cos(\widetilde{V}c)}}\frac{\gamma^{(-)}}{\sin(\gamma^{(-)})}\norm{\beta^{(-)}}_{w^{k-1}}.
\end{split}
\end{align}
The terms at the denominator can be further bounded with a term that does not depend on $c,c'$. In fact,
\begin{align*}
\frac{1}{C^{(-)}_k}\abs{\ones^\top W^{-1}_{k-1}\cos(\widetilde{V}'c')} = \frac{1}{C^{(-)}_k}\abs{\sum_i \frac{1}{w^{k-1}_{i}}\cos(\widetilde{V}'c')_i} \geq \frac{1}{C^{(-)}_k}\abs{\sum_i \frac{1}{w^{k-1}_i}\cos(\gamma^{(-)})} = \cos(\gamma^{(-)})
\end{align*}
for $\gamma\in (0,\frac{\pi}{2})$, as $c'\in\mathcal{B}' \implies \norm{c'}_\infty \leq \gamma^{(-)} \implies \gamma^{(-)}\geq \cos(\gamma^{(-)})$. 
With the same bound for the other term we arrive at
\begin{align}
\begin{split}
    \norm{f^{(+)}(c,c')}_{w^{k+1}} &\leq  \frac{1}{\sigma^\uparrow}\frac{\gamma^{(+)}}{\sin(\gamma^{(+)})\cos(\gamma^{(-)})}\norm{\beta^{(+)}}_{w^{k+1}}\\
    \norm{f^{(-)}(c,c')}_{w^{k-1}} &\leq \frac{1}{\sigma^\downarrow}\frac{\gamma^{(-)}}{\sin(\gamma^{(-)})\cos(\gamma^{(+)})}\norm{\beta^{(-)}}_{w^{k-1}}\, ,
\end{split}
\end{align}
from which the thesis easily follows by repeating the steps in \cref{eq:get_bound}.

\section{Proof of stability in Theorem \ref{theorem:make_eq_reachable}}\label{appendix:make_eq_reachable}

Let us prove the stability part of \Cref{theorem:make_eq_reachable} for the projection onto higher dimensional simplices $\theta^{(+)}$.

Following the proof in \cref{appendix:proof}, we write the dynamics of the coefficients $c$ of $\theta^{(+)}$ in the orthonormal basis of the reachable subspace given by the matrix $\widetilde{V}$. We have
\begin{align}  
\Dot{c} = \widetilde{V}^*\omega^{(+)} - \sigma^\uparrow\widetilde{\Lambda}\widetilde{V}^*\sin(\widetilde{V}c + \alpha_{k+1})\, ,
\end{align}
whose Jacobian matrix is given as 
\begin{align}  
\widetilde{J}^{(+)}(c) = -\sigma^\uparrow \widetilde{\Lambda} \widetilde{V}^* \mathrm{diag}\left (\cos\left(\widetilde{V}c + \alpha_{k+1}\right)\right)\widetilde{V}\, . 
\end{align}

We then evaluate the Jacobian the equilibrium solution $\theta^{(+)}_* = \widetilde{V}c_*$ and replace the value of $\alpha_{k+1}$ prescribed by the theorem, that is $\alpha_{k+1} = \arcsin\left(\frac{\beta^{(+)}}{\sigma^\uparrow}\right) - \widetilde{V}c_*$, resulting in
\begin{align}  
J^{(+)}(c_*) &= -\sigma^\uparrow \widetilde{\Lambda} \widetilde{V}^* \mathrm{diag}\cos\left(\arcsin\left(\frac{\beta^{(+)}}{\sigma^\uparrow}\right)\right)\widetilde{V} = -\sigma^\uparrow \widetilde{\Lambda} \widetilde{V}^* \mathrm{diag}\sqrt{\ones - \left(\frac{\beta^{(+)}}{\sigma^\uparrow}\right)^2}\widetilde{V}\, .
\end{align}
This matrix has the same eigenvalues as
\begin{align}   
\widetilde{J}^{(+)} = -\sigma^\uparrow \widetilde{\Lambda}^\frac{1}{2} \widetilde{V}^* \mathrm{diag}\sqrt{\ones - \left(\frac{\beta^{(+)}}{\sigma^\uparrow}\right)^2}\widetilde{V}\widetilde{\Lambda}^\frac{1}{2}\, ,
\end{align}
which is Hermitian and negative definite, as $\widetilde{V}c = 0 \iff c=0$ because the columns of $\widetilde{V}$ are a basis of $\Ima D^k$, and $\sqrt{1-(\beta^{(+)}_i/\sigma^\uparrow
)^2} > 0$ because, by hypothesis, $\sigma^\uparrow > \norm{\beta^{(+)}}_\infty$.